\documentclass[11pt]{article}

\usepackage{fullpage}
\usepackage{bbm}
\usepackage{amsmath}
\usepackage{amsthm}
\usepackage{amssymb}
\usepackage{upref}
\usepackage{amsfonts}
\usepackage{fancybox}
\usepackage{epic}
\usepackage{eepic}
\usepackage{graphicx}
\usepackage{graphics}
\usepackage{subcaption}
\usepackage{float}
\usepackage{wrapfig}
\usepackage{times}
\usepackage{nicefrac}
\usepackage{xspace}
\usepackage{color}
\usepackage[ruled,vlined,linesnumbered]{algorithm2e}
\usepackage{romannum}
\usepackage{datetime}
\usepackage[numbers,sort]{natbib}
\usepackage{upgreek}
\usepackage{enumerate}
\usepackage[shortlabels]{enumitem}
\usepackage[margin=2.5cm]{geometry}
\usepackage{microtype}
\usepackage{url}
\usepackage{hyperref}

\newtheorem{claim}{Claim}
\newtheorem{defn}{Definition}
\newtheorem{lemma}{Lemma}
\newtheorem{theorem}{Theorem}

\newtheorem{corollary}{Corollary}

\newtheorem{prop}{Proposition}	

\DeclareMathOperator*{\argmax}{argmax}

\newcommand{\PCSMExtended}{\textsc{Packing-Covering Submodular Maximization}\xspace}
\newcommand{\MPCSM}{\textsc{(MatroidPCSM)}\xspace}
\newcommand{\MPCSMExtended}{\textsc{Matroid Packing-Covering Submodular Maximization}\xspace}
\newcommand{\PCMSM}{\textsc{(MultiPCSM)}\xspace}
\newcommand{\PCSM}{\textsc{(PCSM)}\xspace}
\newcommand{\PCMSMExtended}{\textsc{Packing-Covering Multiple Submodular Maximization}\xspace}

\newcommand{\zero}{\mathbf{0}}
\newcommand{\one}{\mathbf{1}}
\newcommand{\x}{\mathbf{x}}

\newcommand{\bc}{\mathbf{c}}
\newcommand{\p}{\mathbf{p}}

\newcommand{\E}{\mathbb{E}}
\newcommand{\bP}{\mathbf{P}}
\newcommand{\tP}{\widetilde{\bP}}
\newcommand{\e}{\mathrm{e}}
\newcommand{\eps}{\varepsilon}
\newcommand{\opt}{f(O)}

\newcommand{\cN}{\mathcal{N}}
\newcommand{\cM}{\mathcal{M}}
\newcommand{\cI}{\mathcal{I}}
\newcommand{\cP}{\mathcal{P}}
\newcommand{\cC}{\mathcal{C}}

\newcommand{\cL}{\mathcal{L}}
\newcommand{\R}{\mathbb{R}}
\newcommand{\N}{\mathbb{N}}
\newcommand{\bC}{\mathbf{C}}
\newcommand{\bV}{\mathbf{V}}
\newcommand{\tC}{\widetilde{\bC}}
\newcommand{\bx}{\mathbf{x}}
\newcommand{\by}{\mathbf{y}}
\newcommand{\ie}{{\em i.e.}}
\newcommand{\br}{\mathbf{r}}
\newcommand{\bs}{\mathbf{s}}
\newcommand{\bw}{\mathbf{w}}
\newcommand{\residual}{\cN \setminus (E_0 \cup E_1)}
\newcommand{\sol}{E_1 \cup R_D}
\newcommand{\rd}{\br_D}
\newcommand{\rdi}{(\rd)_i}
\newcommand{\sd}{\bs_D}
\newcommand{\sdj}{(\sd)_j}
\newcommand{\tildx}{\widetilde{\bx}^*}
\newcommand{\ind}{q}

\title{A Tight Approximation for Submodular Maximization with Mixed Packing and Covering Constraints}

\author{Eyal Mizrachi \thanks{Computer Science Department, Technion, Haifa 32000, Israel. \tt{eyalmiz@cs.technion.ac.il}} \and
	    Roy Schwartz \thanks{Computer Science Department, Technion, Haifa 32000, Israel. \tt{schwartz@cs.technion.ac.il}} \and
	    Joachim Spoerhase \thanks{Department of Computer Science, Aalto University, Espoo, Finland \& University of W{\"u}rzburg, Germany. \tt{joachim.spoerhase@aalto.fi}} \and
	    Sumedha Uniyal \thanks{Department of Computer Science, Aalto University, Espoo, Finland. \tt{sumedha.uniyal@aalto.fi}}}
\date{}

\begin{document}
\pagenumbering{arabic}
\maketitle

\begin{abstract}

Motivated by applications in machine learning, such as subset selection and data summarization, we consider the problem of maximizing a monotone submodular function subject to mixed packing and covering constraints.
We present a tight approximation algorithm that for any constant $\eps >0$ achieves a guarantee of $1-\nicefrac[]{1}{\e}-\eps$ while violating only the covering constraints by a multiplicative factor of $1-\eps$.
Our algorithm is based on a novel enumeration method, which unlike previous known enumeration techniques, can handle both packing and covering constraints.
We extend the above main result by additionally handling a matroid independence constraints as well as finding (approximate) pareto set optimal solutions when multiple submodular objectives are present. Finally, we propose a novel and purely combinatorial dynamic programming approach that can be applied to several special cases of the problem yielding not only {\em deterministic} but also considerably faster algorithms.
For example, for the well studied special case of only packing constraints (Kulik {\em et. al.} [Math. Oper. Res. `13] and Chekuri {\em et. al.} [FOCS `10]), we are able to present the first deterministic non-trivial approximation algorithm.
We believe our new combinatorial approach might be of independent interest.

\end{abstract}


\section{Introduction}\label{sec:Introduction}
The study of combinatorial optimization problems with a submodular objective has attracted much attention in the last decade.
A set function $f:2^{\cN}\rightarrow \R_+$ over a ground set $\cN$ is called {\em submodular} if it has the {\em diminishing returns} property:
$f(A\cup \left\{ i\right\}) - f(A) \geq f(B\cup \left\{ i\right\}) - f(B)$ for every $A\subseteq B\subseteq \cN$ and $i\in \cN\setminus B$.\footnote{An equivalent definition is: $ f(A)+f(B)\geq f(A\cup B)+f(A\cap B)$ for every $ A,B\in \cN$.}
Submodular functions capture the principle of economy of scale, prevalent in both theory and real world applications.
Thus, it is no surprise that combinatorial optimization problems with a submodular objective arise in numerous disciplines, {\em e.g.}, machine learning and data mining \cite{B13,BHK14}, algorithmic game theory and social networks \cite{DRS12,HMS08,HK14,KKT15,SU13}, and economics \cite{SA11}. 
Additionally, many classical problems in combinatorial optimization are in fact submodular in nature, {\em e.g.}, maximum cut and maximum directed cut \cite{GW95,HZ01,H01,K72,KKMO07}, maximum coverage \cite{F98,KMN99}, generalized assignment problem \cite{CK05,CKR06,FV06,FGMS06}, maximum bisection \cite{ABG16,FJ95}, and facility location \cite{AS99,CFN77b,CFN77a}.

In this paper we consider the problem of maximizing a monotone\footnote{$f$ is monotone if $f(S)\leq f(T)$ for every $ S\subseteq T\subseteq \cN$.} submodular function given mixed packing and covering constraints.
In addition to being a natural problem in its own right, it has further real world applications.

As a motivating example consider the subset selection task in machine learning \cite{G14,GKT12,KT11} (also refer to Kulesza and Taskar \cite{KT12} for a thorough survey).
In the subset selection task the goal is to select a diverse subset of elements from a given collection.
One of the prototypical applications of this task is the document summarization problem \cite{KT11,HL10,HL11}: given textual units the objective is to construct a short summary by selecting a subset of the textual units that is both representative and diverse.
The former requirement, representativeness, is commonly achieved by maximizing a submodular objective function, {\em e.g.}, graph based \cite{HL10,HL11} or log subdeterminant \cite{KT11}.
The latter requirement, diversity, is typically tackled by penalizing the submodular objective for choosing similar textual units (this is the case for both of the above two mentioned submodular objectives).
However, such an approach results in a submodular objective which is not necessarily non-negative thus making it extremely hard to cope with.
As opposed to penalizing the objective, a remarkably simple and natural approach to tackle the diversity requirement is by the introduction of covering constraints.
For example, one can require that for each topic that needs to appear in the summary, a sufficient number of textual units that refer to it are chosen.
Unfortunately, to the best of our knowledge there is no previous work in the area of submodular maximization that incorporates general covering constraints.\footnote{There are works on exact cardinality constraints for non-monotone submodular functions, which implies a special, uniform covering constraint~\cite{bfns14,LMNS09,v09}.}

Let us now formally define the main problem considered in this paper.
We are given a monotone submodular function $f:2^{\cN}\rightarrow \R_+$ over a ground set $\cN=\left\{ 1,2,\ldots,n\right\}$. 
Additionally, there are $p$ packing constraints given by $\bP\in \R_+^{p\times n}$, and $c$ covering constraints given by $\bC\in \R_+^{c\times n}$ (all entries of $\bP$ and $\bC$ are non-negative).
Our goal is to find a subset $S\subseteq \cN$ that satisfies all packing and covering constraints that maximizes the value of $f$:
\begin{align}
\max \left\{ f(S):S\subseteq \cN, \bP \mathbf{1} _{S}\leq \mathbf{1}_p, \bC \mathbf{1} _{S}\geq \mathbf{1}_c\right\}. \label{Objective}
\end{align}
In the above $\mathbf{1}_S\in \R^n$ is the indicator vector for $S\subseteq \cN$ and $\one_k\in \R^k$ is a vector of dimension~$k$ whose coordinates are all $1$.
We denote this problem as \PCSMExtended \PCSM. It is assumed we are given a feasible instance, \ie, there exists $S\subseteq \cN$ such that $\bP \mathbf{1} _{S}\leq \mathbf{1}_p$ and $\bC \mathbf{1} _{S}\geq \mathbf{1}_c $.

As previously mentioned, \PCSM captures several well known problems as a special case when only a single packing constraint is present ($p=1$ and $c=0$): maximum coverage \cite{KMN99}, and maximization of a monotone submodular function given a knapsack constraint \cite{s04,W82} or a cardinality constraint \cite{NWF78}.
For all of these special cases an approximation of $\left( 1-\nicefrac[]{1}{\e}\right)$ is achievable and known to be tight \cite{NW78} (even for the special case of a coverage function \cite{F98}).
When a constant number of knapsack constraints is given ($p=O(1)$ and $c=0$) Kulik {\em et al.} \cite{KST13} presented a tight $\left( 1-\nicefrac[]{1}{\e}-\eps \right)$-approximation for any constant $\eps >0$.
An alternative algorithm with the same guarantee was given by Chekuri {\em et al.} \cite{chekuriVZ10-rand-exch}.

\vspace{5pt}
\noindent {\bf{{Our Results:}}}
We present a tight approximation guarantee for \PCSM when the number of constraints is constant.
Recall that we assume we are given a feasible instance, {\em i.e.}, there exists $S\subseteq \cN$ such that $\bP \mathbf{1} _{S}\leq \mathbf{1}_p$ and $\bC \mathbf{1} _{S}\geq \mathbf{1}_c $.
The following theorem summarizes our main result.
From this point onwards we denote by $O$ some fixed optimal solution to the problem at hand.
\begin{theorem}\label{thrm:MainApproximation}
	For every constant $\eps >0$, assuming $p$ and $c$ are constants, there exists a randomized polynomial time algorithm for \PCSM running in time $n^{\text{poly}(1/\varepsilon)}$ that outputs a solution $S\subseteq \cN$ that satisfies: $(1)$ $f(S) \geq \left(1-\nicefrac[]{1}{\e}-\eps\right) f(O)$; and $(2)$ $\bP \mathbf{1}_S\leq \mathbf{1}_p$ and $\bC \mathbf{1}_S\geq (1-\eps) \mathbf{1}_c $.
\end{theorem}

\noindent We note four important remarks regarding the tightness of Theorem \ref{thrm:MainApproximation}:
\begin{enumerate}
\item The loss of $1-\nicefrac[]{1}{\e}$ in the approximation cannot be avoided, implying that our approximation guarantee is (virtually) tight.
The reason is that no approximation better than $1-\nicefrac[]{1}{\e}$ can be achieved even for the case where only a single packing constraint is present \cite{NW78}.
\item The assumption that the number of constraints is constant is unavoidable.
The reason is that if the number of constraints is not assumed to be constant, then even with a linear objective \PCSM captures the maximum independent set problem.
Hence, no approximation better than $n^{-(1-\eps)},$ for any constant $\eps > 0$, is possible \cite{H96}.\footnote{If the number of packing constraints $p$ is super-constant then approximations are known only for special cases with ``loose'' packing constraints, {\em i.e.}, $P_{i,\ell}\leq O(\varepsilon ^2/\ln{p})$ (see, {\em e.g.}, \cite{chekuriVZ10-rand-exch}).}  
\item No {\em true} approximation with a finite approximation guarantee is possible, \ie, finding a solution $S\subseteq \cN$ such that $\bP \mathbf{1} _{S}\leq \mathbf{1}_p$ and $\bC \mathbf{1} _{S}\geq \mathbf{1}_c $ with no violation of the constraints. 
The reason is that one can easily encode the subset sum problem using a single packing and a single covering constraint.
Thus, just deciding whether a feasible solution exists, regardless of its cost, is already NP-hard.
\item Guaranteeing one-sided feasibility, {\em i.e.}, finding a solution which does not violate the packing constraints and a violates the covering constraint only by a factor of $\varepsilon$, cannot be achieved in time $n^{o(1/\varepsilon)}$ unless the exponential time hypothesis fails (see Appendix~\ref{apx-sec:lower-bound-reduction}). 
\end{enumerate}
Therefore, we can conclude that our main result (Theorem \ref{thrm:MainApproximation}) provides the best possible guarantee for the \PCSM problem.
We also note that all previous work on the special case of only packing constraints \cite{chekuriVZ10-rand-exch,KST13} have the same running time of $n^{\text{poly}(1/\varepsilon)}$.

We present additional extensions of the above main result.
The first extension deals with \PCSM where we are also required that the output is an independent set in a given matroid $\cM = (\cN,\cI)$.
We denote this problem by \MPCSMExtended \MPCSM, and it is defined as follows:
$\max \left\{ f(S):S\subseteq \cN, \bP \mathbf{1} _{S}\leq \mathbf{1}_p, \bC \mathbf{1} _{S}\geq \mathbf{1}_c,S\in \cI\right\}$ . 
As in \PCSM, we assume we are given a feasible instance, {\em i.e.}, there exists $S\subseteq \cN$ such that $ \bP \mathbf{1} _{S}\leq \mathbf{1}_p$, $ \bC \mathbf{1} _{S}\geq \mathbf{1}_c$, and $S\in \cI $.
Our result is summarized in the following theorem.
\begin{theorem}\label{thrm:MatroidExtend}
For every constant $\eps > 0$, assuming $p$ and $c$ are constants, there exists a randomized polynomial time algorithm for \MPCSM that outputs a solution $S\in \cI$ that satisfies: $(1)$ $ f(S)\geq \left( 1-\nicefrac[]{1}{\e}-\eps\right) f(O)$; and $(2)$ $\bP \mathbf{1}_S\leq \mathbf{1}_p$ and $\bC \mathbf{1}_S\geq (1-\eps) \mathbf{1}_c$.
\end{theorem}

The second extension deals with the multi-objective variant of \PCSM where we wish to optimize over several monotone submodular objectives.
We denote this problem by \PCMSMExtended \PCMSM.
Its input is identical to that of \PCSM except that instead of a single objective $f$ we are given $t$ monotone submodular functions $f_1,\ldots,f_t:2^{\cN}\rightarrow \mathbb{R}_+$.
As before, we assume we are given a feasible instance, {\em i.e.}, there exists $S\subseteq \cN$ such that $ \bP \mathbf{1} _{S}\leq \mathbf{1}_p$ and $ \bC \mathbf{1} _{S}\geq \mathbf{1}_c$.
Our goal is to find pareto set solutions considering the $t$ objectives.
To this end we prove the following theorem.
\begin{theorem}\label{thrm:MultiExtend}
For every constant $\eps > 0$, assuming $p$, $c$ and $t$ are constants, there exists a randomized polynomial time algorithm for \PCMSM that for every target values $v_1,\ldots,v_t$ either: $(1)$ finds a solution $S\subseteq \cN$ where $ \bP \mathbf{1} _{S}\leq \mathbf{1}_p$ and $ \bC \mathbf{1} _{S}\geq (1-\eps)\mathbf{1}_c$ such that for every $1\leq i\leq t$: $f_i(S)\geq \left( 1-\nicefrac[]{1}{\e}-\eps\right) v_i$; or $(2)$ returns a certificate that there is no solution $S\subseteq \cN$, where $ \bP \mathbf{1} _{S}\leq \mathbf{1}_p$ and $ \bC \mathbf{1} _{S}\geq \mathbf{1}_c$ such that for every $1\leq i\leq t$: $f_i(S)\geq v_i$.
\end{theorem}
We also note that Theorems \ref{thrm:MatroidExtend} and \ref{thrm:MultiExtend} can be combined such that we can handle \PCMSM where a matroid independence constraint is present, in addition to the given packing and covering constraints, achieving the same guarantees as in Theorem \ref{thrm:MultiExtend}. 

All our previously mentioned results employ a continuous approach and are based on the multilinear relaxation, and thus are inherently {\em randomized}.\footnote{
Known techniques to efficiently evaluate the multilinear extension are randomized, {\em e.g.}, \cite{CCPV11}.}
We present a new combinatorial greedy-based dynamic programming approach for submodular maximization that enables us, for several well studied special cases of \PCSM, to obtain deterministic and considerably faster algorithms.
Perhaps the most notable result is the first deterministic non-trivial algorithm for maximizing a monotone submodular function subject to a constant number of packing constraints (previous works \cite{chekuriVZ10-rand-exch,KST13} are randomized).
\begin{theorem}
\label{thm:submodular-order1-pack}
For every constants $\eps >0$ and $p\in \mathbb{N}$, there exists a deterministic algorithm for maximizing a monotone submodular function subject to $p$ packing constraints, that runs in time $O(n^{\textnormal{poly}(1/\eps)})$ and achieves an approximation of $\nicefrac[]{1}{e}-\varepsilon$. 
\end{theorem}

The interesting special case of \PCSM is when a single packing and a single covering constraints are present ($p=c=1$) is summarized in the following theorem.
\begin{theorem}
\label{thm:submodular-one}
For every constant $\eps >0$ and $p=c=1$, there exists a deterministic algorithm for \PCSM running in time $O(n^{1/\eps})$ that outputs a solution $S\subseteq \cN$ that satisfies: $(1)$ $f(S) \geq 0.353 f(O)$; and $(2)$ $\bP \mathbf{1}_S\leq (1+\eps)\mathbf{1}_p$ and $\bC \mathbf{1}_S\geq (1-\eps)\mathbf{1}_c$. For the case when the packing constraint is a cardinality constraint, {\em i.e.}, $ \bP=\mathbf{1}_n^{\intercal}/k$, we can further guarantee that $\bP \mathbf{1}_S\leq \mathbf{1}_p$ and a running time of $O(\nicefrac{n^4}{\eps})$.

\end{theorem}
\vspace{5pt}
\noindent {\bf{{Our Techniques:}}}
Our main result is based on a continuous approach: first a continuous relaxation is formulated, second it is (approximately) solved, and finally the fractional solution is rounded into an integral solution.
Similarly to the previous works of \cite{chekuriVZ10-rand-exch,KST13}, which focus on the special case of only packing constraints, the heart of the algorithm lies in an enumeration preprocessing phase that chooses and discards some of the elements prior to formulating the relaxation.
The enumeration preprocessing step of \cite{chekuriVZ10-rand-exch,KST13} is remarkably simple and elegant.
It enumerates over all possible collections of large elements the optimal solution chooses, {\em i.e.}, elements whose size exceeds some fixed constant in at least one of the packing constraints and are chosen by the optimal solution.\footnote{An additional part of the preprocessing involves enumerating over collections of elements whose marginal value is large with respect to the objective $f$, however this part of the enumeration is not affected by the presence of covering constraints and thus is ignored in the current discussion.}
All remaining large elements not in the guessed collection are discarded.
This enumeration terminates in polynomial time and ensures that no large elements are left in any of the packing constraints.
Thus, once no large elements remain concentration bounds can be applied.
For the correct guess, any of the several known randomized rounding techniques can be employed (alongside a simple rescaling) to obtain an approximation of $1-\nicefrac[]{1}{e}-\eps$ (here $\eps>0$ is a constant that is used to determine which elements are considered large).
Unfortunately, this approach fails in the presence of covering constraints since an optimal solution can choose {\em many} large elements in any given covering constraint.
One can naturally adapt the above known preprocessing by enumerating over all possible collections of covering constraints that the optimal solution $O$ covers using only large elements.
However, this leads to an approximation of $1-\nicefrac[]{1}{e}-\varepsilon$ while {\em both} packing and covering constraints are violated by a multiplicative factor of $1 \pm \varepsilon$.
We aim to obtain one sided violation of the constraints, {\em i.e.}, only the covering constraints are violated by a factor of $1 - \varepsilon$ whereas the packing constraints are fully satisfied.


Avoiding constraint violation is possible in the presence of pure packing constraints~\cite{chekuriVZ10-rand-exch,KST13}. Known approaches for the latter are crucially based on \emph{removing} elements in a pre-processing and post-processing step in order to guarantee that concentration bounds hold. For mixed constraints, these known removal operations may, however, arbitrarily violate the covering constraints. Our approach aims at pre-processing the input instance via partial enumeration so as to avoid discarding elements by ensuring that the remaining elements are ``locally'' small relatively to the \emph{residual} constraints. If this property would hold scaling down the solution by a factor $1/(1+\eps)$ would be sufficient to avoid violation of the packing constraints. Unfortunately, we cannot guarantee this to hold for all constraints. Rather, for some {\em critical} constraints locally large elements may still be present. We introduce a novel enumeration process that detects these critical constraints, {\em i.e.}, constraints that are prone to violation.
Such constraints are given special attention as the randomized rounding might cause them to significantly deviate from the target value.
Unlike the previously known preprocessing method, our enumeration process handles covering constraints with much care and it takes into account the {\em actual} coverage of the optimal solution $O$ of each of the covering constraints.
Combining the above, alongside a postprocessing phase that discards large elements from critical packing constraints, suffices to yield the desired result.


We also independently present a novel purely combinatorial greedy-based dynamic programming approach that in some cases yields deterministic and considerably faster algorithms.
In our approach we maintain a table that contains greedy approximate solutions for all possible packing and covering values. Using this table we extend the simple greedy process by populating each table entry with the most profitable extension of the previous table entries. In this way we are able to simulate (in a certain sense) \emph{all} possible sequences of packing and covering values for the greedy algorithm, ultimately leading to a good feasible solution. 
To estimate the approximation factor we employ a factor-revealing linear program.
To the best of our knowledge, this is the first time a dynamic programming based approach is used for submodular optimization. 
We believe our new combinatorial dynamic programming approach is of independent interest.

\section{Preliminaries}\label{sec:Preliminaries}
In this paper we assume the standard {\em value oracle} model, where the algorithm can only access the given submodular function $f$ with queries of the form: what is $f(S)$ for a given $S$?
The running time of the algorithm is measured by the total number of value oracle queries and arithmetic operations it performs.
Additionally, let us define $f_A(S) \triangleq f(A \cup S) - f(A)$ for any subsets $A, S \subseteq \cN$.
Furthermore, let $f_A(\ell) \triangleq f_A(\{\ell\})$.

The multiliear extension $F:[0,1]^{\cN}\rightarrow \mathbb{R}_+$ of a given set function $f:2^{\cN}\rightarrow \mathbb{R}_+$ is: 
\begin{align*}
F(\bx) \triangleq \sum _{R\subseteq \cN} f(R)\prod _{\ell \in R}x_{\ell} \prod _{\ell \notin R} \left(1-x_{\ell}\right)~~~~~~~\forall \bx \in [0,1]^{\cN}.
\end{align*}
Additionally, we make use of the following theorem that provides the guarantees of the continuous greedy algorithm of \cite{CCPV11}.\footnote{We note that the actual guarantee of the continuous greedy algorithm is $(1-\nicefrac{1}{e}-o(1))$. However, for simplicity of presentation, we can ignore the $o(1)$ term due to the existence of a loss of $\varepsilon$ (for any constant $\varepsilon$) in all of our theorems.}

\begin{theorem}[Chekuri {\em et al.} \cite{CCPV11}]\label{thrm:ContGreedy}
We are given a ground set $\cN$, a monotone submodular function $f:2^{\cN}\rightarrow \mathbb{R}_+$, and a polytope $\mathcal{P}\subseteq [0,1]^{\cN}$.
If $\mathcal{P}\neq \emptyset$ and one can solve in polynomial time $\text{argmax} \left\{ \bw ^T \bx:\bx\in \mathcal{P}\right\} $ for any $ \bw \in \mathbb{R}^{\cN}$, then there exists a polynomial time algorithm that finds $\bx\in \mathcal{P}$ where $ F(\bx)\geq \left( 1-\nicefrac[]{1}{e}\right)F(\bx^*)$.
Here $\bx^*$ is an optimal solution to the problem: $ \max \left\{ F(\by):\by\in \mathcal{P}\right\}$.
\end{theorem}

\section{Algorithms for the \PCSM Problem}\label{sec:multi-linear-relaxation}
\vspace{5pt}
\noindent {\bf{{Preprocessing -- Enumeration with Mixed Constraints:}}}
We define a {\em guess} $D$ to be a triplet $(E_0,E_1,\mathbf{c}')$, where $E_0\subseteq \mathcal{N}$ denotes elements that are discarded, $E_1\subseteq \mathcal{N}$ denotes elements that are chosen, and $\mathbf{c}'\in \mathbb{R}_+^c$ represents a rough estimate (up to a factor of $1+\varepsilon$) of how much an optimal solution $O$ covers each of the covering constraints, {\em i.e.}, $\mathbf{C} \mathbf{1}_O$.
Let us denote by $\tilde{\mathcal{N}} \triangleq \mathcal{N} \setminus (E_0\cup E_1)$ the remaining undetermined elements with respect to guess $D$.

We would like to define when a given fixed guess $D=(E_0,E_1,\mathbf{c}')$ is {\em consistent}, and to this end we introduce the notion of {\em critical} constraints.
For the $i$\textsuperscript{th} packing constraint the residual value that can still be packed is: $(\mathbf{r}_D)_i\triangleq 1-\sum _{\ell \in E_1}\mathbf{P}_{i,\ell}$, where $\mathbf{r}_D\in \mathbb{R}^p$.
For the $j$\textsuperscript{th} covering constraint the residual value that still needs to be covered is: $(\mathbf{s}_D)_j \triangleq \max \left\{ 0,\mathbf{c}'_j-\sum _{\ell \in E_1}\mathbf{C} _{j,\ell}\right\}$, where $\mathbf{s}_D\in \mathbb{R}^c$.
A packing constraint $i$ is called {\em critical} if $(\mathbf{r}_D)_i\leq \delta$, and a covering constraint $j$ is called {\em critical} if $(\mathbf{s}_D)_j\leq \delta \mathbf{c}' _j$ ($\delta \in (0,1)$ is a parameter to be chosen later).
Thus, the collections of critical packing and covering constraints, for a given guess $D$, are given by:
$ Y_D \triangleq \{ i=1,\ldots,p: (\mathbf{r}_D)_i \leq \delta\}$ and $Z_D \triangleq \{ j=1,\ldots,c: (\mathbf{s}_D)_j\leq \delta \mathbf{c}' _j\}$.
Moreover, elements are considered {\em large} if their size is at least some factor $\alpha$ of the residual value of some non-critical constraint ($\alpha \in (0,1)$ is a parameter to be chosen later).
Formally, the collection of large elements with respect to the packing constraints is defined as $P_D\triangleq \{ \ell \in \tilde{\mathcal{N}}: \exists i\ \notin Y_D \text{ s.t. } \mathbf{P}_{i,\ell} \geq \alpha (\mathbf{r}_D)_i\}$, and the collection of large elements with respect to the covering constraints is defined as $C_D=\{ \ell \in \tilde{\mathcal{N}}:\exists j\notin Z_D \text{ s.t. } \mathbf{C}_{j,\ell}\geq \alpha (\mathbf{s}_D)_j\}$.
It is important to note, as previously mentioned, that the notion of a large element is with respect to the residual constraint, as opposed to previous works \cite{chekuriVZ10-rand-exch,KST13} where the definition is with respect to the original constraint.
Let us now formally define when a guess $D$ is called {\em consistent}.
\begin{defn}\label{Definition:consistent}
A guess $D=(E_0,E_1,\mathbf{c}')$ is {\em consistent} if: $(1)$ $E_0\cap E_1 = \emptyset$; $(2)$ $\mathbf{c}'\geq \mathbf{1}_c$; $(3)$ $\mathbf{P} \mathbf{1}_{E_1} \leq \mathbf{1}_p$; and $(4)$ $P_D=C_D=\emptyset$.
\end{defn}
Intuitively, requirement $(1)$ states that a variable cannot be both chosen and discarded, $(2)$ states that the each covering constraint is satisfied by an optimal solution $O$, $(3)$ states the chosen elements $E_1$ do not violate the packing constraints, and $(4)$ states that no large elements remain in any non-critical constraint.

Finally, we need to define when a consistent guess is {\em correct}.
Assume without loss of generality that $O=\{ o_1,\ldots,o_k\}$ and the elements of $O$ are ordered greedily: $ f_{\{ o_1,\ldots,o_i\}}( o_{i+1}) \leq f_{ \{o_1,\ldots,o_{i-1}\}} ( o_i)$ for every $i=1,\ldots,k-1$. In the following definition $\gamma$ is a parameter to be chosen later.
\begin{defn}\label{Definition:correct}
A consistent guess $D=(E_0,E_1,\mathbf{c}')$ is called {\em correct} with respect to $O$ if: $(1)$ $E_1\subseteq O$; $(2)$ $E_0\subseteq \bar{O}$; $(3)$ $\{ o_1,\ldots,o_{\gamma}\}\subseteq E_1$; and $(4)$ $\mathbf{c}' \leq \mathbf{C} \mathbf{1}_O \leq (1+\delta ) \mathbf{c}'$.
\end{defn}
Intuitively, requirement $(1)$ states that the chosen elements $E_1$ are indeed elements of $O$, $(2)$ states that no element of $O$ is discarded, $(3)$ states that the $\gamma$ elements of largest marginal value are all chosen, and $(4)$ states that $\mathbf{c}'$ represents (up to a factor of $1+\delta$) how much $O$ actually covers each of the covering constraints.

We are now ready to present our preprocessing algorithm (Algorithm \ref{alg:OneSidedGuessesEnumeration}), which produces a list $\mathcal{L}$ of consistent guesses that is guaranteed to contain at least one guess that is also correct with respect to $O$.
Lemma \ref{lem:os-correct-guess} summarizes this, its proof appears in the appendix.

\begin{algorithm*}[H] \label{alg:OneSidedGuessesEnumeration}
	\caption{Preprocessing}
	$\cL \leftarrow \varnothing$ \\
	\ForEach{$j_1,\dots,j_c\in\{0,1,\dots,\lceil\log_{1+\delta} n\rceil\}$}{
		Let $\bc' = ((1+\delta)^{j_1},\dots,(1+\delta)^{j_c})$\\
		\ForEach{$E_1\subseteq \cN$ such that $|E_1| \leq \gamma + \nicefrac{(p+c)}{(\alpha \delta)}$}{
			Let $H = (\varnothing, E_1, \bc')$\\
			Let $E_0 = \{\ell \in \cN \setminus E_1: f_{E_1}(\ell) > (\gamma^{-1}) f(E_1)\} \cup P_H \cup C_H$\label{alg:E_0}\\
		        Set $D=(E_0, E_1, \bc')$\\
		        If $D$ is consistent according to Definition \ref{Definition:consistent} add it to $\cL$.
		}
	}
	Output $\cL$.
\end{algorithm*}
\begin{lemma}\label{lem:os-correct-guess}
The output $\mathcal{L}$ of Algorithm~\ref{alg:OneSidedGuessesEnumeration} contains at least one guess $D$ that is correct with respect to some optimal solution $O$.
\end{lemma}
\begin{proof}
Fix any optimal solution $O$.
At least one of the vectors $\bc'$ enumerated by Algorithm~\ref{alg:OneSidedGuessesEnumeration} satisfies property~$(4)$ in Definition~\ref{Definition:correct} with respect to $O$.
Let us fix an iteration in which such a $\bc'$ is enumerated.
Define the ``large'' elements $O$ has with respect to this $\bc'$:
\begin{align}
O_L \triangleq \left\{ \ell \in O:\exists i \text{ s.t. } \mathbf{P}_{i,\ell}\geq \alpha \delta\right\} \cup \left\{ \ell \in O:\exists j \text{ s.t. } \mathbf{C}_{j,\ell}\geq \alpha \delta\bc' _j\right\} .\label{definition:LargeOPT}
\end{align}
Denote by $O_{\gamma}\triangleq \{ o_1,\ldots,o_{\gamma}\}$ the $\gamma$ elements of $O$ with the largest marginal (recall the ordering of $O$ satisfies: $ f_{\{ o_1,\ldots,o_i\}}( o_{i+1}) \leq f_{ \{o_1,\ldots,o_{i-1}\}} ( o_i)$).
Let us fix $E_1\triangleq O_{\gamma}\cup O_L$ and choose $ H\triangleq (\emptyset, E_1,\bc')$.
Clearly, $|E_1|\leq \gamma + \nicefrac[]{(p+c)}{(\alpha \delta)}$ since $|O_{\gamma}|=\gamma$ and $|O_L|\leq \nicefrac[]{(p+c)}{(\alpha\delta)}$.
Hence, we can conclude that $H$ is considered by Algorithm~\ref{alg:OneSidedGuessesEnumeration}.

We fix the iteration in which the above $H$ is considered and show that the resulting $D=(E_0,E_1,\bc')$ of this iteration is correct and consistent (recall that Algorithm~\ref{alg:OneSidedGuessesEnumeration} chooses $E_0=\{\ell \in \cN \setminus E_1: f_{E_1}(\ell) > (\gamma^{-1}) f(E_1)\} \cup P_H \cup C_H$).
The following two observations suffice to complete the proof:

\noindent {\em Observation 1:} $\forall \ell \in O \cup (\mathcal{N}\setminus E_0)$: $ f_{E_1}(\ell)\leq \gamma ^{-1}f(E_1)$.

\noindent {\em Observation 2:} $ O\cap P_H = \emptyset$ and $O\cap C_H=\emptyset $.

\noindent Clearly properties $(1)$ and $(3)$ of Definition~\ref{Definition:correct} are satisfied by construction of $E_1$, $H$, and subsequently $D$.
Property $(2)$ of Definition~\ref{Definition:correct} requires the above two observations, which together imply that no element of $O$ is added to $E_0$ by Algorithm~\ref{alg:OneSidedGuessesEnumeration}.
Thus, all four properties of Definition~\ref{Definition:correct} are satisfied, and we focus on showing that the above $D$ is consistent according to Definition~\ref{Definition:consistent}.
Property $(1)$ of Definition~\ref{Definition:consistent} follows from properties $(1)$ and $(2)$ of Definition~\ref{Definition:correct}.
Property $(2)$ of Definition~\ref{Definition:consistent} follows from the choice of $\bc'$.
Property $(3)$ of Definition~\ref{Definition:consistent} follows from the feasibility of $O$ and property $(1)$ of Definition~\ref{Definition:correct}.
Lastly, property $(4)$ of Definition~\ref{Definition:consistent} follows from the fact that $P_D\subseteq P_H$ and that $P_H\subseteq E_0$, implying that $P_D=\emptyset$ (the same argument applies to $C_D$).
We are left with proving the above two observations.

We start with proving the first observation.
Let $\ell \in O\cup (\mathcal{N}\setminus E_0)$.
If $ \ell \in \mathcal{N}\setminus E_0$ then the observation follows by the construction of $E_0$ in Algorithm~\ref{alg:OneSidedGuessesEnumeration}.
Otherwise, $\ell \in O$.
If $ \ell \in O_{\gamma}$ then we have that $f_{E_1}(\ell)=0$ since $O_{\gamma}\subseteq E_1$.
Otherwise $\ell \in O\setminus O_{\gamma}$.
Note: $$ f_{E_1}(\ell) \leq f_{O_{\gamma}}(\ell)\leq \gamma ^{-1}f(O_{\gamma})\leq \gamma ^{-1}f(E_1).$$
The first inequality follows from diminishing returns and $O_{\gamma}\subseteq E_1$.
The third and last inequality follows from the monotonicity of $f$ and $O_{\gamma}\subseteq E_1$.
Let us focus on the second inequality, and denote $O=\{ o_1,\ldots,o_k\}$ and the sequence $a_i\triangleq f_{\{ o_1,\ldots,o_{i-1}\}}(o_i)$.
The sequence of $a_i$s is monotone non-increasing by the ordering of $O$ and the monotonicity of $f$ implies that all $a_i$s are non-negative.
Note that $a_1+\ldots+a_{\gamma}=f(O_{\gamma})$, thus implying that $ f_{O_{\gamma}}(\ell)\leq \gamma ^{-1}f(O_{\gamma})$ for every $ \ell \in \{ o_{\gamma +1},\ldots,o_k\}$ (otherwise $a_1+\ldots + a_{\gamma} > f(O_{\gamma})$).
The second inequality above, {\em i.e.}, $ f_{O_{\gamma}}(\ell)\leq \gamma ^{-1}f(O_{\gamma})$, now follows since $\ell \in O\setminus O_{\gamma} = \{ o_{\gamma +1},\ldots,o_k\}$.

Let us now focus on proving the second observation.
Let us assume on the contrary that there is an element $\ell$ such that $\ell\in O\cap P_H$.
Recall that $P_H=\{ \ell\in \mathcal{N}\setminus E_1:\exists i\notin Y_H \text{ s.t. } \mathbf{P}_{i,\ell}\geq \alpha (\mathbf{r}_H)_i\}$ where $Y_H=\{ i:(\mathbf{r}_H)_i\leq \delta\}$.
This implies that $ \ell \in O\setminus E_1$, namely that $\ell \notin O_L$, from which we derive that for {\em all} packing constraint $i$ we have that $\mathbf{P}_{i,\ell}<\alpha \delta$.
Since $\ell \in P_H$ we conclude that there exists a packing constraint $i$ for which $ (\mathbf{r}_H)_i\leq \mathbf{P}_{i,\ell}/\alpha$.
Combining the last two bounds we conclude that $(\mathbf{r}_H)_i < \delta$, which implies that the $i$\textsuperscript{th} packing constraint is critical, {\em i.e.}, $i\in Y_H$.
This is a contradiction, and hence $O\cap P_H=\emptyset$.
A similar proof applies to $C_H$ and the covering constraints. \hfill $\square$
\end{proof}

\vspace{5pt}
\noindent {\bf{{Randomized Rounding:}}}
Before presenting our main rounding algorithm, let us define the residual problem we are required to solve given a consistent guess $D$.
First, the residual objective $g:2^{\tilde{\mathcal{N}}}\rightarrow \mathbb{R}_+$ is defined as: $g(S)\triangleq f(S\cup E_1) -f(E_1)$ for every $S\subseteq \tilde{\mathcal{N}}$.
Clearly, $g$ is submodular, non-negative, and monotone.
Second, let us focus on the feasible domain and denote by $\tilde{\mathbf{P}}$ ($\tilde{\mathbf{C}}$) the submatrix of $\mathbf{P}$ ($\mathbf{C}$) obtained by choosing all the columns in $\tilde{\mathcal{N}}$.
Hence, given $D=(E_0,E_1,\mathbf{c}')$ the residual problem is:
\begin{align}
\max \{ g(S) + f(E_1) : S\subseteq \tilde{\mathcal{N}}, \tilde{\mathbf{P}}\mathbf{1}_S\leq \mathbf{r}_D, \tilde{\mathbf{C}}\mathbf{1}_S\geq \mathbf{s}_D\}.\label{OneSidedObjectiveResidual}
\end{align}
In order to formulate the multilinear relaxation of (\ref{OneSidedObjectiveResidual}), consider the following two polytopes:
$\mathcal{P}\triangleq \{ \bx \in [0,1]^{\tilde{\mathcal{N}}}:\tilde{\mathbf{P}}\bx \leq \mathbf{r}_D\}$ and $ \mathcal{C}\triangleq \{ \bx \in [0,1]^{\tilde{\mathcal{N}}}:\tilde{\mathbf{P}}\bx \geq \mathbf{s}_D\}$.
Let $G:[0,1]^{\tilde{\mathcal{N}}}\rightarrow \mathbb{R}_+$ be the multilinear extension of $g$.
Thus, the continuous multilinear relaxation of (\ref{OneSidedObjectiveResidual}) is:
\begin{align}
\max\left\{ f(E_1)+G(\bx):\bx\in [0,1]^{\tilde{\mathcal{N}}}, \bx \in \cP\cap \cC\right\}.\label{ObjectiveContinuous}
\end{align}

Our algorithm performs randomized rounding of a fractional solution to the above relaxation (\ref{ObjectiveContinuous}).
However, this is not enough to obtain our main result and an additional post-processing step is required in which additional elements are discarded.
Since covering constraints are present, one needs to perform the post-processing step in great care.
To this end we denote by $L_D$ the collection of large elements with respect to some critical packing constraint: $L_D\triangleq \{ \ell \in \tilde{\mathcal{N}}:\exists i\in Y_D \text{ s.t. } \mathbf{P}_{i,\ell}\geq \beta \mathbf{r}_D\}$ ($\beta \in (0,1)$ is a parameter to be chosen later).
Intuitively, we would like to discard elements in $L_D$ since choosing any one of those will incur a violation of a packing constraint.
We are now ready to present our rounding algorithm (Algorithm~\ref{alg:OneSidedMainAlgorithm}).

\begin{algorithm*}[H]
	\caption{$(f,\cN,\bP,\bC)$}\label{alg:OneSidedMainAlgorithm}
	Use Algorithm \ref{alg:OneSidedGuessesEnumeration} to obtain a list of guesses $\cL$. \\
	\ForEach{$D=(E_0, E_1,\bc') \in \cL$}{
          Use Theorem~\ref{thrm:ContGreedy} to compute an approximate solution $\bx^*$ to problem (\ref{ObjectiveContinuous}).\label{alg:OneSidedMainAlgorithm-ContGreedy}\\
		Scale down $\x^*$ to $\bar{\x} = \x^*/(1+\delta)$\\
		Let $R_D$ be such that for every $\ell \in \tilde{\mathcal{N}}$ independently: $\Pr \left[ \ell \in R_D\right] = \bar{\x}_{\ell}$.\\
		Let $R'_D = R_D \setminus L_D$.\\
		$S_D \leftarrow E_1 \cup R'_D$.
		}
	$S_{alg} \leftarrow \argmax \left\{f(S_D) : D\in \cL, \bP \cdot \one_{S_D} \leq \one_p, \bC \cdot \one_{S_D} \geq (1-\eps)\one _c \right\}$
\end{algorithm*}
We note that Line $6$ of Algorithm \ref{alg:OneSidedMainAlgorithm} is the post-processing step where all elements of $L_D$ are discarded.
Our analysis of Algorithm \ref{alg:OneSidedMainAlgorithm} shows that in an iteration a correct guess $D$ is examined, with a constant probability, $S_D$ satisfies the packing constraints, violates the covering constraint by only a fraction of $\varepsilon$, and $f(S_D)$ is sufficiently high.

The following lemma gives a lower bound on the value of the fractional solution $\bar{\mathbf{x}}$ computed by Algorithm \ref{alg:OneSidedMainAlgorithm} (for a full proof refer to Lemma~\ref{apx-lem:os-quality-frac}, Appendix~\ref{apx-subsec:one-sided-algo}).

\begin{lemma}\label{lem:os-quality-frac}
  If $D\in \mathcal{L}$ is correct then in the iteration of Algorithm~\ref{alg:OneSidedMainAlgorithm} it is examined the resulting $\bar{\x}$ satisfies: $ G(\bar{\x})\geq (1-\nicefrac{1}{\e}-\delta)\opt-f(E_1)$.
\end{lemma}

Let us now fix an iteration of Algorithm \ref{alg:OneSidedMainAlgorithm} for which $D$ is not only consistent but also correct (the existence of such an iteration is guaranteed by Lemma~\ref{lem:os-correct-guess}).
Intuitively, Algorithm~\ref{alg:OneSidedMainAlgorithm} performs a straightforward randomized rounding where each element $\ell\in \tilde{\mathcal{N}}$ is independently chosen with a probability that corresponds to its fractional value in the solution of the multilinear relaxation (\ref{ObjectiveContinuous}).
However, two key ingredients in Algorithm~\ref{alg:OneSidedMainAlgorithm} are required in order to achieve an $\varepsilon$ violation of the covering constraints and no violation of the packing constraints: $(1)$ {\em scaling}: prior to the randomized rounding $\bx ^*$ is scaled down by a factor $(1+\delta)$ (line $4$ in Algorithm~\ref{alg:OneSidedMainAlgorithm}); and $(2)$ {\em post-processing}: after the randomized rounding all chosen large elements in a critical packing constraint are discarded (line $6$ in Algorithm~\ref{alg:OneSidedMainAlgorithm}).

The first ingredient above (scaling of $\bx ^*$) allows us to prove using standard concentration bounds that with good probability all non-critical packing constraints are not violated.
However, when considering critical packing constraints this does not suffice and the second ingredient above (discarding $L_D$) is required to show that with good probability even the critical packing constraints are not violated.
While discarding $L_D$ is beneficial when considering packing constraints, it might have a destructive effect on both the covering constraints and the value of the objective.
To remedy this we argue that with high probability only few elements in $L_D$ are actually discarded, {\em i.e.}, $|R_D\cap L_D|$ is sufficiently small.
Combining the latter fact with the assumption that the current guess $D$ is not only consistent but also correct, according to Definition~\ref{Definition:correct}, allows us to prove the following lemma (for a full proof refer to Lemma~\ref{apx-lem:MainApproximation}, Appendix~\ref{apx-subsubsec:MainLemma}).

\begin{lemma}\label{lem:RoundingFinal}
For any constant $\eps>0$, choose constants $\alpha=\delta ^3$, $\beta=\nicefrac[]{\delta^2}{(3b)}$, $\gamma = \nicefrac[]{1}{\delta ^3}$, and $\delta < \min \{ \nicefrac[]{1}{(15(p+c))},\nicefrac[]{\eps}{(2+30(p+c)^2)}\}$. With a probability of at least $\nicefrac[]{1}{2}$ Algorithm~\ref{alg:OneSidedMainAlgorithm} outputs a solution $S_{alg} $ satisfying: $(1)$ $\mathbf{P}\mathbf{1}_{S_{alg}}\leq \mathbf{1}_p$; $(2)$ $\mathbf{C}\mathbf{1}_{S_{alg}}\geq (1-\eps)\mathbf{1}_c$; and $(3)$ $ f(S_{alg})\geq (1-\nicefrac[]{1}{e}-\varepsilon)f(O)$.
\end{lemma}
The above lemma suffices to prove Theorem~\ref{thrm:MainApproximation}, as it immediately implies it.

\section{Greedy Dynamic Programming}\label{sec:greedy-dp}
In this section, we present a novel algorithmic approach for submodular maximization that leads to \emph{deterministic} and considerably faster approximation algorithms in several settings. 
Perhaps the most notable application of our approach is Theorem~\ref{thm:submodular-order1-pack}. To the best of our knowledge, it provides the first {\em deterministic} non-trivial approximation algorithm for maximizing a monotone submodular function subject to packing constraints.
To highlight the core idea of our approach, we first present a vanilla version of the greedy dynamic programming approach applied to \PCSM that gives a constant-factor approximation and satisfies the packing constraints, but violates the covering constraints by a factor of $2$ and works in pseudo-polynomial time.

\noindent {\bf{Vanilla Greedy Dynamic Programming:}}
Let us start with a sketch of the algorithm's definition and analysis.
For simplicity of presentation, we assume in the current discussion relating to pseudo-polynomial time algorithms that $\bC \in \N_+^{c \times n}$ and $\bP \in \N_+^{p\times n}$.
Let $\p \in \N_+^p$ and $ \bc \in \N_+^c$ be the packing and covering requirements, respectively.
A solution $S\subseteq \cN$ is feasible if and only if $\bC \cdot \one_S \geq \bc$ and $\bP \cdot \one_S \leq \p$.
We also use the following notations: $c_{\max}=\|\bc\|_{\infty}$, $p_{\max}=\|\p\|_{\infty}$, and $[s]_0=\{0,\dots,s\}$ for every integer $s$.

We define our dynamic programming as follows:
for every $\ind\in[n]_0$, $\bc'\in[n \cdot c_{\max}]_0^c$, and $\p'\in[p_{\max}]_0^p$ a table entry $T[q,\bc',\p']$ is defined and it stores an approximate solution $S$ of cardinality $\ind$ with $\bC \cdot \one_S = \bc'$ and $\bP \cdot \one_S = \p'$.
\footnote{We introduce a dummy solution $\bot$ for denoting undefined table entries, and initialize the entire table with $\bot$.
For the exact details we refer to Appendix~\ref{subsec:vanilla-greedy-dp}.}
For the base case, we set $T[0,\zero_c,\zero_p] \gets \varnothing$.
For populating $T[q,\bc',\p']$ when $q>0$, we examine every set of the form $T[q-1,\bc'-\bC_\ell,\p'-\bP_\ell] \cup \{\ell\}$, where $\ell$ satisfies $\ell\in\cN\setminus T[q-1,\bc'-\bC_\ell,\p'-\bP_\ell]$, $\bc'-\bC_\ell \geq 0$, and $\p'-\bP_\ell \geq 0$.
Out of all these sets, we assign the most valuable one to $T[q,\bc',\p']$.
Note that this operation stores a {\em greedy approximate} solution in the table entry $T[q,\bc',\p']$.
The output of our algorithm is the best of the solutions $T[q,\bc',\p']$, for $1 \leq q \leq n$, $\bc' \geq \bc/2$ and $\p' \leq \p$.

Let us now sketch the analysis of the above algorithm.  Consider an
optimal solution $O$ and appropriately assign to each $o\in O$ a
marginal value $g(o)$ such that $f(O)=\sum_{o\in O}g(o)$. We then
inductively construct an order $o_1,\dots,o_k$ of $O$ with the
intention of upper bounding for every prefix $O_q=\{o_1,\dots,o_q\}$
the value $g(O_q)$ in terms of the value $f(S_q)$ of the table entry
$S_q:=T[q,\bC\one_{O_q}, \bP\one_{O_q}]$ \emph{corresponding to
  $O_q$}. The construction of the sequence $o_1,\dots,o_k$ divides
$[k]$ into $m$ phases where $m$ is a positive integer parameter. A
(possibly empty) phase $i\in [m]$ is characterized by the following
property. Consider a prefix $O_q$ and its corresponding table entry
$S_q$. If $q$ is in phase $i$ then there exists an element
$o_{q+1}\in O\setminus O_{q}$ such that adding $o_{q+1}$ to $S_q$
increases $f$ by at least an amount of $(1-i/m)g(o_{q+1})$. We set
$O_{q+1}=O_q\cup\{o_{q+1}\}$. Thus, in earlier phases we make more
progress in the corresponding dynamic programming solution $S_q$
relative to $g(O_q)$ than in later phases. Additionally, we can prove
a complementing inequality. At the end of phase $i\in[m]$ all elements
in $O \setminus O_{q}$ increase $f$ by no more than
$(1-i/m)g(o_{q+1})$. We prove that this implies that $f(S_q)$ is at
least $i/m\cdot g(O\setminus O_q)$ and thus large relatively to the
\emph{complement} of $O_q$. We set up a factor-revealing linear
program that constructs the worst distribution of the marginal values
over the phases that satisfy the above inequalities. This linear
program gives for every $m$ a lower bound on the approximation
ratio. Analytically, we can show that if $m$ tends to infinity the
optimum value of the LP converges to~$\nicefrac{1}{\e}$. This leads to
the following lemma (for its proof refer to Appendix~\ref{apx-subsec:factor-rev-lp}).

\begin{lemma}
 \label{lem:greedy-dp-vanilla}
Assuming $p$ and $c$ are constants, the vanilla greedy dynamic programming algorithm for \PCSM runs in pseudo-polynomial time $O(n^2p_{\max}c_{\max})$ and outputs a solution $S\subseteq \cN$ that satisfies: $(1)$ $f(S) \geq \left(\nicefrac[]{1}{\e}\right)\cdot f(O)$, $(2)$ $\bP \mathbf{1}_S\leq \p$ and $\bC \mathbf{1}_S\geq \nicefrac{1}{2}\cdot \bc$.
\end{lemma}

\noindent {\bf{Applications and Extensions of Greedy Dynamic Programming Approach}}

We briefly explain the applications of the approach to the various specific settings and the required tailored algorithmic extensions to the vanilla version of the algorithm.

\noindent {\bf\em{Scaling, guessing and post-processing for packing constraints}}
An immediate consequence of Lemma~\ref{lem:greedy-dp-vanilla} is a
{\em deterministic} $\left( \nicefrac[]{1}{e}\right)$-approximation for the case of constantly many packing constraints that runs in pseudo-polynomial time. We can apply standard scaling techniques to achieve truly polynomial time. This may, however, introduce a violation of the constraints within a factor of $(1+\varepsilon)$.
To avoid this violation, we can apply a pre-processing and post-processing by Kulik et al.~\cite{KST13} to achieve Theorem~\ref{thm:submodular-order1-pack}.

\noindent {\bf\em{Forbidden sets for a single packing and a single covering constraints.}}
In this setting we are able to ensure a $(1-\eps)$-violation of the covering constraints by using the concept of \emph{forbidden sets}.
Intuitively, we exclude the elements of these set from being included to the dynamic programming table in order to be able to complete the table entries to solutions with only small violation.

Fix some $\eps>0$. By guessing we assume that we know the set $G$ of all, at most $1/\eps$ elements $\ell$ from the optimum solution with $\bP_\ell>\eps\cdot \p$. We can guess $G$ using brute force in $n^{O(1/\eps)}$ time. This allows us to remove all elements with $\bP_\ell \geq \eps\cdot \p$ from the instance. Let $\cN'$ be the rest of the elements. (For consistency reasons, we use bold-face vector notation here also for dimension one.)

Fix an order of $\cN'$ in which the elements are sorted in a non-increasing order of $\bC_\ell/\bP_\ell$ values, breaking ties arbitrarily. Let $\cN_i$ be the set of the first $i$ elements in this order. For any $\p'\leq \p$, let $F_{\p'}$ be the smallest set $\cN_i$ with $\bP\one_{\cN_i}\geq \p-\p'$. Note that the profit of $F_{\p'}$ is at least the profit of any subset of $\cN'$ with packing value at most $\p-\p'$ and that the packing value of $F_{\p'}$ is no larger than $(1+\eps)\p-\p'$. Also note that for any $0 \leq \p'\leq \p'' \leq \p$, it holds that $F_{\p''} \subseteq F_{\p'}$.

Now we explain the modified Greedy-DP that incorporates the guessing and the forbidden sets ideas. Let $G$ be the set of the guessed big elements as described above.
For the base case, we set $T[\bC \one_G, \bP \one_G] = G$ and $T[\bc', \p']=\bot$ for all table entries with $\bc' \neq \bC \one_G$ or $\p' \neq \bP \one_G$.

In order to compute $T[\bc', \p']$, we look at every set of the form $T[\bc'-\bC_\ell,\p'-\bP_\ell] \cup \{\ell\}$, where $\ell\in\cN \setminus (T[\bc'-\bC_\ell,\p'-\bP_\ell] \cup F_{\p'})$, $\bc'- \bC_\ell \geq 0$, and $\p'- \bP_\ell \geq 0$. Notice that we forbid elements belonging to $F_{\p'}$ to be included in any table entry of the form $T[\bc', \p']$. Now out of all these sets, we assign the most valuable set to $T[\bc', \p']$. The output of our algorithm is the best of the solutions $T[\bc', \p'] \cup F_{\p'}$, such that $\bc' + \bC \one_{F_{\p'}} \geq \bc$.

By means of a more sophisticated factor-revealing LP, we obtain
Theorem~\ref{thm:submodular-one}. Finally, if the packing constraint
is actually a cardinality constraint we can assume that
$\eps<1/\p$. Hence, there will be no violation of the cardinality
constraint and also guessing can be avoided.


\section{Extensions: Matroid Independence and Multi-Objective}
Refer to the Appendix~\ref{App:Extensions} for the extensions that deal with a matroid independence constraint and with multiple objectives.

\subsubsection*{Acknowledgement} Joachim Spoerhase and Sumedha Uniyal thank an anonymous reviewer for pointing them to the fact that Theorem~\ref{thrm:ContGreedy} also applies to polytopes that are not down-closed, which makes it possible to apply a randomized rounding approach.
\newpage

\bibliography{Bibliography}

\begin{thebibliography}{10}

\bibitem{AS99}
A.~A. Ageev and M.~I. Sviridenko.
\newblock An 0.828 approximation algorithm for the uncapacitated facility
  location problem.
\newblock {\em Discrete Appl. Math.}, 93(2-3):149--156, July 1999.

\bibitem{SA11}
Shabbir Ahmed and Alper Atamt{\"u}rk.
\newblock Maximizing a class of submodular utility functions.
\newblock {\em Mathematical Programming}, 128(1):149--169, Jun 2011.

\bibitem{ABG16}
Per Austrin, Siavosh Benabbas, and Konstantinos Georgiou.
\newblock Better balance by being biased: {A} 0.8776-approximation for max
  bisection.
\newblock {\em {ACM} Trans. Algorithms}, 13(1):2:1--2:27, 2016.

\bibitem{B13}
Francis Bach.
\newblock {\em Learning with Submodular Functions: A Convex Optimization
  Perspective}.
\newblock Now Publishers Inc., Hanover, MA, USA, 2013.

\bibitem{BHK14}
L.~Bordeaux, Y.~Hamadi, and P.~Kohli.
\newblock {\em Tractability: Practical Approaches to Hard Problems}.
\newblock Cambridge University Press, 2014.

\bibitem{bfns14}
Niv Buchbinder, Moran Feldman, Joseph Naor, and Roy Schwartz.
\newblock Submodular maximization with cardinality constraints.
\newblock In {\em Proc. 25th Annual {ACM-SIAM} Symposium on Discrete Algorithms
  (SODA'14)}, pages 1433--1452, 2014.

\bibitem{bru16}
Jaros{\l}aw Byrka, Bartosz Rybicki, and Sumedha Uniyal.
\newblock An approximation algorithm for uniform capacitated k-median problem
  with $1+\epsilon$ capacity violation.
\newblock In {\em International Conference on Integer Programming and
  Combinatorial Optimization, {(IPCO'16)}}, pages 262--274. Springer, 2016.

\bibitem{CCPV11}
Gruia Calinescu, Chandra Chekuri, Martin P\'{a}l, and Jan Vondr\'{a}k.
\newblock Maximizing a monotone submodular function subject to a matroid
  constraint.
\newblock {\em SIAM J. Comput.}, 40(6):1740--1766, December 2011.

\bibitem{CK05}
Chandra Chekuri and Sanjeev Khanna.
\newblock A polynomial time approximation scheme for the multiple knapsack
  problem.
\newblock {\em SIAM Journal on Computing}, 35(3):713--728, 2005.

\bibitem{CVZ10-arxiv}
Chandra Chekuri, Jan Vondr{\'{a}}k, and Rico Zenklusen.
\newblock Randomized pipage rounding for matroid polytopes and applications.
\newblock {\em CoRR}, abs/0909.4348, 2009.
\newblock URL: \url{http://arxiv.org/abs/0909.4348}.

\bibitem{chekuriVZ10-rand-exch}
Chandra Chekuri, Jan Vondr{\'{a}}k, and Rico Zenklusen.
\newblock Dependent randomized rounding via exchange properties of
  combinatorial structures.
\newblock In {\em Proc. 51th Annual {IEEE} Symposium on Foundations of Computer
  Science (FOCS'10)}, pages 575--584, 2010.

\bibitem{CL06-Chernoff}
Fan Chung and Linyuan Lu.
\newblock Concentration inequalities and martingale inequalities: A survey.
\newblock {\em Internet Mathematics}, 3(1):79--127, 2006.

\bibitem{cn02}
Julia Chuzhoy and Joseph~Seffi Naor.
\newblock Covering problems with hard capacities.
\newblock In {\em Proc. 43rd Annual {IEEE} Symposium on Foundations of Computer
  Science, (FOCS'02)}, pages 481--489. IEEE, 2002.

\bibitem{CKR06}
Reuven Cohen, Liran Katzir, and Danny Raz.
\newblock An efficient approximation for the generalized assignment problem.
\newblock {\em Inf. Process. Lett.}, 100(4):162--166, November 2006.

\bibitem{CFN77b}
Gerard Cornuejols, Marshall Fisher, and George~L. Nemhauser.
\newblock On the uncapacitated location problem.
\newblock In {\em Studies in Integer Programming}, volume~1 of {\em Annals of
  Discrete Mathematics}, pages 163 -- 177. Elsevier, 1977.

\bibitem{CFN77a}
Gerard Cornuejols, Marshall~L. Fisher, and George~L. Nemhauser.
\newblock Location of bank accounts to optimize float: An analytic study of
  exact and approximate algorithms.
\newblock {\em Management Science}, 23(8):789--810, 1977.

\bibitem{dl16}
H.~G{\"{o}}kalp Demirci and Shi Li.
\newblock Constant approximation for capacitated k-median with
  ($1+\epsilon$)-capacity violation.
\newblock In {\em 43rd International Colloquium on Automata, Languages, and
  Programming, {(ICALP'16)}}, pages 73:1--73:14, 2016.

\bibitem{DRS12}
Shaddin Dughmi, Tim Roughgarden, and Mukund Sundararajan.
\newblock Revenue submodularity.
\newblock {\em Theory of Computing}, 8(1):95--119, 2012.

\bibitem{FV06}
U.~Feige and J.~Vondrak.
\newblock Approximation algorithms for allocation problems: Improving the
  factor of 1 - 1/e.
\newblock In {\em 2006 47th Annual IEEE Symposium on Foundations of Computer
  Science (FOCS'06)}, pages 667--676, 2006.

\bibitem{F98}
Uriel Feige.
\newblock A threshold of ln n for approximating set cover.
\newblock {\em J. ACM}, 45(4):634--652, July 1998.

\bibitem{FGMS06}
Lisa Fleischer, Michel~X Goemans, Vahab~S Mirrokni, and Maxim Sviridenko.
\newblock Tight approximation algorithms for maximum general assignment
  problems.
\newblock In {\em Proc. 17th annual ACM-SIAM Symposium on Discrete Algorithm,
  {(SODA'06)}}, pages 611--620. SIAM, 2006.

\bibitem{FJ95}
Alan Frieze and Mark Jerrum.
\newblock Improved approximation algorithms for max k-cut and max bisection.
\newblock In Egon Balas and Jens Clausen, editors, {\em Integer Programming and
  Combinatorial Optimization}, pages 1--13. Springer Berlin Heidelberg, 1995.

\bibitem{G14}
J.~Gillenwater.
\newblock {\em {Approximate Inference for Determinantal Point Processes}}.
\newblock PhD thesis, University of Pennsylvania, 2014.

\bibitem{GKT12}
Jennifer Gillenwater, Alex Kulesza, and Ben Taskar.
\newblock Near-optimal map inference for determinantal point processes.
\newblock In F.~Pereira, C.~J.~C. Burges, L.~Bottou, and K.~Q. Weinberger,
  editors, {\em Advances in Neural Information Processing Systems 25}, pages
  2735--2743. Curran Associates, Inc., 2012.

\bibitem{GW95}
Michel~X. Goemans and David~P. Williamson.
\newblock Improved approximation algorithms for maximum cut and satisfiability
  problems using semidefinite programming.
\newblock {\em J. ACM}, 42(6):1115--1145, November 1995.

\bibitem{HZ01}
Eran Halperin and Uri Zwick.
\newblock Combinatorial approximation algorithms for the maximum directed cut
  problem.
\newblock In {\em Proceedings of the Twelfth Annual ACM-SIAM Symposium on
  Discrete Algorithms}, SODA '01, pages 1--7, 2001.

\bibitem{HMS08}
Jason Hartline, Vahab Mirrokni, and Mukund Sundararajan.
\newblock Optimal marketing strategies over social networks.
\newblock In {\em Proceedings of the 17th International Conference on World
  Wide Web}, WWW '08, pages 189--198, 2008.

\bibitem{H96}
Johan H{\aa}stad.
\newblock Clique is hard to approximate within $n^{(1-\varepsilon)}$.
\newblock In {\em Acta Mathematica}, pages 627--636, 1996.

\bibitem{H01}
Johan H{\aa}stad.
\newblock Some optimal inapproximability results.
\newblock {\em J. ACM}, 48(4):798--859, July 2001.

\bibitem{HK14}
Xinran He and David Kempe.
\newblock Stability of influence maximization.
\newblock In {\em Proceedings of the 20th ACM SIGKDD International Conference
  on Knowledge Discovery and Data Mining}, KDD '14, pages 1256--1265, 2014.

\bibitem{jms02}
Kamal Jain, Mohammad Mahdian, and Amin Saberi.
\newblock A new greedy approach for facility location problems.
\newblock In {\em Proceedings of the 34th Annual ACM Symposium on Theory of
  Computing, (STOC'02)}, pages 731--740. ACM, 2002.

\bibitem{K72}
Richard~M. Karp.
\newblock {\em Reducibility among Combinatorial Problems}, pages 85--103.
\newblock Springer US, 1972.

\bibitem{kpp04}
H.~Kellerer, U.~Pferschy, and D.~Pisinger.
\newblock {\em Knapsack Problems}.
\newblock Springer, 2004.

\bibitem{KKT15}
David Kempe, Jon Kleinberg, and \'{E}va Tardos.
\newblock Maximizing the spread of influence through a social network.
\newblock {\em Theory of Computing}, 11(4):105--147, 2015.

\bibitem{KKMO07}
Subhash Khot, Guy Kindler, Elchanan Mossel, and Ryan O'Donnell.
\newblock Optimal inapproximability results for max‐cut and other
  2‐variable csps?
\newblock {\em SIAM Journal on Computing}, 37(1):319--357, 2007.

\bibitem{KMN99}
Samir Khuller, Anna Moss, and Joseph~Seffi Naor.
\newblock The budgeted maximum coverage problem.
\newblock {\em Information Processing Letters}, 70(1):39--45, 1999.

\bibitem{KT11}
Alex Kulesza and Ben Taskar.
\newblock Learning determinantal point processes.
\newblock In {\em {UAI} 2011, Proceedings of the Twenty-Seventh Conference on
  Uncertainty in Artificial Intelligence, Barcelona, Spain, July 14-17, 2011},
  pages 419--427, 2011.

\bibitem{KT12}
Alex Kulesza and Ben Taskar.
\newblock Determinantal point processes for machine learning.
\newblock {\em Foundations and Trends in Machine Learning}, 5(2-3):123--286,
  2012.

\bibitem{KST13}
Ariel Kulik, Hadas Shachnai, and Tami Tamir.
\newblock Approximations for monotone and nonmonotone submodular maximization
  with knapsack constraints.
\newblock {\em Mathematics of Operations Research}, 38(4):729--739, 2013.
\newblock preliminary version appeared in SODA'09.

\bibitem{LMNS09}
Jon Lee, Vahab~S Mirrokni, Viswanath Nagarajan, and Maxim Sviridenko.
\newblock Non-monotone submodular maximization under matroid and knapsack
  constraints.
\newblock In {\em Proc. 41st Annual ACM Symposium on Theory Of Computing,
  {(STOC'09)}}, pages 323--332. ACM, 2009.

\bibitem{Li16}
Shi Li.
\newblock Approximating capacitated k-median with $(1 + \epsilon)k$ open
  facilities.
\newblock In {\em Proceedings of the 27th Annual ACM-SIAM Symposium on Discrete
  Algorithms, (SODA'16)}, pages 786--796. SIAM, 2016.

\bibitem{HL10}
Hui Lin and Jeff Bilmes.
\newblock Multi-document summarization via budgeted maximization of submodular
  functions.
\newblock In {\em Human Language Technologies: The 2010 Annual Conference of
  the North American Chapter of the Association for Computational Linguistics},
  HLT '10, pages 912--920, 2010.

\bibitem{HL11}
Hui Lin and Jeff Bilmes.
\newblock A class of submodular functions for document summarization.
\newblock In {\em Proceedings of the 49th Annual Meeting of the Association for
  Computational Linguistics: Human Language Technologies - Volume 1}, HLT '11,
  pages 510--520, 2011.

\bibitem{MSSU18}
Eyal Mizrachi, Roy Schwartz, Joachim Spoerhase, and Sumedha Uniyal.
\newblock A tight approximation for submodular maximization with mixed packing
  and covering constraints.
\newblock {\em CoRR}, abs/1804.10947, 2018.
\newblock URL: \url{http://arxiv.org/abs/1804.10947}.

\bibitem{NWF78}
George~L. Nemhauser, Laurence~A. Wolsey, and Marshall~L. Fisher.
\newblock An analysis of approximations for maximizing submodular set functions
  - {I}.
\newblock {\em Math. Program.}, 14(1):265--294, 1978.

\bibitem{NW78}
George~L Nemhauser and Leonard~A Wolsey.
\newblock Best algorithms for approximating the maximum of a submodular set
  function.
\newblock {\em Mathematics of operations research}, 3(3):177--188, 1978.

\bibitem{PY00}
Christos~H. Papadimitriou and Mihalis Yannakakis.
\newblock On the approximability of trade-offs and optimal access of web
  sources.
\newblock In {\em {FOCS}}, pages 86--92. {IEEE} Computer Society, 2000.

\bibitem{PW10}
Mihai Patrascu and Ryan Williams.
\newblock On the possibility of faster {SAT} algorithms.
\newblock In {\em Proc. 21st Annual {ACM-SIAM} Symposium on Discrete Algorithms
  (SODA'10)}, pages 1065--1075, 2010.

\bibitem{SU13}
Andreas~S. Schulz and Nelson~A. Uhan.
\newblock Approximating the least core value and least core of cooperative
  games with supermodular costs.
\newblock {\em Discrete Optimization}, 10(2):163 -- 180, 2013.

\bibitem{s04}
Maxim Sviridenko.
\newblock A note on maximizing a submodular set function subject to a knapsack
  constraint.
\newblock {\em Oper. Res. Lett.}, 32(1):41--43, 2004.

\bibitem{v09}
Jan Vondr{\'{a}}k.
\newblock Symmetry and approximability of submodular maximization problems.
\newblock In {\em Proc. 50th Annual {IEEE} Symposium on Foundations of Computer
  Science (FOCS'09)}, pages 651--670. IEEE, 2009.

\bibitem{Jan18}
Jan Vondrak.
\newblock Personal communication, 2018.

\bibitem{W82}
Laurence~A. Wolsey.
\newblock Maximising real-valued submodular functions: Primal and dual
  heuristics for location problems.
\newblock {\em Mathematics of Operations Research}, 7(3):410--425, 1982.

\end{thebibliography}

\newpage
\appendix

\section{Algorithms for the \PCSM Problem}  \label{App:OneSidedFeasibility}

Below, we give a full technical description of the proof of our main Theorem~\ref{thrm:MainApproximation}. We first describe a pre-processing step followed by a multilinear relaxation based randomized rounding algorithm which includes a post-processing step in the end. We refer to our techniques in Section~\ref{sec:Introduction} and Section~\ref{sec:multi-linear-relaxation} for a comprehensive intuitive exposition.

\subsection{Preprocessing: Enumeration with Mixed Constraints}
\label{apx-subsec:one-sided-preprocess}

We define a {\em guess} $D$ to be a triplet $(E_0,E_1,\mathbf{c}')$, where $E_0\subseteq \mathcal{N}$ denotes elements that are discarded, $E_1\subseteq \mathcal{N}$ denotes elements that are chosen, and $\mathbf{c}'\in \mathbb{R}_+^c$ represents a rough estimate (up to a factor of $1+\varepsilon$) of how much an optimal solution $O$ covers each of the covering constraints, {\em i.e.}, $\mathbf{C} \mathbf{1}_O$.
Let us denote by $\tilde{\mathcal{N}} \triangleq \mathcal{N} \setminus (E_0\cup E_1)$ the remaining undetermined elements with respect to guess $D$.

We would like to define when a given fixed guess $D=(E_0,E_1,\mathbf{c}')$ is {\em consistent}, and to this end we introduce the notion of {\em critical} constraints.
For the $i$\textsuperscript{th} packing constraint the residual value that can still be packed is: $(\mathbf{r}_D)_i\triangleq 1-\sum _{\ell \in E_1}\mathbf{P}_{i,\ell}$, where $\mathbf{r}_D\in \mathbb{R}^p$.
For the $j$\textsuperscript{th} covering constraint the residual value that still needs to be covered is: $(\mathbf{s}_D)_j \triangleq \max \left\{ 0,\mathbf{c}'_j-\sum _{\ell \in E_1}\mathbf{C} _{j,\ell}\right\}$, where $\mathbf{s}_D\in \mathbb{R}^c$.
A packing constraint $i$ is called {\em critical} if $(\mathbf{r}_D)_i\leq \delta$, and a covering constraint $j$ is called {\em critical} if $(\mathbf{s}_D)_j\leq \delta \mathbf{c}' _j$ ($\delta \in (0,1)$ is a parameter to be chosen later).
Thus, the collections of critical packing and covering constraints, for a given guess $D$, are given by:
$ Y_D \triangleq \{ i=1,\ldots,p: (\mathbf{r}_D)_i \leq \delta\}$ and $Z_D \triangleq \{ j=1,\ldots,c: (\mathbf{s}_D)_j\leq \delta \mathbf{c}' _j\}$.
Moreover, elements are considered {\em large} if their size is at least some factor $\alpha$ of the residual value of some non-critical constraint ($\alpha \in (0,1)$ is a parameter to be chosen later).
Formally, the collection of large elements with respect to the packing constraints is defined as $P_D\triangleq \{ \ell \in \tilde{\mathcal{N}}: \exists i\ \notin Y_D \text{ s.t. } \mathbf{P}_{i,\ell} \geq \alpha (\mathbf{r}_D)_i\}$, and the collection of large elements with respect to the covering constraints is defined as $C_D=\{ \ell \in \tilde{\mathcal{N}}:\exists j\notin Z_D \text{ s.t. } \mathbf{C}_{j,\ell}\geq \alpha (\mathbf{s}_D)_j\}$.
It is important to note, as previously mentioned, that the notion of a large element is with respect to the residual constraint, as opposed to previous works \cite{chekuriVZ10-rand-exch,KST13} where the definition is with respect to the original constraint.
Let us now formally define when a guess $D$ is called {\em consistent}.
\begin{defn}\label{apx-def:OneSidedConsistentGuess}
	A guess $D=(E_0,E_1, \bc')$ is {\em consistent} if:
	\begin{enumerate}
		\item\label{apx-item:oneSidedConsistentDisjoint} $E_0\cap E_1 = \varnothing$.
		\item\label{apx-item:oneSidedConsistentGuessCover}   $\bc'\geq\one_c$.
		\item\label{apx-item:oneSidedConsistentGuessPack} $\bP\one_{E_1} \leq \one_p$.
		\item\label{apx-item:oneSidedConsistentGuessNoBig} $P_D = C_D = \varnothing$.
	\end{enumerate}
\end{defn}

Intuitively, requirement $(1)$ states that a variable cannot be both chosen and discarded, $(2)$ states that the each covering constraint is satisfied by an optimal solution $O$, $(3)$ states the chosen elements $E_1$ do not violate the packing constraints, and $(4)$ states that no large elements remain in any non-critical constraint.

\paragraph*{Correct Guesses}
Finally, we need to define when a consistent guess is {\em correct}.
Assume without loss of generality that $O=\{ o_1,\ldots,o_k\}$ and the elements of $O$ are ordered greedily: $ f_{\{ o_1,\ldots,o_i\}}( o_{i+1}) \leq f_{ \{o_1,\ldots,o_{i-1}\}} ( o_i)$ for every $i=1,\ldots,k-1$. In the following definition $\gamma$ is a parameter to be chosen later.

\begin{defn}\label{apx-def:OneSidedCorrectGuess}
  A consistent guess $D=(E_0,E_1,\bc')$ is called \emph{correct} with respect to $O$ if:
	\begin{enumerate}
		\item\label{apx-item:corr-os-e1} $E_1\subseteq O$,
		\item\label{apx-item:corr-os-e0} $E_0\subseteq \overline{O}$,
		\item\label{apx-item:corr-os-value} $\{o_1,\dots, o_\gamma\}\subseteq E_1$,
		\item\label{apx-item:corr-os-cov} $\one_c\leq\bc'\leq\bC\one_{O} < (1+\delta)\bc'$.
	\end{enumerate}
\end{defn}

Intuitively, requirement $(1)$ states that the chosen elements $E_1$ are indeed elements of $O$, $(2)$ states that no element of $O$ is discarded, $(3)$ states that the $\gamma$ elements of largest marginal value are all chosen, and $(4)$ states that $\mathbf{c}'$ represents (up to a factor of $1+\delta$) how much $O$ actually covers each of the covering constraints.

We are now ready to present our preprocessing algorithm (Algorithm \ref{apx-alg:OneSidedGuessesEnumeration}), which produces a list $\mathcal{L}$ of consistent guesses that is guaranteed to contain at least one guess that is also correct with respect to $O$.
Lemma \ref{apx-lem:os-correct-guess} summarizes this, its proof appears in the appendix.

\begin{algorithm*}[H] \label{apx-alg:OneSidedGuessesEnumeration}
	\caption{Preprocessing}
	$\cL \leftarrow \varnothing$ \\
	\ForEach{$j_1,\dots,j_c\in\{0,1,\dots,\lceil\log_{1+\delta} n\rceil\}$}{
		Let $\bc' = ((1+\delta)^{j_1},\dots,(1+\delta)^{j_c})$\\
		\ForEach{$E_1\subseteq \cN$ such that $|E_1| \leq \gamma + \nicefrac{(p+c)}{(\alpha \delta)}$}{
			Let $H = (\varnothing, E_1, \bc')$\\
			Let $E_0 = \{\ell \in \cN \setminus E_1: f_{E_1}(\ell) > (\gamma^{-1}) f(E_1)\} \cup P_H \cup C_H$\label{apx-alg:E_0}\\
		        Set $D=(E_0, E_1, \bc')$\\
		        If $D$ is consistent according to Definition \ref{Definition:consistent} add it to $\cL$.
		}
	}
	Output $\cL$.
\end{algorithm*}
\begin{lemma}\label{apx-lem:os-correct-guess}
The output $\mathcal{L}$ of Algorithm~\ref{apx-alg:OneSidedGuessesEnumeration} contains at least one guess $D$ that is correct with respect to some optimal solution $O$.
\end{lemma}
\begin{proof}
Fix any optimal solution $O$.
At least one of the vectors $\bc'$ enumerated by Algorithm~\ref{apx-alg:OneSidedGuessesEnumeration} satisfies property~$(4)$ in Definition~\ref{Definition:correct} with respect to $O$.
Let us fix an iteration in which such a $\bc'$ is enumerated.
Define the ``large'' elements $O$ has with respect to this $\bc'$:
\begin{align}
O_L \triangleq \left\{ \ell \in O:\exists i \text{ s.t. } \mathbf{P}_{i,\ell}\geq \alpha \delta\right\} \cup \left\{ \ell \in O:\exists j \text{ s.t. } \mathbf{C}_{j,\ell}\geq \alpha \delta\bc' _j\right\} .\label{definition:LargeOPT}
\end{align}
Denote by $O_{\gamma}\triangleq \{ o_1,\ldots,o_{\gamma}\}$ the $\gamma$ elements of $O$ with the largest marginal (recall the ordering of $O$ satisfies: $ f_{\{ o_1,\ldots,o_i\}}( o_{i+1}) \leq f_{ \{o_1,\ldots,o_{i-1}\}} ( o_i)$).
Let us fix $E_1\triangleq O_{\gamma}\cup O_L$ and choose $ H\triangleq (\emptyset, E_1,\bc')$.
Clearly, $|E_1|\leq \gamma + \nicefrac[]{(p+c)}{(\alpha \delta)}$ since $|O_{\gamma}|=\gamma$ and $|O_L|\leq \nicefrac[]{(p+c)}{(\alpha\delta)}$.
Hence, we can conclude that $H$ is considered by Algorithm~\ref{apx-alg:OneSidedGuessesEnumeration}.

We fix the iteration in which the above $H$ is considered and show that the resulting $D=(E_0,E_1,\bc')$ of this iteration is correct and consistent (recall that Algorithm~\ref{apx-alg:OneSidedGuessesEnumeration} chooses $E_0=\{\ell \in \cN \setminus E_1: f_{E_1}(\ell) > (\gamma^{-1}) f(E_1)\} \cup P_H \cup C_H$).
The following two observations suffice to complete the proof:

\noindent {\em Observation 1:} $\forall \ell \in O \cup (\mathcal{N}\setminus E_0)$: $ f_{E_1}(\ell)\leq \gamma ^{-1}f(E_1)$.

\noindent {\em Observation 2:} $ O\cap P_H = \emptyset$ and $O\cap C_H=\emptyset $.

\noindent Clearly properties $(1)$ and $(3)$ of Definition~\ref{Definition:correct} are satisfied by construction of $E_1$, $H$, and subsequently $D$.
Property $(2)$ of Definition~\ref{Definition:correct} requires the above two observations, which together imply that no element of $O$ is added to $E_0$ by Algorithm~\ref{apx-alg:OneSidedGuessesEnumeration}.
Thus, all four properties of Definition~\ref{Definition:correct} are satisfied, and we focus on showing that the above $D$ is consistent according to Definition~\ref{Definition:consistent}.
Property $(1)$ of Definition~\ref{Definition:consistent} follows from properties $(1)$ and $(2)$ of Definition~\ref{Definition:correct}.
Property $(2)$ of Definition~\ref{Definition:consistent} follows from the choice of $\bc'$.
Property $(3)$ of Definition~\ref{Definition:consistent} follows from the feasibility of $O$ and property $(1)$ of Definition~\ref{Definition:correct}.
Lastly, property $(4)$ of Definition~\ref{Definition:consistent} follows from the fact that $P_D\subseteq P_H$ and that $P_H\subseteq E_0$, implying that $P_D=\emptyset$ (the same argument applies to $C_D$).
We are left with proving the above two observations.

We start with proving the first observation.
Let $\ell \in O\cup (\mathcal{N}\setminus E_0)$.
If $ \ell \in \mathcal{N}\setminus E_0$ then the observation follows by the construction of $E_0$ in Algorithm~\ref{apx-alg:OneSidedGuessesEnumeration}.
Otherwise, $\ell \in O$.
If $ \ell \in O_{\gamma}$ then we have that $f_{E_1}(\ell)=0$ since $O_{\gamma}\subseteq E_1$.
Otherwise $\ell \in O\setminus O_{\gamma}$.
Note: $$ f_{E_1}(\ell) \leq f_{O_{\gamma}}(\ell)\leq \gamma ^{-1}f(O_{\gamma})\leq \gamma ^{-1}f(E_1).$$
The first inequality follows from diminishing returns and $O_{\gamma}\subseteq E_1$.
The third and last inequality follows from the monotonicity of $f$ and $O_{\gamma}\subseteq E_1$.
Let us focus on the second inequality, and denote $O=\{ o_1,\ldots,o_k\}$ and the sequence $a_i\triangleq f_{\{ o_1,\ldots,o_{i-1}\}}(o_i)$.
The sequence of $a_i$s is monotone non-increasing by the ordering of $O$ and the monotonicity of $f$ implies that all $a_i$s are non-negative.
Note that $a_1+\ldots+a_{\gamma}=f(O_{\gamma})$, thus implying that $ f_{O_{\gamma}}(\ell)\leq \gamma ^{-1}f(O_{\gamma})$ for every $ \ell \in \{ o_{\gamma +1},\ldots,o_k\}$ (otherwise $a_1+\ldots + a_{\gamma} > f(O_{\gamma})$).
The second inequality above, {\em i.e.}, $ f_{O_{\gamma}}(\ell)\leq \gamma ^{-1}f(O_{\gamma})$, now follows since $\ell \in O\setminus O_{\gamma} = \{ o_{\gamma +1},\ldots,o_k\}$.

Let us now focus on proving the second observation.
Let us assume on the contrary that there is an element $\ell$ such that $\ell\in O\cap P_H$.
Recall that $P_H=\{ \ell\in \mathcal{N}\setminus E_1:\exists i\notin Y_H \text{ s.t. } \mathbf{P}_{i,\ell}\geq \alpha (\mathbf{r}_H)_i\}$ where $Y_H=\{ i:(\mathbf{r}_H)_i\leq \delta\}$.
This implies that $ \ell \in O\setminus E_1$, namely that $\ell \notin O_L$, from which we derive that for {\em all} packing constraint $i$ we have that $\mathbf{P}_{i,\ell}<\alpha \delta$.
Since $\ell \in P_H$ we conclude that there exists a packing constraint $i$ for which $ (\mathbf{r}_H)_i\leq \mathbf{P}_{i,\ell}/\alpha$.
Combining the last two bounds we conclude that $(\mathbf{r}_H)_i < \delta$, which implies that the $i$\textsuperscript{th} packing constraint is critical, {\em i.e.}, $i\in Y_H$.
This is a contradiction, and hence $O\cap P_H=\emptyset$.
A similar proof applies to $C_H$ and the covering constraints. \hfill $\square$
\end{proof}

\subsection{Algorithm}
\label{apx-subsec:one-sided-algo}

Before presenting our main rounding algorithm, let us define the residual problem we are required to solve given a consistent guess $D$.
First, the residual objective $g:2^{\tilde{\mathcal{N}}}\rightarrow \mathbb{R}_+$ is defined as: $g(S)\triangleq f(S\cup E_1) -f(E_1)$ for every $S\subseteq \tilde{\mathcal{N}}$.
Clearly, $g$ is submodular, non-negative, and monotone.
Second, let us focus on the feasible domain and denote by $\tilde{\mathbf{P}}$ ($\tilde{\mathbf{C}}$) the submatrix of $\mathbf{P}$ ($\mathbf{C}$) obtained by choosing all the columns in $\tilde{\mathcal{N}}$.
Hence, given $D=(E_0,E_1,\mathbf{c}')$ the residual problem is:
\begin{align}
\max \{ g(S) + f(E_1) : S\subseteq \tilde{\mathcal{N}}, \tilde{\mathbf{P}}\mathbf{1}_S\leq \mathbf{r}_D, \tilde{\mathbf{C}}\mathbf{1}_S\geq \mathbf{s}_D\}.\label{apx-OneSidedObjectiveResidual}
\end{align}
In order to formulate the multilinear relaxation of (\ref{apx-OneSidedObjectiveResidual}), consider the following two polytopes:
$\mathcal{P}\triangleq \{ \bx \in [0,1]^{\tilde{\mathcal{N}}}:\tilde{\mathbf{P}}\bx \leq \mathbf{r}_D\}$ and $ \mathcal{C}\triangleq \{ \bx \in [0,1]^{\tilde{\mathcal{N}}}:\tilde{\mathbf{P}}\bx \geq \mathbf{s}_D\}$.
Let $G:[0,1]^{\tilde{\mathcal{N}}}\rightarrow \mathbb{R}_+$ be the multilinear extension of $g$.
Thus, the continuous multilinear relaxation of (\ref{apx-OneSidedObjectiveResidual}) is:
\begin{align}
\max\left\{ f(E_1)+G(\bx):\bx\in [0,1]^{\tilde{\mathcal{N}}}, \bx \in \cP\cap \cC\right\}.\label{apx-ObjectiveContinuous}
\end{align}

Our algorithm performs randomized rounding of a fractional solution to the above relaxation (\ref{apx-ObjectiveContinuous}).
However, this is not enough to obtain our main result and an additional post-processing step is required in which additional elements are discarded.
Since covering constraints are present, one needs to perform the post-processing step in great care.
To this end we denote by $L_D$ the collection of large elements with respect to some critical packing constraint: $L_D\triangleq \{ \ell \in \tilde{\mathcal{N}}:\exists i\in Y_D \text{ s.t. } \mathbf{P}_{i,\ell}\geq \beta \mathbf{r}_D\}$ ($\beta \in (0,1)$ is a parameter to be chosen later).
Intuitively, we would like to discard elements in $L_D$ since choosing any one of those will incur a violation of a packing constraint.
We are now ready to present our rounding algorithm (Algorithm~\ref{apx-alg:OneSidedMainAlgorithm}).

\begin{algorithm*}[H]
	\caption{$(f,\cN,\bP,\bC)$}\label{apx-alg:OneSidedMainAlgorithm}
	Use Algorithm~\ref{apx-alg:OneSidedGuessesEnumeration} to obtain a list of guesses $\cL$. \\
	\ForEach{$D=(E_0, E_1,\bc') \in \cL$}{
          Use Theorem~\ref{thrm:ContGreedy} to compute an approximate solution $\bx^*$ to problem (\ref{apx-ObjectiveContinuous}).\label{apx-alg:OneSidedMainAlgorithm-ContGreedy}\\
		Scale down $\x^*$ to $\bar{\x} = \x^*/(1+\delta)$\\
		Let $R_D$ be such that for every $\ell \in \tilde{\mathcal{N}}$ independently: $\Pr \left[ \ell \in R_D\right] = \bar{\x}_{\ell}$.\\
		Let $R'_D = R_D \setminus L_D$.\\
		$S_D \leftarrow E_1 \cup R'_D$.
		}
	$S_{alg} \leftarrow \argmax \left\{f(S_D) : D\in \cL, \bP \cdot \one_{S_D} \leq \one_p, \bC \cdot \one_{S_D} \geq (1-\eps)\one _c \right\}$
\end{algorithm*}
We note that Line $6$ of Algorithm~\ref{apx-alg:OneSidedMainAlgorithm} is the post-processing step where all elements of $L_D$ are discarded.
Our analysis of Algorithm~\ref{apx-alg:OneSidedMainAlgorithm} shows that in an iteration a correct guess $D$ is examined, with a constant probability, $S_D$ satisfies the packing constraints, violates the covering constraint by only a fraction of $\varepsilon$, and $f(S_D)$ is sufficiently high.

The following lemma gives a lower bound on the value of the fractional solution $\bar{\mathbf{x}}$ computed by Algorithm~\ref{apx-alg:OneSidedMainAlgorithm}.

\begin{lemma}\label{apx-lem:os-quality-frac}
  If $D\in \mathcal{L}$ is correct then in the iteration of Algorithm~\ref{apx-alg:OneSidedMainAlgorithm} it is examined the resulting $\bar{\x}$ satisfies: $ G(\bar{\x})\geq (1-\nicefrac{1}{\e}-\delta)\opt-f(E_1)$.
\end{lemma}
\begin{proof}
  Let $D=(E_0,E_1,\bc')$ be a correct guess with respect to $O$ and
  let $S'=O\setminus E_1$ and $\cN' = \residual$. Because of
  Properties~\ref{apx-item:corr-os-e1} and~\ref{apx-item:corr-os-e0}, we have
  that $S'$ satisfies $\bP_{\cN'}\one_{S'}\leq\br_{D}$ and
  $\bC_{\cN'}\one_{S'}\geq\bs_{D}$. If $\x^*$ is as computed by
  the continuous greedy algorithm, then by Theorem~\ref{thrm:ContGreedy}, we have
\begin{displaymath}
  G(\x^*)\geq
  (1-\nicefrac{1}{e})(f(E_1)+G(\one_{S'}))-f(E_1)=(1-\nicefrac{1}{e})\opt-f(E_1)\,.
\end{displaymath}
To complete the proof observe that $G$ is concave along the
direction of the non-negative vector $\x^*$ (see \cite{CCPV11}) and thus
$G(\bar{\x})=G(\x^*/(1+\delta))\geq G(\x^*)/(1+\delta)$.
\end{proof}

\subsubsection{Main Lemma}
\label{apx-subsubsec:MainLemma}
Under the assumption that we are in the iteration in which the guessed $D=(E_0, E_1, \bc')$ is the \emph{correct guess}, in this section we prove our main lemma which directly implies Theorem~\ref{thrm:MainApproximation}.

To prove this we first write below some properties for the set $R_D$ outputted by running the independent rounding procedure on the vector $\bar{\x}$, which is the vector obtained by scaling the continuous greedy solution $\x^*$ by a factor $1+\delta$. 

Let $\cN' \triangleq \residual$ be the set of residual elements. Let $X_\ell$ be a random variable that indicates whether the element $\ell \in \cN'$ is in $R_D$ or not. Note that $X_\ell$ has been sampled independently according to $\bar{\x}$, \ie, $\Pr\left[X_\ell = 1 \right] = \bar{\x}_\ell = \nicefrac{\x^*_\ell}{(1+\delta)}$.

Since each element $\ell \in \cN'$ has been sampled independently according to $\bar{\x}$, hence $\Pr\left[\ell \in R_D \right] = \bar{\x}_\ell$. Using this and the properties of $\x^*$, it is easy to see that the following claim holds.

\begin{claim}
\label{apx-clm:exp-indep-round}
Following properties hold for the random set $R_D$.
\begin{enumerate}
\item\label{apx-item:exp-pack} $\E\left[ \sum_{\ell \in R_D} \bP_{i, \ell} \right] \leq \rdi/(1+\delta)$ for each $i \in [p]$
\item\label{apx-item:exp-cover} $\E\left[ \sum_{\ell \in R_D} \bC_{j, \ell} \right] \geq \sdj/(1+\delta)$ for each $j \in [c]$.
\end{enumerate}
\end{claim}

We know that all elements in the residual instance are small, if we ignore the critical constraints. Now we derive the probability for various types of constraints.

\begin{claim}\label{apx-clm:OneSidedProbNonTinyPackCover}
For any $i \in [p] \setminus Y_D$ and $j \in [c] \setminus Z_D$ that is not a critical constraint, 
\begin{enumerate}
    \item $\Pr\left[\sum_{\ell \in \sol} \bP_{i, \ell} > 1 \right] \leq \exp\left(-\frac{\delta^2}{3\alpha}\right)$.
    \item $\Pr\left[\sum_{\ell \in \sol} \bC_{j, \ell} < (1-2\delta) \bc'_j \right]  \leq \exp\left(-\frac{\delta^2}{3\alpha}\right)$.
\end{enumerate}
\end{claim}
\begin{proof}
Now for any packing constraint $i \in [p] \setminus Y_D$ and for each $\ell \in \cN'$, let us define the scaled matrix $\tP$ such that $\tP_{i, \ell} = \bP_{i, \ell}/(\alpha \rdi) \leq 1$. The last inequality follows from Defn.~\ref{apx-def:OneSidedCorrectGuess}.\ref{apx-item:oneSidedConsistentGuessNoBig}. Notice that $\E[\sum_{\ell \in R_D} \tP_{i, \ell}] \leq \nicefrac{\rdi}{(1+\delta) \alpha \rdi} = \nicefrac{1}{(1+\delta)\alpha}$ by Claim~\ref{apx-clm:exp-indep-round}.\ref{apx-item:exp-pack}. Now, applying a generalization of Chernoff bound (Theorem3.3, ~\cite{CL06-Chernoff}) with $X=\sum_{\ell \in \cN'} \tP_{i, \ell} X_\ell$, we obtain
\begin{align*}
\Pr\left[ \sum_{\ell \in \sol} \bP_{i, \ell} > 1\right] & = \Pr\left[\sum_{\ell \in R_D} \bP_{i, \ell} > \rdi \right] \\ 
			    &= \Pr\left[\sum_{\ell \in R_D} \tP_{i, \ell} > \nicefrac{1}{\alpha}\right]\\
                             & = \Pr\left[\sum_{\ell \in R_D} \tP_{i, \ell} > (1+\delta)\frac{1}{(1+\delta)\alpha}\right]\\
                             & \leq \exp\left(-\frac{1}{(1+\delta)\alpha} \cdot (\delta^2)/3\right)\\
                             & \leq \exp\left(-\frac{\delta^2}{3\alpha}\right)\,.
\end{align*}

Similarly, for each covering constraint $j \in [c] \setminus Z_D$, and each $\ell \in \cN'$, we define the scaled matrix $\tC$ such that $\tC_{j, \ell} = \bC_{j, \ell}/(\alpha \sdj) \leq 1$. $\E[\sum_{\ell \in R_D} \tC_{j, \ell}] \geq \nicefrac{\sdj}{(1+\delta) \alpha \sdj} = \nicefrac{1}{(1+\delta)\alpha}$ by Claim~\ref{apx-clm:exp-indep-round}.\ref{apx-item:exp-cover}. Again, applying Theorem3.3, ~\cite{CL06-Chernoff} with $X=\sum_{\ell \in \cN'} \tC_{j, \ell} X_\ell$, we obtain

\begin{align*}
\Pr\left[ \sum_{\ell \in \sol} \bC_{j, \ell} < (1 - 2 \delta) \bc'_j \right] & = \Pr\left[\sum_{\ell \in R_D} \bC_{j, \ell} <(1 - 2 \delta) \sdj \right] \\ 
			    &= \Pr\left[\sum_{\ell \in R_D} \tC_{j, \ell} < \nicefrac{(1 - 2 \delta)}{\alpha}\right]\\
                             & = \Pr\left[\sum_{\ell \in R_D} \tC_{j, \ell} < (1 - \delta)\frac{1}{(1+\delta)\alpha}\right]\\
                             & \leq \exp\left(-\frac{1}{(1+\delta)\alpha} \cdot (\delta^2)/2\right)\\
                             & \leq \exp\left(-\frac{\delta^2}{3\alpha}\right)\,.
\end{align*}

\end{proof}

For any $i \in Y_D$, let $S_D^i, L_D^i \subseteq \cN'$ be the set of small, large elements respectively, such that 

\[S_D^i \triangleq \left\{ \ell \in \cN': \bP_{i, \ell} < \beta \cdot \rdi \right\} \]
\[L_D^i \triangleq \cN'\setminus S_D^i \]

Note that for every $i \in Y_D$, $L_D^i \subseteq L_D$ and $L_D = \bigcup_{i\in Y_D} L_D^i$. Now using the same calculations as the previous claim, we get the following claim.

\begin{claim}\label{apx-clm:OneSidedProbTinyPack}
For any critical packing constraint $i\in Y_D$,
\begin{enumerate}
\item $\Pr\left[\sum_{\ell \in E_1 \cup (R_D \cap S_D^i)} \bP_{i, \ell} > 1 \right] \leq \exp\left(-\frac{\delta^2}{3\beta}\right)$.
\item $\Pr[\sum_{\ell \in R_D} \bP_{i, \ell} > 10b \cdot \rdi] \leq \frac{1}{(1+\delta) 10b}$, for any constant.
\end{enumerate}
\end{claim}
\begin{proof}
For any critical packing constraint $i \in Y_D$ and for each $\ell \in S_D^i$, we again define the scaled matrix $\tP$ such that $\tP_{i, \ell} = \bP_{i, \ell}/(\beta \rdi) \leq 1$. The last inequality follows from Defn.~\ref{apx-def:OneSidedCorrectGuess}.\ref{apx-item:oneSidedConsistentGuessNoBig}. Notice that $\E[\sum_{\ell \in R_D \cap S_D^i} \tP_{i, \ell}] \leq \E[\sum_{\ell \in R_D} \tP_{i, \ell}] \leq \nicefrac{\rdi}{(1+\delta) \beta \rdi} = \nicefrac{1}{(1+\delta)\beta}$ by Claim~\ref{apx-clm:exp-indep-round}.\ref{apx-item:exp-pack}. Applying Chernoff bound with $X=\sum_{\ell \in S_D^i} \tP_{i, \ell} X_\ell$, we obtain
\begin{align*}
\Pr\left[ \sum_{\ell \in E_0 \cup (R_D \cap S_D^i)} \bP_{i, \ell} > 1\right] & = \Pr\left[\sum_{\ell \in R_D \cap S_D^i} \bP_{i, \ell} > \rdi \right] \\ 
			    &= \Pr\left[\sum_{\ell \in R_D \cap S_D^i} \tP_{i, \ell} > \nicefrac{1}{\beta}\right]\\
                             & = \Pr\left[\sum_{\ell \in R_D \cap R_D} \tP_{i, \ell} > (1+\delta)\frac{1}{(1+\delta)\beta}\right]\\
                             & \leq \exp\left(-\frac{1}{(1+\delta)\beta} \cdot (\delta^2)/3\right)\\
                             & \leq \exp\left(-\frac{\delta^2}{3\beta}\right)\,.
\end{align*}

For the second part, using Markov's inequality and Claim~\ref{apx-clm:exp-indep-round}.\ref{apx-item:exp-pack},  we get
\begin{displaymath}
\Pr\left[\sum_{\ell \in R_D} \bP_{i, \ell} > 10b \rdi\right] \stackrel{\text{Markov Ineq.}}\leq E\left[ \sum_{\ell \in R_D} \bP_{i, \ell} \right]/(10b \rdi) \leq \frac{1}{ 10b (1+\delta)} \,.
\end{displaymath}

\end{proof}

Since there are at most $b$ critical constraints with probability at most $1/10$ there is some critical constraint that is violated by more than a factor of $10b$.

For any covering constraint $j \in Z_D$, the fact that $\sdj \leq \delta \bc'_j$ gives the following claim.

\begin{claim}\label{apx-clm:OneSidedTinyCover}
For any covering constraint $j \in Z_D$, $\sum_{\ell \in E_1} \bC_{j, \ell} \geq (1-\delta) \bc'_j$
\end{claim}

By Lemma~\ref{apx-lem:os-quality-frac}, we have that $G(\bar{\x})\geq(1-\nicefrac{1}{\e}-\delta)\opt-f(E_1)$ where $\bar{\x}$ is the fractional solution computed in line~\ref{apx-alg:OneSidedMainAlgorithm-ContGreedy} of Algorithm~\ref{apx-alg:OneSidedMainAlgorithm} in the iteration when the algorithm enumerates $D$.

\begin{claim}\label{apx-clm:OneSidedProbObj}
$\Pr[f(\sol)<(1-\nicefrac{1}{\e}-2\delta)\opt] \leq  \exp\left(-\frac{\gamma\delta^2}{2}\right)$.
\end{claim}
\begin{proof}
We have by Theorem 1.3, Chekuri et al.~\cite{CVZ10-arxiv} and using $g(\ell)=f_{E_1}(\ell)\leq\gamma^{-1}\opt$ for all $\ell\in\cN'$
\begin{align*}
    \Pr\left[f(\sol)<\left(1-\frac{1}{\e}-2\delta\right)\opt\right] & = \Pr\left[g(R_D)<\left(1-\frac{1}{\e}-2\delta\right)\opt-f(E_1)\right]\\
     & \leq \Pr\left[g(R_D)<G(\bar{\x})-\delta\opt\right]\\
     & = \Pr\left[\frac{g(R_D)}{\gamma^{-1}\opt}<\frac{G(\bar{\x})}{\gamma^{-1}\opt}-\gamma\delta\right]\\
     & = \Pr\left[\frac{g(R_D)}{\gamma^{-1}\opt}<\left(1-\frac{\delta\opt}{G(\bar{\x})}\right)\cdot \frac{G(\bar{\x})}{\gamma^{-1}\opt}\right]\\
     & \leq \exp\left(\frac{-G(\bar{\x})}{2\gamma^{-1}\opt}\cdot \left(\frac{\delta\opt}{G(\bar{\x})}\right)^2\right)\\
     & \leq  \exp\left(-\frac{\gamma\delta^2}{2}\right)\,.
\end{align*}

\end{proof}


Now we fix the parameter $\alpha = \delta^3$, $\beta = \nicefrac{\delta^2}{3b}$ and $\gamma = \nicefrac{1}{\delta^3}$ and get the following claim.
\begin{claim}\label{apx-clm:OneSidedProperties}
For any positive $\delta \leq \nicefrac{1}{15b^3}$, with probability at least $\nicefrac{1}{2}$ we get the following properties for the intermediate solution $\sol$.
\begin{enumerate}
\item\label{apx-item:pack-non-tiny} For all $i \in [p] \setminus Y_D$, we have $\sum_{\ell \in \sol} \bP_{i, \ell} \leq 1$. 
\item\label{apx-item:pack-tiny}  For all $i \in Y_D$, we have $\sum_{\ell \in R_D} \bP_{i, \ell} \leq 10b \rdi$ and $\sum_{\ell \in E_1 \cup (R_D \cap S_D^i)} \bP_{i, \ell} \leq 1$. 
\item\label{apx-item:cover-non-tiny}  For all $j \in [c] \setminus Z_D$, we have $\sum_{\ell \in \sol} \bC_{j, \ell} \geq (1 - 2 \delta) \bc'_j$.
\item\label{apx-item:cover-tiny}  For all $j \in Z_D$, we have $\sum_{\ell \in E_1} \bC_{j, \ell} \geq (1 - \delta) \bc'_j$.
\item\label{apx-item:objective}  $f(\sol) \geq (1-\nicefrac{1}{\e}-2\delta)\opt$.
\end{enumerate}
\end{claim}
\begin{proof}
Using the union bound on probability bounds from Claims~\ref{apx-clm:OneSidedProbNonTinyPackCover}--~\ref{apx-clm:OneSidedProbObj}, the probability that any of the Properties~\ref{apx-item:pack-non-tiny}--~\ref{apx-item:objective} does not hold, for any $1 \leq b$, is at most
\begin{displaymath}
  2b\cdot e^{-\nicefrac{\delta^2}{(3\alpha)}}+b\cdot e^{-\nicefrac{\delta^2}{(3\beta)}}+b\cdot 1/(10b)+e^{\nicefrac{-\gamma\delta^2}{2}}= 2b\cdot e^{-\nicefrac{1}{(3\delta)}}+b\cdot e^{-b}+e^{-1/(2\delta)}+\frac{1}{10} \leq \frac{1}{2}\,.
\end{displaymath}
\end{proof}

Recall that $R_D'$ arises from $R_D$ by removing $L_D$ from $R_D$ which contain elements from $\cN'$ such that $\bP_{i, \ell} > \beta \rdi$ for some critical constraint $i \in Y_D$. The condition from item~\ref{apx-item:pack-tiny} imply that for every such constraint $i \in Y_D$,  $|L_D^i| \leq 10b/\beta$, thus overall, $|L_D| \leq 10b^2/\beta$.

It is easy to see (by Claim~\ref{apx-clm:OneSidedProperties}) that after this step all packing constraints are satisfied. For any critical covering constraint $j \in Z_D$, the set $E_1$ itself has cover value $\geq (1-\delta)  \bc'_j$ (Claim~\ref{apx-clm:OneSidedProperties}.\ref{apx-item:cover-tiny}). For each non-critical covering constraint $j \in [c] \setminus Z_D$, we get the following bound on the loss in covering value.

\begin{claim} \label{apx-clm:OneSidedFinalCoverNonTiny} For each non-critical covering constraint $j \in [c] \setminus Z_D$, \[\sum_{\ell \in E_1 \cup R'_D} \bC_{j, \ell} \geq (1 - (30b^3 + 2) \delta) \cdot  \bc'_j\]
\end{claim}
\begin{proof}
For any non-critical covering constraint $j \in [c] \setminus Z_D$, since the cover value for each element $\ell \in \cN'$ is at most $\alpha \cdot \sdj $, the loss in cover value after removing $L_D$ is at most $(\nicefrac{10b^2}{\beta}) \cdot \alpha \cdot \sdj \leq (\nicefrac{10b^2\alpha}{\beta}) \cdot  \bc'_j =  (30b^3\delta) \cdot  \bc'_j$. Combining this bound with the bound in Claim~\ref{apx-clm:OneSidedProperties}.\ref{apx-item:cover-non-tiny}, we get  our claim. 
\end{proof}

Finally we get the following bound on the objective function value for the solution $E_1 \cup R'_D$.
\begin{claim} \label{apx-clm:OneSidedFinalObj}
$f(E_1 \cup R'_D) \geq (1 - (30b^3 + 2) \delta) \opt$.
\end{claim}
\begin{proof}
By Observation~$1$ in proof for Lemma~\ref{apx-lem:os-correct-guess}, for any residual element $\ell \in \cN'$ we have $f(\ell)_{E_1} \leq \gamma^{-1} f(E_1)$. Since $E_1 \subseteq O$ and by the fact that $f$ is a monotone submodular functions, we get $f_{E_1 \cup R'_D}(\ell) \leq \gamma^{-1} \opt$. Let $\ell_1, \ell_2 \dots \ell_m$ be any arbitrary ordering of elements of $L_D \cap R_D$ and let $L_q \triangleq \{\ell_1, \ell_2 \dots \ell_q\}$. Notice that $m \leq |L_D| \leq \nicefrac{10b^2}{\beta}$. Overall, we get
\begin{align*}
f(E_1 \cup R_D) - f(E_1 \cup R'_D) 	& \leq \sum_{q=1}^m f(E_1 \cup R'_D \cup L_q) - f(E_1 \cup R'_D \cup L_{q-1}) \\
			  				& \leq \sum_{q=1}^m f(E_1 \cup R'_D \cup \ell_q) - f(E_1 \cup R'_D \cup \ell_{q-1}) \\
			  				& \leq \sum_{q=1}^m \gamma^{-1} \opt \\
			  				& \leq (\nicefrac{10b^3}{\beta\gamma}) \opt = (30b^3 \delta) \opt
\end{align*}
Combining it with the bound in Claim~\ref{apx-clm:OneSidedProperties}.\ref{apx-item:objective} we get the claim.
\end{proof}

Overall, we get our Main Lemma.

\begin{lemma}
\label{apx-lem:MainApproximation}
For any fixed $0 < \eps$, if we choose $\alpha = \delta^3$, $\beta = \nicefrac{\delta^2}{3b}$, $\gamma = \nicefrac{1}{\delta^3}$ and $\delta <\min\left\{\nicefrac{1}{(15b)}, \nicefrac{\eps}{(30b^3 + 2)} \right\}$, with probability at least $\nicefrac{1}{2}$, the Algorithm~\ref{apx-alg:OneSidedMainAlgorithm} outputs a solution $S \subseteq \cN$ in time $n^{\text{poly}(1/\varepsilon)}$, such that 
\begin{enumerate}
\item \label{apx-item:alg-pack} $\bP \one_{E_1 \cup R_D'} \leq \one_p$,
\item \label{apx-item:alg-cover}  $\bC \one_{E_1 \cup R_D'} \geq (1 - \eps) \one_c$.
\item \label{apx-item:alg-obj}  $f(E_1 \cup R_D' ) \geq (1-\nicefrac{1}{\e}- \eps) \opt$.
\end{enumerate}
\end{lemma}

The above lemma suffices to prove Theorem~\ref{thrm:MainApproximation}, as it immediately implies it.



\section{Greedy Dynamic Programming} \label{App:GreedyDP}
\subsection{Vanilla Greedy DP}\label{subsec:vanilla-greedy-dp}

To highlight the core idea of our approach, we first present a vanilla version of the greedy dynamic programming approach applied to \PCSM that gives a constant-factor approximation and satisfies the packing constraints, but violates the covering constraints by a factor of $2$ and works in pseudo-polynomial time. See Lemma~\ref{lem:greedy-dp-vanilla}

\subsubsection{Algorithm}\label{subsubsec:forbiddenAlgorithm}

For simplicity of presentation, we assume in the current discussion relating to pseudo-polynomial time algorithms that $\bC \in \N_+^{c \times n}$ and $\bP \in \N_+^{p\times n}$.
Let $\p \in \N_+^p$ and $ \bc \in \N_+^c$ be the packing and covering requirements, respectively.
A solution $S\subseteq \cN$ is feasible if and only if $\bC \one_S \geq \bc$ and $\bP \one_S \leq \p$.
We also use the following notations: $c_{\max}=\|\bc\|_{\infty}$, $p_{\max}=\|\p\|_{\infty}$, and $[s]_0=\{0,\dots,s\}$ for every integer $s$.

We define our dynamic programming as follows:
for every $\ind\in[n]_0$, $\bc'\in[n \cdot c_{\max}]_0^c$, and $\p'\in[p_{\max}]_0^p$ a table entry $T[q,\bc',\p']$ is defined and it stores an approximate solution $S$ of cardinality $\ind$ with $\bC \one_S = \bc'$ and $\bP \one_S = \p'$. We introduce a dummy solution $\bot$ for denoting undefined table entries, and initialize the entire table with $\bot$. We work with the convention that $f(\bot)=-\infty$ and that $S\cup\bot=\bot$ for every set $S\subseteq\cN$. For brevity, we define $\bP_\ell \triangleq \bP \one_{\ell}$ and $\bC_\ell \triangleq \bC \one_{\ell}$ for any element $\ell \in \cN$.

For the base case, we set $T[0,\zero_c,\zero_p] \gets \varnothing$.
For populating $T[q,\bc',\p']$ when $q>0$, we examine every set of the form $T[q-1,\bc'-\bC_\ell,\p'-\bP_\ell] \cup \{\ell\}$, where $\ell$ satisfies $\ell\in\cN\setminus T[q-1,\bc'-\bC_\ell,\p'-\bP_\ell]$, $\bc'-\bC_\ell \geq 0$, and $\p'-\bP_\ell \geq 0$.
Out of all these sets, we assign the most valuable one to $T[q,\bc',\p']$.
Note that this operation stores a {\em greedy approximate} solution in the table entry $T[q,\bc',\p']$.
The output of our algorithm is the best of the solutions $T[q,\bc',\p']$, for $1 \leq q \leq n$, $\bc' \geq \bc/2$ and $\p' \leq \p$. See Algorithm~\ref{alg:vanilla-greedydp} for pseudo code.
\begin{algorithm*}[H] \label{alg:vanilla-greedydp}
  \caption{Vanilla Greedy Dynamic Program}
  create a table $T\colon [n]_0\times[n \cdot c_{\max}]_0^c\times[p_{\max}]_0^p\rightarrow 2^{\cN}$ initialized with entries $\bot$\\
  $T[0,\zero_c,\zero_p] \gets \varnothing$\\
  \For{$\ind=0$ \KwTo $n$}{ \ForEach{$\bc'\in[n \cdot c_{\max}]_0^c$ and
      $\p'\in[p_{\max}]_0^p$}{
      \ForEach{$\ell\in\cN\setminus T[\ind,\bc',\p']$}{
        $\bc''\gets\bc'+\bC_{\ell}$, $\p''\gets\p'+\bP_{\ell}$\\
        $T[\ind+1,\bc'',\p'']\gets\arg\max\{f(T[\ind+1,\bc'',\p'']),f(T[\ind,\bc',\p']\cup\{\ell\})\}$\\
      }
    } }
  Output     $\displaystyle\argmax_{q,\bc' \geq \bc/2, \p' \leq \p} f(T[q,\bc',\p'])$.
  \end{algorithm*} 

\subsubsection{A Warmup Analysis}

Let $O$ be an optimal set solution. Let us consider an arbitrary permutation of $O$, say $\{o^1,o^2,\dots, o^k\}$. Let $O^i = \{o^1, \dots, o^i\}$ be the set of the first $i$ elements in this permutation. Let $O^0=\emptyset$. We introduce the function $g\colon O \rightarrow \R_+$ for denoting the marginal value of the elements in $O$. More precisely, let $g(o^i)=f_{O^{i-1}}(o^i)$. Note that $f(O)= \sum_{\ell \in O} g(\ell)$. Let for any subset $S \subseteq O$, $g(S) = \sum_{\ell \in S}g(\ell)$.

\begin{lemma}
\label{apx-lem:subset}
For every subset $S\subseteq O$, we have $f(S) \geq g(S)$.
\end{lemma}
\begin{proof}
Let $o^{i_1},\dots, o^{i_p}$ with $i_1<i_2<\dots<i_p$ be the elements of $S$ in the order as they appear in $O$. Let $S_j$ be the set of the first $j$ elements in $S$, that is, $S^j=S \cap O^{i_j}$ for $j=1,\dots,p$, and let $S^0=\varnothing$. By submodularity of $f$ and $S^{j-1}\subseteq O^{i_j-1}$ we have that
\begin{displaymath}
f(S) = \sum_{j=1}^p f(S^j)-f(S^{j-1}) \geq \sum_{j=1}^p f(O^{i_j})-f(O^{i_j-1}) = \sum_{j=1}^p g(o^{i_j}) = g(S).
\end{displaymath}
\end{proof}

\begin{lemma}
\label{apx-lem:submodular-const}
There exists a table entry $T[q,\bc_q, \p_q]$ for some $0\leq q \leq k$, such that $f(T[q,\bc_q, \p_q]) \geq \frac14 f(O)$ and such that there exists a $q$-subset $O_q=\{o_1,\dots , o_q \} \subseteq O$ with packing value equal to $\p_q \leq \p$ and covering value equal to $\bc_q \geq \bc / 2$.
\end{lemma}
\begin{proof}
Let us assume that the statement of the lemma is not true. We prove below by induction on $\ind$ that under this assumption the following even stronger claim holds thereby leading to a contradiction.
\begin{claim}\label{apx-claim:sequence-const}
For every $\ind \in[k]_0$ there is an $\ind$-subset $O_\ind=\{o_1,\dots,o_\ind\} \subseteq O$ with packing value $\p_\ind \leq \p$ and covering value $\bc_\ind$, such that one of the following holds
\begin{enumerate}[(i)]
\item \label{apx-item:1} $f(T[\ind, \bc_\ind, \p_\ind])\geq \frac14 f(O)$,
\item \label{apx-item:2} $f(T[\ind, \bc_\ind, \p_\ind])\geq \frac12 g(O_\ind)$.
\end{enumerate}
\end{claim}
Note that if this claim is true then for $q=k$ we directly get a contradiction.

For the base case $q = 0$, Property~\ref{apx-item:2} is trivially true for $O_\ind=\varnothing$, $\bc_\ind = \mathbf{0}$ and $\p_\ind=\mathbf{0}$.

For the inductive step let $\ind \geq 1$ and assume that the claim already holds for $\ind-1$. To this end, let $O_{\ind-1}$, $\bc_{\ind-1}$ and $\p_{\ind-1}$ be as in this claim. Let $S_{\ind-1}=T[\ind-1,\bc_{\ind-1}, \p_{\ind-1}]$.

Now, we distinguish the two cases where $S_{\ind-1}$ satisfies Property~\ref{apx-item:1} or Property~\ref{apx-item:2}, respectively.
First, assume that $f(S_{\ind-1}) \geq \frac14 f(O)$. Let  $\widetilde{O}= O \setminus (O_{\ind-1} \cup S_{\ind-1})$. If $\widetilde{O} \neq \varnothing$ then pick $o_{\ind} \in \widetilde{O}$ and let $O_\ind = O_{\ind-1} \cup \{o_{\ind}\}$, $\bc_{\ind} = \bC \one_{O_\ind}$, and $\p_{\ind} = \bP\one_{O_\ind}$. Moreover, $f(T[\ind, \bc_{\ind}, \p_{\ind}]) \geq f(S_{\ind-1} \cup \{o_{\ind}\}) \geq f(S_{\ind-1}) \geq \frac14 f(O)$ completing the inductive step. On the other hand, if $\widetilde{O} = \varnothing$ then $ O\setminus O_{\ind-1}\subseteq S_{\ind-1}$. Hence $\bc_{\ind-1} \geq \bC\one_{O\setminus O_{\ind-1}} $. Combining this with $\bC\one_{S_{\ind-1}}=\bc_{\ind-1}=\bC\one_{O_{\ind-1}}$ and $\bC\one_{O}\geq\bc$ we get $\bC\one_{S_{\ind-1}}\geq \bc/2$. This contradicts our assumption that the statement of the lemma is not true.

In the case when $S_{\ind-1}$ satisfies Property~\ref{apx-item:2}, we have $f(O_{\ind-1})\geq \frac12 g(O_{\ind-1})$. We can also assume w.\,l.\,o.\,g.\ that Property~\ref{apx-item:1} does \emph{not} hold for $\ind-1$. Now we distinguish two sub-cases. In the first sub-case there exists some $o_\ind \in O \setminus O_{\ind-1}$ such that $f_{S_{\ind-1}}(o_\ind) > \frac12 g(o_\ind)$.
Note that $o_\ind \notin S_{\ind-1}$, since otherwise the left hand side of the inequality would be zero while $g(o_\ind)$ is non-negative, which would be a contradiction. Now, let $O_{\ind} = O_{\ind - 1} \cup \{o_\ind\}$, $\bc_\ind = \bc_{\ind-1} + \bC_{o_\ind}$ and $\p_\ind=\p_{\ind-1}+\bP_{o_\ind}$. Hence the DP could potentially add this element $o_\ind$ to the entry $T[\ind-1,\bc_{\ind-1}, \p_{\ind-1}]$ to get $T[\ind,\bc_\ind, \p_\ind]$. Now verify that
\begin{align*}
f(T[\ind,\bc_\ind, \p_\ind]) & \geq f( S_{\ind-1}\cup\{o_\ind\}) \geq f( S_{\ind-1})+\frac12 g(o_\ind)\\
                     & \geq \frac12 g(O_{\ind-1})+\frac12 g(o_{\ind})\\
                     & = \frac12 \sum_{i=1}^{\ind} g(o_i)\\
                     & = \frac12 g(O_\ind)\,.
\end{align*}

In the second sub-case, for all $a \in O \setminus O_{\ind-1}$, we have $f_{S_{\ind-1}}(a) \leq \frac12 g(a)$. We derive below a contradiction to our assumption that the lemma is not true.

Let $O\setminus O_{\ind-1} = \{a_1,\dots,a_m\}$. By submodularity of~$f$ we have $f(S_{\ind-1} \cup \{a_1, a_2, \dots a_j\}) - f(S_{\ind-1} \cup \{a_1, a_2, \dots a_{j-1}\}) \leq f_{S_{\ind-1}}(a_j) \leq \frac12 g(a_j)$ for all $j=1,2, \dots m$. Adding all these inequalities, we get
\begin{align*}
f(S_{\ind-1} \cup (O \setminus O_{\ind-1}))-f(S_{\ind-1}) &= \sum_{j=1}^m [f(S_{\ind-1}\cup \{a_1,\dots,a_j\}) \\
& \qquad\qquad -f(S_{\ind-1}\cup\{a_1,\dots,a_{j-1}\})]\\
& \leq \frac12\sum_{j=1}^m g(a_j)\\
& = \frac12 g(O \setminus O_{\ind-1})\,.
\end{align*}

Rearranging, we get
\begin{align*}
f(S_{\ind-1}) & \geq f(S_{\ind-1} \cup (O \setminus O_{\ind-1})) - \frac12 g(O\setminus O_{\ind-1})\\
 & \geq f(O \setminus O_{\ind-1}) - \frac12 g(O \setminus O_{\ind-1})\\
 & \stackrel{\text{Lem.~\ref{apx-lem:subset}}}{\geq} \frac12 g(O \setminus O_{\ind-1})\,.
\end{align*}

Adding this to $f(S_{\ind-1}) \geq \frac12 g(O_{\ind-1})$ , we get $4 f(S_{\ind-1}) \geq g(O)=f(O)$. This contradicts the assumption that the claim of the lemma is not true.
\end{proof}

Now, the following lemma follows directly from Lemma~\ref{apx-lem:submodular-const}.
 \begin{lemma}
 \label{apx-lem:submodular-const-0.25}
  There is an algorithm for \PCSM that outputs in pseudo-polynomial time \\ $O(n^2c_{\max}p_{\max})$ a $0.25$-approximate solution with covering value at least $\bc/2$ and packing value $\p$.
 \end{lemma}

\subsubsection{Factor-Revealing LP}
\label{apx-subsec:factor-rev-lp}

In this section, we develop a factor-revealing LP for an improved analysis of the approximation ratio of the above-described greedy DP. Note that in the previous analysis we looked at only \emph{one phase} in which we account for our gain based on whether or not the current element gives us a marginal value of more than $1/2$ times the marginal value that it contributes to the optimal solution. But in reality for the elements added in the beginning, we gain almost the same value as in the optimal solution. The ratio of gain decreases until we gain zero value when adding any element from the optimal solution that is still not in our approximate solution. In this section, we analyze our DP using a factor-revealing LP by discretization of the marginal value ratios to $1-\frac{1}{m}, 1-\frac{2}{m}, \dots, \frac{1}{m}$ and get lower bounds for the partial solutions at the end of each phase. Then we embed these inequalities into a factor-revealing LP and show that for the worst distribution of the optimal solution among the phases, the approximate solution is at most a factor $1/\e$ away from the optimum solution.

The $i$-th phase corresponds to the phase in which we will gain at least a $(1-i/m)$-fraction of marginal value if we add an element from the optimal solution during that iteration. We keep on adding these elements to $O_i$, until no such element remains. $A_i$ corresponds to the solution at the end of the $i$-th phase. $A_m$ is the solution at the end of this procedure. Now we estimate the value of the approximate solution $A_m$ as compared to the optimal solution.

For the purpose of analysis, by scaling, we assume that $f(O) = \sum_{o \in O} g(o)=1$. The following lemma is the basis for the \emph{factor-revealing} LP below.

\begin{lemma}
  \label{apx-lem:submodular-const-multi-phase}
  Let $m\geq 1$ be an integral parameter.  We can pick for each $i \in [m]$ a set $O_i = \{o^i_1, o^i_2, \dots, o^i_{\ind_i}\} \subseteq O$ (possibly empty) such that the following holds. For $i\neq j$, we have that $O_i \cap O_j = \varnothing$. Let $L_i = \sum_{j=1}^i \ind_j$, $Q_i = \cup_{j=1}^iO_j$,  $\bc_i=\bC \one_{Q_i}$, $\p_i=\bP \one_{Q_i}$ and let $A_i := T[L_i, \bc_i, \p_i]$ be the corresponding DP cell. Then $\bC \one_{A_m} \geq \bc/2$ and the following inequalities hold.
\begin{enumerate}
\item \label{apx-ineq:t0} $f(A_0) =  g(O_0)$ where $A_0 = O_0 = \varnothing$,
\item \label{apx-ineq:t1} $f(A_i) \geq f(A_{i-1}) +  (1-i/m) g(O_i) \,\, \forall i \in [m]$ and
\item  \label{apx-ineq:t2} $f(A_i) \geq \frac{i}{m} \left(1- \sum_{j \leq i} g(O_j)\right) \,\, \forall i \in \{0\} \cup [m]$.
\end{enumerate}
\end{lemma}
\begin{proof}
We prove this by using induction on $i$. For $i=0$, both inequalities \ref{apx-ineq:t0} and \ref{apx-ineq:t2} are trivially true.

For inductive step let $i \geq 1$ and  assume that both inequalities \ref{apx-ineq:t1} and \ref{apx-ineq:t2} are true for every $j < i$. We start by defining $O_i :=  \varnothing$, $\ind_i :=  0$ and $A_i := A_{i-1}$. Note that inequality \ref{apx-ineq:t1} is true for this choice of $O_i$ and $A_i$. Now, if there is an $o \in O \setminus (\cup_{j=1}^{i} O_j)$ such that $f_{A_i} (o) > (1-i/m) g(o)$, then it implies that $f(T[L_i + 1, \bC \one_{A_i\cup \{ o \}}, \bP\one_{A_i\cup \{ o \}}]) \geq f(A_i \cup \{ o \}) > f(A_i) + (1-i/m) g(o) \geq f(A_{i-1}) + (1-i/m)g(O_i) + (1-i/m) g(o) =  f(A_{i-1}) + (1-i/m)g(O_i\cup\{ o \})$. Note that $o \notin A_i$, since otherwise the left hand side of the inequality would be zero while $g(o)$ is a non-negative quantity. Hence, we can extend $O_i \rightarrow O_i \cup \{ o \}$, $\ind_i \rightarrow \ind_i + 1$, $Q_i'  \triangleq \cup_{j=1}^{i} O_{j} \cup \{ o \}$ and $A_i \rightarrow f(T[L_i + 1, \bC \one_{Q_i'}, \bP \one_{Q_i'}])$. Hence, inequality \ref{apx-ineq:t1} remains true after performing this operation on $O_i$ and $A_i$. We keep on doing this until for all $o \in O \setminus (\cup_{j=1}^{i} O_j)$,  $f_{A_i} (o) \leq (1-i/m) g(o)$.

Now we prove the inequality \ref{apx-ineq:t2}. Let $O \setminus (\cup_{j=1}^i O_j) = \{ a_1, a_2, \dots, a_q\}$. By submodularity we have that $f(A_i \cup \{a_1, \dots, a_j\})-f(A_i \cup \{ a_1, \dots, a_{j-1}\})\leq f_{A_i} (a_j) \leq (1-i/m) g(a_j)$ for each $j=1, \dots, q$. Adding up these inequalities for all $j=1, \dots, q$ we get that
\begin{align*}
f(A_i \cup O \setminus \cup_{j=1}^i O_j)-f(A_i) &= \sum_{j=1}^q [f(A_i \cup \{a_1,\dots,a_j\})\\
& \qquad\qquad - f(A_i \cup \{ a_1, \dots, a_{j-1}\})]\\
& \leq \left(1-\frac{i}{m}\right)  \sum_{j=1}^q g(a_j)\\
& = \left(1-\frac{i}{m}\right) g(O \setminus \cup_{j=1}^i O_j)\,.
\end{align*}
Rearranging, we get
\begin{align*}
f(A_i) & \geq f(A_i \cup (O \setminus \cup_{j=1}^i O_j)) - \left(1-\frac{i}{m}\right) g(O \setminus \cup_{j=1}^i O_j)\\
             & \geq f(O \setminus \cup_{j=1}^i O_j) - \left(1-\frac{i}{m}\right) g(O \setminus \cup_{j=1}^i O_j)\\
             & \stackrel{\text{Lem.~\ref{apx-lem:subset}}}{\geq} \frac{i}{m} g(O \setminus \cup_{j=1}^i O_j) \\
             & = \frac{i}{m} (1 - g(\cup_{j=1}^i O_j))\,.
\end{align*}
Hence inequality \ref{apx-ineq:t2} is also true for this $A_i$ and $O_i$ and hence the induction follows.

To ensure $\bC\one_{A_m}\geq\bc/2$ we use a modified construction of the last phase. In particular, for constructing $A_m$ we again start with $A_m = A_{m-1}$ and $O_m = \varnothing$. In the iterative process we keep on adding elements $o \in O\setminus (A_m \cup(\cup_{i=1}^m O_i))$ to $A_m$ and $O_m$ until $O\setminus (A_m \cup(\cup_{i=1}^m O_i)) = \varnothing$. For the set $A_m$ thereby constructed, both inequalities are trivially true. The process also implies that $\bC \one_{A_m} \geq \bC \one_{O\setminus (\cup_{i=1}^m O_i)}$. Combining this with the fact that $\bC \one_{A_m} = \bC \one_{\cup_{i=1}^m O_i}$ we get $\bC \one_{A_m} \geq \bc/2$, which proves the lemma.
\end{proof}

Below we describe a \emph{factor-revealing} LP that captures the above-described multi-phase analysis for the greedy DP algorithm. The idea is to introduce variables for the quantities in the inequalities in the previous lemma and determining the minimum ratio that can be guaranteed by these inequalities.


\begin{align}
\min \,\, & \textstyle a_m 							      && \mbox{s.t.} &&  \hspace{16em}  \mbox{(LP)} \nonumber \\
           & \textstyle a_1 \geq \left(1-\frac{1}{m}\right) o_1;           	      &&~\label{apx-lp:c0} \\
           & \textstyle a_i \geq a_{i-1} +  \left(1-\frac{i}{m}\right) o_i            && \forall i \in [m] \setminus \{1\};~\label{apx-lp:c1} \\
           & \textstyle  a_i \geq \frac{i}{m} \left(1- \sum_{j \leq i} o_j\right)   && \forall i \in [m];~\label{apx-lp:c2} \\
           & \textstyle a_i \geq 0, \,\, o_i \geq 0 			     && \forall i \in [m].
\end{align}

The variable $o_i$ corresponds to the marginal value $g(O_i)$ for the set $O_i$ in our analysis. Variables $a_i$ correspond to the quantities $f(A_i)$ for the approximate solution $A_i$ for each phase $i=1, 2, \dots, m$. We add all the inequalities we proved in Lemma~\ref{apx-lem:submodular-const-multi-phase} as the constraints for this LP. Note that since $f(O)=1$, the minimum possible value of $a_m$ will correspond to a lower bound on the approximation ratio of our algorithm.

The following is the dual for the above LP.
\begin{align}
\max \,\,             & \textstyle \sum_{i = 1}^{m} \frac{i}{m} y_i                                    && \mbox{s.t.} && \hspace{13em} \mbox{(DP)}  \nonumber \\
		       & \textstyle x_i + y_i - x_{i+1} \leq 0                                               && \forall i \in [m-1];~\label{apx-du:ai} \\
                        & \textstyle x_m + y_m \leq 1; 				   		         &&~\label{apx-du:an} \\
                        & \textstyle \sum_{j \geq i} \frac{j}{m}y_j - (1-\frac{i}{m}) x_i \leq 0	 && \forall i \in [m];~\label{apx-du:oi} \\
                        & \textstyle x_i \geq 0, \,\, y_i \geq 0 			    			 && \forall i \in [m].
\end{align}

\paragraph{Upper and Lower Bounds for the Factor-Revealing LP}
\label{apx-subsec:LP-bounds}
We can analytically prove that the of optimum value of the LP converges to $1/\e$. For this we show that the optimal value of the LP is at least $\left(1 - \frac{1}{m}\right)^m$. To this end, we will show a feasible dual solution with the same value.  In particular, we let $x_i = \left(1-\frac{1}{m}\right)^{m-i}$ for $i \in [m]$, $y_i = \frac{1}{m}\left(1-\frac{1}{m}\right)^{m-i-1}$ for $i \in [m-1]$ and $y_m = 0$ and show the following lemma.
\begin{lemma}
\label{apx-lem:submodular-const-dual}
The above solution is a feasible solution for (DP). Moreover, the value of this solution is $(1-1/m)^m$.
\end{lemma}

\begin{proof}
The constraint $x_m + y_m \leq 1$ is trivially satisfied since $x_m=1$ and $y_m=0$.

Now let us consider a constraint of the form $x_i + y_i - x_{i+1} \leq 0$. For these constraints the left hand side is equal to
\[
\left(1-\frac{1}{m}\right)^{m-i} + \frac{1}{m}\left(1-\frac{1}{m}\right)^{m-i-1} - \left(1-\frac{1}{m}\right)^{m-i-1}  = \left(1-\frac{1}{m}\right)^{m-i} - \left(1-\frac{1}{m}\right)^{m-i-1}\left(1 - \frac{1}{m}\right) = 0
\]
hence they are satisfied.

For the second kind of constraints $\sum_{j \geq i} \frac{j}{m}y_j - \left(1-\frac{i}{m}\right) x_i \leq 0$, the left hand side is equal to
\[
\sum_{j = i}^{m-1} \frac{j}{m} \frac{1}{m}\left(1-\frac{1}{m}\right)^{m-j-1} - \left(1-\frac{i}{m}\right) \left(1-\frac{1}{m}\right)^{m-i}
\]
First, we simplify the first term
$T := \sum_{j=i}^{m-1} \frac{j}{m} \frac{1}{m}(1-\frac{1}{m})^{m-j-1} = \frac{1}{m^2}\sum_{j = i}^{m-1} j \left(1-\frac{1}{m}\right)^{m-j-1}$,
which is an arithmetic-geometric progression. Now we multiply both sides by $(1-1/m)$ to get,
$(1-\frac{1}{m})T = \frac{1}{m^2}\sum_{j = i}^{m-1} j (1-\frac{1}{m})^{m-j}$. Subtracting the second equality from first we get,
\begin{align*}
\frac{T}{m} &= \frac{1}{m^2} \left( (m-1) -  \sum_{j=1}^{m-i-1} \left(1-\frac{1}{m}\right)^j - i \left(1-\frac{1}{m}\right)^{m-i} \right) \\
			&= \frac{1}{m^2} \left( (m-1) - \left(1-\frac{1}{m}\right) \frac{1 - \left(1-\frac{1}{m}\right)^{m-i-1}}{\frac{1}{m}} - i \left(1-\frac{1}{m}\right)^{m-i} \right) \\
            &= \frac{1}{m^2} \left( (m-1) - m \left(1-\frac{1}{m}\right) + m \left(1-\frac{1}{m}\right)^{m-i} - i \left(1-\frac{1}{m}\right)^{m-i} \right) \\
            &= \frac{1}{m} \left( \left(1-\frac{1}{m}\right)^{m-i} - \frac{i}{m} \left(1-\frac{1}{m}\right)^{m-i} \right) \\
            &= \frac{1}{m} \left( \left(1-\frac{i}{m}\right)\left(1-\frac{1}{m}\right)^{m-i} \right)
\end{align*}
Which implies that the left hand side of these constraints becomes $0$ and hence they are satisfied. All the $x_i$'s and $y_i$'s are trivially non-negative, which proves the feasibility.

Now we will show that the objective value for this solution is $(1-\frac{1}{m})^m$. Let objective function value corresponding to this solution be $S$. Hence,
\[
S = \sum_{i = 1}^{m} \frac{i}{m} y_i = y_m +\sum_{i = 1}^{m-1} \frac{i}{m} \frac{1}{m}\left(1-\frac{1}{m}\right)^{m-i-1} = \frac{1}{m^2}\sum_{i = 1}^{m-1} i \left(1-\frac{1}{m}\right)^{m-i-1}
\]

This is an arithmetic-geometric progression. Hence we multiply both sides by $\left(1-\frac{1}{m}\right)$ to get,
\[
\left(1-\frac{1}{m}\right)S = \left(1-\frac{1}{m}\right) \frac{1}{m^2}\sum_{i = 1}^{m-1} i \left(1-\frac{1}{m}\right)^{m-i-1} = \frac{1}{m^2}\sum_{i = 1}^{m-1} i \left(1-\frac{1}{m}\right)^{m-i}
\]

Subtracting the second equality from first we get,
\begin{align*}
\frac{S}{m} &=  -\frac{1}{m^2} \left(1-\frac{1}{m}\right)^{m-1} + \frac{m-1}{m^2} - \frac{1}{m^2} \sum_{i = 2}^{m-1} \left(1-\frac{1}{m}\right)^{m-i} \\
            &= \frac{1}{m} \left( -\frac{1}{m} \left(1-\frac{1}{m}\right)^{m-1} +\left(1-\frac{1}{m}\right) - \frac{1}{m}\left(1-\frac{1}{m}\right)  \sum_{i = 0}^{m-3} \left(1-\frac{1}{m}\right)^{i} \right) \\
            &= \frac{1}{m} \left( -\frac{1}{m} \left(1-\frac{1}{n}\right)^{m-1} +\left(1-\frac{1}{m}\right) - \frac{1}{m} \left(1-\frac{1}{m}\right) \frac{1-\left(1-\frac{1}{m}\right)^{m-2}}{1/m} \right) \\
            & = \frac{1}{n} \left( -\frac{1}{m} \left(1-\frac{1}{m}\right)^{m-1} + \left(1-\frac{1}{m}\right)  - \left(1-\frac{1}{m}\right) + \left(1-\frac{1}{m}\right)^{m-1} \right) \\
            &= \frac{1}{m} \left(1-\frac{1}{m}\right)^m
\end{align*}
Hence $S=\left(1-\frac{1}{m}\right)^m$, which proves the lemma.
\end{proof}

Now we show that there is a primal feasible solution with matching value to show that the solution is indeed optimal. Let $a_i = \frac{i}{m}\left(1-\frac{1}{m}\right)^{i}$ for $i \in [m]$, $o_i = \frac{1}{m}\left(1-\frac{1}{m}\right)^{i-1}$ for $i \in [m-1]$ and $o_m = 1- \sum_{i=1}^{m-1} o_i$.

\begin{lemma}
\label{apx-lem:submodular-const-lp}
The above solution is a feasible solution the factor-revealing LP. Moreover, the value of this solution is $\left(1-\frac{1}{m}\right)^m$.
\end{lemma}

\begin{proof}
Note that $a_1 = \frac{1}{m}\left(1-\frac{1}{m}\right) = \left(1-\frac{1}{m}\right) o_1$. Now for any inequality of the type $a_i \geq a_{i-1} +  \left(1-\frac{i}{m}\right) o_i$, the left-hand side is $\frac{i}{n}\left(1-\frac{1}{m}\right)^{i}$. The right-hand side equals
\[ \frac{i-1}{m}\left(1-\frac{1}{m}\right)^{i-1} +  \left(1-\frac{i}{m}\right) \frac{1}{m}\left(1-\frac{1}{m}\right)^{i-1}  = \left(1-\frac{1}{m}\right)^{i-1} \left(\frac{i}{m} - \frac{1}{m} + \frac{1}{m} - \frac{i}{m^2} \right) = \frac{i}{m}  \left(1-\frac{1}{m}\right)^{i} \]
hence they are satisfied.

Now for any inequality of the type $a_i \geq \frac{i}{m} (1- \sum_{j \leq i} o_j)$ for $i \in [m-1]$, the left-hand side is $\frac{i}{m}(1-\frac{1}{m})^{i}$. The right-hand side equals
\[\frac{i}{m} \left(1- \sum_{j =1}^i \frac{1}{m}(1-\frac{1}{m})^{j-1} \right) = \frac{i}{m} \left(1- \frac{1}{m}\frac{1 - (1-\frac{1}{m})^i}{1/m} \right) =  \frac{i}{m} \left(1- 1 + (1-\frac{1}{m})^i \right) = \frac{i}{m}(1-\frac{1}{m})^{i} \]

Hence they are satisfied as well. Inequality $a_m \geq  (1- \sum_{j = 1}^m o_j)$ is trivially true by definition of $o_n$. All the $a_i$'s and $o_i$'s (except $o_m$) are trivially non-negative. Finally, $o_m= 1 - \sum_{j = 1}^{m-1} o_j = 1 - \sum_{j =1}^{m-1} \frac{1}{m}\left(1-\frac{1}{m}\right)^{j-1} = 1 - \frac{1}{m}\frac{1 - \left(1-\frac{1}{m}\right)^{m-1}}{1/m} = 1 - 1 + \left(1-\frac{1}{m}\right)^{m-1} \geq 0$, which proves the feasibility for the solution. The objective value for the LP is $a_m=\left(1-\frac{1}{m}\right)^m$, hence the lemma follows.
\end{proof}

From the above to lemmas it follows that the the bounded provided by (LP) converges to $1/\e$ for $m\rightarrow\infty$, which is the approximation ratio of the greedy DP under the above simplifying assumption.



We directly get Lemma~\ref{lem:greedy-dp-vanilla}, using Lemma~\ref{apx-lem:submodular-const-multi-phase},~\ref{apx-lem:submodular-const-dual} and~\ref{apx-lem:submodular-const-lp}.


\subsection{Forbidden sets for a single packing and a single covering constraints}\label{subsec:submodular-one}
In this setting we are able to ensure a $(1-\eps)$-violation of the covering constraints by using the concept of \emph{forbidden sets}.
Intuitively, we exclude the elements of these set from being included to the dynamic programming table in order to be able to complete the table entries to solutions with only small violation.

\subsubsection{Algorithm}\label{subsubsec:forbiddenAlgorithm}

\noindent {\bf\em{Guessing:}} Fix some $\eps>0$. By guessing we assume that we know the set $G$ of all, at most $1/\eps$ elements $\ell$ from the optimum solution with $\bP_\ell\geq \eps\p$. (For consistency reasons, we use bold-face vector notation here also for dimension one.) We can guess $G$ using brute force in $n^{O(1/\eps)}$ time. This allows us to remove all elements with $\bP_\ell \geq \eps\p$ from the instance. Let $\cN'$ be the rest of the elements.

\noindent {\bf\em{Forbidden Sets:}} Fix an order of $\cN'$ in which the elements are sorted in a non-increasing order of $\bC_\ell/\bP_\ell$ values, breaking ties arbitrarily. Let $\cN_i$ be the set of the first $i$ elements in this order. For any $\p'\leq \p$, let $F_{\p'}$ be the smallest set $\cN_i$ with $\bP\one_{\cN_i}\geq \p-\p'$. Note that the profit of $F_{\p'}$ is at least the profit of any subset of $\cN'$ with packing value at most $\p-\p'$ and that the packing value of $F_{\p'}$ is no larger than $(1+\eps)\p-\p'$. Also note that for any $0 \leq \p'\leq \p'' \leq \p$, it holds that $F_{\p''} \subseteq F_{\p'}$.

\noindent {\bf\em{Greedy-DP algorithm with guessing and forbidden sets:}} Let $G$ be the set of the guessed big elements as described above.
For the base case, we set $T[\bC \one_G, \bP \one_G] = G$ and $T[\bc', \p']=\bot$ for all table entries with $\bc' \neq \bC \one_G$ or $\p' \neq \bP \one_G$.

In order to compute $T[\bc', \p']$, we look at every set of the form $T[\bc'-\bC_\ell,\p'-\bP_\ell] \cup \{\ell\}$, where $\ell\in\cN \setminus (T[\bc'-\bC_\ell,\p'-\bP_\ell] \cup F_{\p'})$, such that $\bc'- \bC_\ell \geq 0$ and $\p'- \bP_\ell \geq 0$. Notice that we forbid elements belonging to $F_{\p'}$ to be included in any table entry of the form $T[\bc', \p']$. Now out of all these sets, we assign the most valuable set to $T[\bc', \p']$. The output of our algorithm is the best of the solutions $T[\bc', \p'] \cup F_{\p'}$, such that $\bc' + \bC \one_{F_{\p'}} \geq \bc$.

The pseudo-code of the algorithm can be found below\footnote{For the sake of readability of the pseudo code, we use array notation $F[\p]$ rather than $F_{\p}$ for forbidden sets.}.

  \begin{algorithm*}[H] \label{alg:GreedyDPForbid}
  \caption{Greedy Dynamic Program with Guessing and Forbidden Sets}
  create a table $T\colon [n \cdot c_{\max}]_0 \times[p_{\max}]_0\rightarrow 2^{\cN}$ initialized with entries $\bot$\label{algline:init}\\
  sort the ground set so that $\cN=\{\ell_1, \ell_2, \dots, \ell_n\}$ and $\bC_{\ell_i}/\bP_{\ell_i} \geq \bC_{\ell_{i+1}}/\bP_{\ell_{i+1}}$ for all $i \in [n-1]$\\
  let $\cN_i=\{\ell_1,\dots,\ell_i\}$ for all $i\in[n]$ and let $\cN_0=\emptyset$\\
  create an array $F\colon  [\p]_0 \rightarrow 2^{\cN}$\\
  set $F[\p'] = \cN_i$, where $i$ is the smallest index in $[n]_0$ such that $\bP\one_{\cN_i}\geq \p-\p'$\\
  $S \gets \varnothing$ \\
  \ForEach{$G \subseteq \cN$, such that $|G| \leq 1/\eps, \bP\one_{G}\leq \p$}{
    $T[\bC\one_{G},\bP\one_{G}]\gets G$\\
    \ForEach{$\bc'\in [n \cdot c_{\max}]_0^c$ and $\p' \in [p_{\max}]_0^p$}{
        \ForEach{$\ell\in\cN \setminus (T[\bc',\p'] \cup F[\p'])$\label{algline:forEachElement}}{
            $\bc''\gets\bc'+\bC_{\ell}$, $\p''\gets\p'+\bP_{\ell}$\\
            $T[\bc'',\p''] \gets \arg\max\{f(T[\bc'',\p'']),f(T[\bc',\p']\cup\{\ell\})\}$\label{algline:greedyExtend}\\
        }
    }
    $\displaystyle S \gets \argmax\{f(S), \max_{\substack{\bc', \p':\\ \bc' + \bC \one_{F[\p']} \geq \bc}} \{ f(T[\bc',\p']\cup F[\p'])\}\}$\\
  }
  Output $S$.
  \end{algorithm*}

\subsubsection{Analysis}\label{subsubsec:forbiddenAnalysis}

\paragraph{A Warmup}
As for the vanilla version of the algorithm, we start by
giving a combinatorial proof that gives already a ratio of
$0.25$. Again, the proof contains some of the ideas and technical
ingredients used in the factor-revealing LP.

We first prove the following simple but crucial observation.
\begin{lemma}
Any table entry $T[\bc', \p']\neq\bot$ with $0\leq \p'\leq \p$ and $0\leq \bc'\leq nc_{\max}$ is disjoint from $F_{\p'}$. 
\end{lemma}
\begin{proof}
We prove the claim for all entries $T[\bc', \p']\neq\bot$ by induction on $\p'$. Hence we consider only entries $\bc'\geq \bC \one_G$ and $\p'\geq \bP \one_G$. 

The claim clearly holds for $T[\bC \one_G, \bP \one_G] = G$ since $G$ is disjoint from $\cN'$ and since $F_{\bP \one_G}\subseteq \cN'$.

Note that any table entry $T[\bc', \bP \one_G]$ with $\bc'>\bC \one_G$ is $\bot$ as there are no zero-weight elements in $\cN'$. Now consider an entry $T[\bc', \p']\neq\bot$ with $\p'>\bP \one_G$. Let $s^*$ be the element used by the DP for computing this entry, that is, $T[\bc',\p']= \mathrm{pred}(s^*)\cup\{s^*\}$, where $\mathrm{pred}(s^*)=T[\bc'',\p'']$, $\bc''=\bc'-\bC_{s^*}$, and $\p''=\p'-\bP_{s^*}$. By line~\ref{algline:forEachElement} Algorithm~\ref{alg:GreedyDPForbid}, we have that $s^* \notin F_{\p'}$. By the inductive hypothesis, we have on the other hand that $T[\bc'',\p'']$ is disjoint from $F_{\p''}$. The claim now follows for $T[\bc',\p']$ because $F_{\p'}\subseteq F_{\p''}$.
\end{proof}

Let $O$ be an optimal set solution. As for the vanilla
version fix an arbitrary permutation of $O$ and define a
$g\colon O \rightarrow \mathbf{N}$ of marginal values with respect to
this permutation.


\begin{lemma}
\label{lem:submodular-one}
There are $\p',\bc' \geq 0$ such that $\p'\leq \p$, $\bc' + \bC\one_{F_{\p'}} \geq \bc$, and $f(T[\bc', \p']\cup F_{\p'})\geq\frac14 f(O)$.
\end{lemma}

\begin{proof}
For the sake of a contradiction let us assume that there is no such table entry as stated in the lemma.

Let $q_0$ be the number of guessed elements in the guessing phase.
We will show that the following claim holds under the assumption that the lemma is not true.
\begin{claim}
For any $q$ with $q_0 \leq q \leq k$ there is a $q$-subset $O_q=\{o_1,\dots,o_q\}$ of $O$ with total packing value $\p_q$ and total covering value $\bc_q$ such that $O_q$ is disjoint from $F_{\p_q}$ and $f(T[\bc_q, \p_q])\geq \frac12 g(O_q)$ holds.
\end{claim}

Note that this claim already yields the desired contradiction by considering the case $q=k$. To this end, note that $O_k=O$, that by monotonicity we have $f(T[\bc_k, \p_k]\cup F_{\p_k})\geq f(T[\bc_k, \p_k])\geq\frac12g(O_k)=\frac12 f(O)$, and also $\bc_k+\bC \one_{F_{\p_k}})\geq \bc_k = \bC\one_O \geq \bc$.

We now prove the above claim by induction on $q$. The claim is true for $q=q_0$ since we can set $O_{q_0} = G$ to the set of guessed elements and $\bc_{q_0} = \bC\one_G$ and $\p_{q_0} = \bP\one_G$ . Then we have $f(T[\bc_{q_0},\p_{q_0}])=f(O_{q_0}) \stackrel{\text{Lem.~\ref{apx-lem:subset}}}\geq g(O_{q_0})$.

For the inductive step assume now that $q \geq q_0+1$ and assume that the claim already holds for $q-1$. To this end, let $O_{q-1}$, $\bc_{q-1}$ and $\p_{q-1}$ be as in this claim. Let $S=T[\bc_{q-1}, \p_{q-1}]$.

We distinguish between two cases. In the first case there is an $\ell_q \in O \setminus (O_{q-1}\cup F_{\p_{q-1}})$ such that $f(S \cup \{ \ell_q \}) - f(S) > \frac12 g(\ell_q)$. Note that $\ell_q \notin S$, since otherwise the left hand side of the inequality will be zero while $g(\ell_q)$ is non-negative, which is a contradiction. Now, let $O_{q} = O_{q - 1} \cup \{\ell_q\}$, $\bc_q = \bc_{q-1} + \bC_{\ell_q}$ and $\p_q=\p_{q-1}+\bP_{\ell_q}$. Note that $\ell_q\notin F_{\p_q}$ since $\ell_q\notin F_{\p_{q-1}}$ and $F_{\p_q}\subseteq F_{\p_{q-1}}$, hence the DP can also add this element $\ell_q$ to the entry $T[\bc_{q-1}, \p_{q-1}]$ to get $T[\bc_q, \p_q]$. Now verify that
\[ f(T[\bc_q, \p_q]) \geq f(S\cup\{\ell_q\}) \geq f(S)+\frac12 g(\ell_q) \geq \frac12 \sum_{\ell_i\in O_{q-1}}g(\ell_i)+\frac12 g(\ell_{q}) = \frac12 \sum_{\ell_i\in O_{q}} g(\ell_i). \]

In the second case, for all $a \in O \setminus (O_{q-1}\cup F_{\p_{q-1}})$, we have $f(S \cup \{ a \}) - f(S) \leq \frac12 g(a)$.  In this case we will arrive at a contradiction to our assumption that the lemma is not true.

For this let us define $S'=S$ and $O'=O_{q-1}$ and look at elements in $(F_{\p_{q-1}} \setminus O') \cap O$. If $\exists b \in (F_{\p_{q-1}} \setminus O') \cap O$, such that $f(S \cup \{b\}) - f(S) > \frac12 g(b)$, then we define $S'=S \cup \{b\}$ and $O'=O' \cup\{b\}$. We know by induction hypothesis that $f(S)\geq \frac12 g(O_{q-1}) = \frac12 g(O')$. These two things together imply that $f(S') \geq \frac12 g(O')$. Now we again check whether there $\exists b \in (F_{\p_{q-1}} \setminus O') \cap O$, such that $f(S' \cup \{b\}) - f(S') > \frac12 g(b)$. If yes, then we define $S'=S \cup \{b\}$ and $O'=O' \cup\{b\}$ and get $f(S') \geq \frac12 g(O')$. We iterate this process until $\forall b \in (F_{\p_{q-1}} \setminus O') \cap O$, $f(S' \cup \{b\}) - f(S') \leq \frac12 g(b)$. Note that $S'$ is a union of $S$ and a few elements from $F_{\p_{q-1}} \cap O$. Hence $S\subseteq S'$ and by submodularity of $f$ we get $f(S' \cup \{ a \}) - f(S') \leq f(S \cup \{ a \}) - f(S)$ for any $a \notin S'$. This in turn implies that for all $a \in O \setminus (O_{q-1}\cup F_{\p_{q-1}})$, we have $f(S' \cup \{ a \}) - f(S') \leq f(S \cup \{ a \}) - f(S) \leq \frac12 g(a)$. Note that $(F_{\p_{q-1}} \setminus O') \cap O$ and $ O \setminus (O_{q-1}\cup F_{\p_{q-1}})$ are disjoint and $((F_{\p_{q-1}} \setminus O') \cap O) \cup(O \setminus (O_{q-1}\cup F_{\p_{q-1}})) = O \setminus O'$. Let $O \setminus O' = \{a_1, a_2, \dots a_m\}$. By submodularity of $f$ we have, $f(S' \cup \{a_1, a_2, \dots a_j\}) - f(S' \cup \{a_1, a_2, \dots a_{j-1}\}) \leq f_{S'}(a_j)\leq \frac12 g(a_j)$ for all $j=1,2, \dots m$. Adding up all these inequalities, we get
\begin{align*}
f(S' \cup (O \setminus O'))-f(S') &= \sum_{j=1}^m [f(S'\cup \{a_1,\dots,a_j\}) \\
& \qquad\qquad -f(S'\cup\{a_1,\dots,a_{j-1}\})]\\
& \leq \frac12\sum_{j=1}^m g(a_j)\\
& = \frac12 g(O \setminus O')\,.
\end{align*}

Rearranging, we get
\begin{align*}
f(S') & \geq f(S' \cup (O \setminus O')) - \frac12 g(O\setminus O')\\
 & \geq f(O \setminus O') - \frac12 g(O \setminus O')\\
 & \stackrel{\text{Lem.~\ref{apx-lem:subset}}}{\geq} \frac12 g(O \setminus O')\,.
\end{align*}

Adding this inequality to the final inequality which we get after our iterative process, i.e. $f(S') \geq \frac12 g(O')$, we get $4 f(S') \geq g(O)$. Also, $f(S \cup F_{\p_{q-1}}) \geq f(S')$ by monotonicity of $f$, together with the previous inequality contradicts the assumption that the claim of the lemma is not true.
\end{proof}

The following lemma follows directly from Lemma~\ref{lem:submodular-one} for $q=k$. 

\begin{lemma}\label{lem:submodular-one-algo}
 The above algorithm outputs for any $\eps>0$ in time $n^{O(1/\eps)}c_{\max}p_{\max}$ a $0.25$-approximate solution with covering value at least $\bc$ and with packing value at most $(1+\eps)\p$.
\end{lemma}
\begin{proof}
 As shown in Lemma~\ref{lem:submodular-one}, the solution output by the algorithm is $0.25$-approximate and has covering value at least $\bc$. If $\p'$ is the weight of the table entry, then the total weight of the solution output is $\p'+\bP \one_{F_{\p'}}\leq \p'+(1+\eps)\p-\p'=(1+\eps)\p$ as claimed.

Observe that the table has $O(nc_{\max}p_{\max})$ entries (Line~\ref{algline:init}, Algorithm~\ref{alg:GreedyDPForbid}) and computing each of the entries takes $O(n)$ time (Line~\ref{algline:forEachElement}, Algorithm~\ref{alg:GreedyDPForbid}), hence we get the stated running time.
\end{proof}

Using standard scaling techniques, we can bring the running time down to polynomial at the expense of also violating the covering constraint by a factor $(1-\eps)$. 

\begin{lemma}
\label{lem:submodular-one-0.25}
There is an algorithm for maximizing a monotone submodular function subject to one covering and one packing constraint that outputs for any $\eps>0$ in $n^{O(1/\eps)}$ time a $0.25$-approximate solution with covering value at least $(1-\eps)\bc$ and with packing value at most $(1+\eps)\p$.
\end{lemma}
\begin{proof}
We will scale our instance and then apply Lemma~\ref{lem:submodular-one-algo}. First, we can assume that $c_{\max} \leq \bc$ because for all elements $\ell$ with covering value strictly larger than $\bc$ we can set $\bC_\ell=\bc$ thereby obtaining an equivalent instance with this property. Similarly, we can assume that $p_{\max}\leq \p$ by removing all elements $\ell$ with $\bP_\ell>\p$ from the instance.

Let us assume that $\eps\leq 1$. Let $K_\bc=\frac{\eps c_{\max}}{n}$ and $K_\p=\frac{\eps p_{\max}}{2n}$ be ``scaling factors'' and set $\bC_\ell'=\lceil \bC_\ell/K_\bc\rceil$ and $\bP_\ell'=\lfloor \bP_\ell/K_\p\rfloor$ for all $ \ell \in \cN$. Moreover, define new covering and packing bounds $\bc'=\lceil \bc/K_\bc\rceil$ and $\p'=\lfloor \p/K_\p\rfloor$.

Let $O$ be the optimum solution for the original, unscaled covering and packing values. Note that $\bC'\one_O\geq \bc'$ and $\bP'\one_O\leq \p'$ and hence $O$ is a feasible solution also with respect to the instance with scaled scaled covering and packing values $\bC_\ell', \bP_\ell'$ for $\ell\in\cN$ and bounds $\bc',\p'$.

Let $S$ be the solution output by the algorithm of Lemma~\ref{lem:submodular-one-algo} in the down-scaled instance where we use the error parameter $\eps'=\eps/2$. By the claim of the lemma, we have $f(S)\geq 0.25 f(O)$ since $O$ is also feasible in the scaled instance. Moreover, $\bC'\one_S\geq \bc'$ and $\bP'\one_S\leq (1+\eps/2)\p'$.

Now, we prove that the solution $S$ obeys the violation bounds in the \emph{original}, unscaled instance as claimed by the lemma. In fact, for any element $\ell \in\cN$ we have that $K_\bc \cdot \cC_\ell'\leq \bC_\ell+K_\bc$. Hence
\begin{align*}
\bC\one_S & \geq K_\bc\cdot \bC'\one_S - nK_\bc\\
 & = K_\bc\cdot \bC'\one_S - \eps c_{\max}\\
 & \geq K_\bc\cdot \bc'-\eps c_{\max}\\
 & \geq K_\bc \frac{\bc}{K_\bc}-\eps c_{\max}\\
 & \geq (1-\eps)\bc \, .
\end{align*}
Similarly, $K_\p\cdot \bP_\ell'\geq \bP_\ell-K_\p$ and thus
\begin{align*}
\bP\one_S & \leq K_\p\cdot \bP'\one_S +nK_\p\\
 & = K_\p\cdot \bP'\one_S + \frac{\eps}{2} p_{\max}\\
 & \leq K_\p\cdot \left(1+\frac{\eps}{2}\right)\p'+\frac{\eps}{2} p_{\max}\\
 & \leq K_\p \left(1+\frac{\eps}{2}\right)\frac{\p}{K_\p}+\frac{\eps}{2} p_{\max}\\
 & \leq (1+\eps)\p\, .
\end{align*}

By Lemma~\ref{lem:submodular-one-algo} the running time is $n^{O(1/\eps)}\frac{c_{\max}}{K_\bc}\frac{p_{\max}}{K_\p}=n^{O(1/\eps)}$.

\end{proof}

\paragraph{Factor-Revealing LP}
\label{sec:submodular-fact-rev-LP}
As in the vanilla version we extend the previous two-phase analysis to
a multi-phase analysis.

For the purpose of analysis, by scaling, we assume that $f(O) = \sum_{o \in O} g(o)=1$. We set the starting solution $A_0 \triangleq G$ to the set of guessed elements. For the analysis, we also assume that the guessed elements are the first ones in the permutation used to define the function $g$, hence $g(G) = f(G)$.

Now we extend the idea of the factor-revealing LP to forbidden sets. We first do our multi-phase analysis until we reach the solution $A_m = T[p(O\setminus F_m), w(O \setminus F_m)]$, where $F_m$ is the forbidden set corresponding to this table entry. Now we will show that the best of the solutions $A_i \cup F_i$ for $i \in [m]$ has the value at least a factor $0.353$ times the value $g(O) = f(O)$.

The following lemma is the basis of our factor-revealing LP.
\begin{lemma}
\label{lem:submodular-one-multi-phase}
Let $m\geq 1$ be a integral parameter denoting the number of phases. We can pick for each $i \in [m] \cup \{0\}$ a set $O_i = \{o^i_1, o^i_2, \dots o^i_{q_i}\} \subseteq O$ (possibly empty) and an $\alpha_i\in [0,1]$ such that the following holds. For any $i\neq j$, we have that $O_i \cap O_j = \emptyset$. Let $Q_i \triangleq\cup_{j=0}^iO_j$, $\bc_i \triangleq \bC \one_{Q_i}$ and $\p_i \triangleq \bP \one_{Q_i}$ and let $A_i \triangleq T[\bc_i, \p_i]$ be the corresponding entry in the DP table. Then set $O_i$ is disjoint from the forbidden set $F_i \triangleq F_{\p_i}$ for any $i$. Finally, the following inequalities hold:
\begin{enumerate}
\item \label{ineq:o0} $f(A_0) =  g(O_0)$,
\item \label{ineq:o1} $f(A_i) \geq f(A_{i-1}) +  \left(1-\frac{i}{m}\right) g(O_i) \qquad \forall i \in [m]$ and
\item \label{ineq:o2} $f(A_i \cup F_i) \geq f(A_i) + \alpha_i g(F_i \cap O) \qquad \forall i \in \{0\} \cup [m]$.
\item \label{ineq:o3} $f(A_i) \geq \frac{i}{m} \left(1- g(F_i \cap O) - \sum_{j=0}^i g(O_j)\right) + (1 - \alpha_i) g(F_i \cap O) \qquad \forall i \in \{0\} \cup [m]$.
\end{enumerate}
\end{lemma}
\begin{proof}
  We prove this by using induction on $i$. Let $G$ be the set of guessed elements from $O$ in the guessing phase. For $i=0$, inequality \ref{ineq:o0} is trivially true by picking $O_0=A_0=T[\bC \one_G, \bP \one_G] = G$. Let $F_0 = F_{\bP \one_G}$. If $g(F_0 \cap O) = 0$, then the inequalities~\ref{ineq:o2} and~\ref{ineq:o3} are trivially true. Otherwise, let $\alpha_0 = \min\{1, \frac{f(A_0 \cup F_0) - f(A_0)}{g(F_0 \cap O)}$ \}. For this choice inequality~\ref{ineq:o2} holds trivially. Also for $\alpha_0 =1$ the inequality~\ref{ineq:o3} holds trivially. Otherwise $f(A_0) = f(A_0 \cup F_0) - \alpha_0 g(F_0 \cap O) \geq f(F_0 \cap O) - \alpha_0 g(F_0 \cap O) \geq (1-\alpha_0) g(F_0 \cap O)$, which concludes the base case.

For the inductive step let $i \geq 1$ and assume that inequalities \ref{ineq:o1}, \ref{ineq:o2} and \ref{ineq:o3} are true for every $j < i$ (For $i=1$ inequality~\ref{ineq:o0} holds instead of inequality~\ref{ineq:o1}.). We start by defining $O_i \triangleq  \emptyset$, $A_i \triangleq A_{i-1}$ and $F_i = F_{i-1}$. Note that the inequality~\ref{ineq:o1} is true for this choice of $O_i$ and $A_i$. Now if there is an $o \in O \setminus F_i \cup (\cup_{j=0}^{i} O_j)$, such that $f(A_i \cup \{o\}) - f(A_i) > (1-\frac{i}{m}) g(o)$, then this implies that $T[\bC \one_{A_i\cup \{o\}}, \bP \one_{A_i\cup \{o\}}] \geq f(A_i \cup \{o\}) > f(A_i) + (1-\frac{i}{m}) g(o) \geq f(A_{i-1}) + (1-\frac{i}{m})g(O_i) + (1-\frac{i}{m}) g(o) =  f(A_{i-1}) + (1-\frac{i}{m})g(O_i\cup\{o\})$. Note that $o \notin A_i$, since otherwise the left hand side of the inequality will be zero and $g(o)$ is non-negative. Let $Q'_i \triangleq\cup_{j=0}^iO_j \cup \{o\}$. Hence, we can extend $O_i \rightarrow O_i \cup \{o\}$, $A_i \rightarrow  f(T[\bC\one_{Q'_i}, \bP\one_{Q'_i}])$ and $F_i \rightarrow F_{\bP\one_{Q'_i}}$. Hence, inequality \ref{ineq:o1} remains true after performing this operation on $O_i$ and $A_i$. We keep on doing this until for all $o \in O \setminus (F_i \cup (\cup_{j=0}^{i} O_j))$,  $f(A_i \cup \{o\}) - f(A_i) \leq (1-\frac{i}{m}) g(o)$. Then the construction of the set $O_i$ is completed.

Now we show that inequalities~\ref{ineq:o2} and~\ref{ineq:o3} hold as well. Let us define $\alpha_i \triangleq \min\left\{1, \frac{f(A_i \cup F_i) - f(A_i)}{g(F_i \cap O)}\right\}$ if $g(F_i \cap O) \neq 0$, and $\alpha_i \triangleq 1$ otherwise. This definition directly implies inequality~\ref{ineq:o2}.

Assume first that $\alpha_i < 1$. Let $O \setminus (F_i \cup (\cup_{j=0}^i O_j)) = \{ a_1, a_2, \dots, a_q\}$. By submodularity we have that $f(A_i \cup \{a_1, \dots, a_j\})-f(A_i \cup \{ a_1, \dots, a_{j-1}\})\leq f(A_i \cup \{a_j\})-f(A_i)\leq (1-i/m) g(a_j)$ for each $j=1, \dots, q$. We also have $f(A_i \cup F_i \cup (O \setminus (F_i\cup(\cup_{j=0}^i O_j)))) - f(A_i \cup (O \setminus (F_i\cup (\cup_{j=0}^i O_j)))) \leq f(A_i \cup F_i) - f(A_i) = \alpha_i g(F_i \cap O)$ as $\alpha_i<1$. By adding up the previous inequalities for all $j=1, \dots, q$ and the last inequality, we get that
\begin{align*}
f(A_i \cup F_i \cup (O \setminus \cup_{j=0}^i O_j))-f(A_i) &= \sum_{j=1}^q [f(A_i \cup \{a_1,\dots,a_j\})\\
& \qquad\qquad - f(A_i \cup \{ a_1, \dots, a_{j-1}\})] + \alpha_i g(F_i \cap O)\\
& \leq \left(1-\frac{i}{m}\right)  \sum_{j=1}^q g(a_j) + \alpha_i g(F_i \cap O)\\
& = \left(1-\frac{i}{m}\right) g(O \setminus ((\cup_{j=0}^i O_j) \cup F_i)) + \alpha_i g(F_i \cap O)\,.
\end{align*}
Rearranging, we get
\begin{align*}
f(A_i)    & \geq f(A_i \cup F_i \cup (O \setminus \cup_{j=0}^i O_j)) - \left(1-\frac{i}{m}\right) g(O \setminus (F_i \cup (\cup_{j=0}^i O_j))) - \alpha_i g(F_i \cap O) \\
             & \geq f((F_i \cap O) \cup (O \setminus (\cup_{j=0}^i O_j \cup F_i))) - \left(1-\frac{i}{m}\right) g(O \setminus (F_i \cup (\cup_{j=0}^i O_j))) - \alpha_i g(F_i \cap O)\\
             & \stackrel{\text{Lem.~\ref{apx-lem:subset}}}{\geq} \frac{i}{m} g(O \setminus (F_i \cup (\cup_{j=0}^i O_j))) + (1- \alpha_i) g(F_i \cap O) \\
             & = \frac{i}{m} \left(1 - g(F_i \cap O) - \sum_{j=0}^i g(O_j)\right) + (1- \alpha_i) g(F_i \cap O)\,.
\end{align*}

The case $\alpha_i = 1$ is similar to the above and in fact simpler. For the sake of completeness, we state it here. Let $O \setminus (F_i\cup(\cup_{j=0}^i O_j)) = \{ a_1, a_2, \dots, a_q\}$. By submodularity we have that $f(A_i \cup \{a_1, \dots, a_j\})-f(A_i \cup \{ a_1, \dots, a_{j-1}\})\leq f(A_i \cup \{a_j\})-f(A_i)\leq (1-i/m) g(a_j)$ for each $j=1, \dots, q$. Adding up these inequalities for all $j=1, \dots, q$ we get that
\begin{align*}
f(A_i \cup (O \setminus (F_i\cup (\cup_{j=0}^i O_j))))-f(A_i) &= \sum_{j=1}^q [f(A_i \cup \{a_1,\dots,a_j\})\\
& \qquad\qquad - f(A_i \cup \{ a_1, \dots, a_{j-1}\})]\\
& \leq \left(1-\frac{i}{n}\right)  \sum_{j=1}^q g(a_j)\\
                                                            & = \left(1-\frac{i}{m}\right) g(O \setminus (F_i\cup(\cup_{j=0}^i O_j)))\,.
\end{align*}
Rearranging, we get
\begin{align*}
  f(A_i) & \geq f(A_i \cup (O \setminus (F_i\cup(\cup_{j=0}^i O_j)))) - \left(1-\frac{i}{m}\right) g(O \setminus (F_i\cup (\cup_{j=0}^i O_j)))\\
             & \geq f(O \setminus (F_i\cup(\cup_{j=0}^i O_j))) - \left(1-\frac{i}{m}\right) g(O \setminus (F_i\cup(\cup_{j=0}^i O_j)))\\
             & \stackrel{\text{Lem.~\ref{apx-lem:subset}}}{\geq} \frac{i}{m} g(O \setminus (F_i\cup(\cup_{j=0}^i O_j))) \\
             & = \frac{i}{m} \left(1 - g(F_i\cap O)- \sum_{j=0}^i g(O_j)\right)\,.
\end{align*}
This implies that inequality \ref{ineq:o3} is also true for this $A_i$ and $O_i$ and hence the induction follows.

Hence inequality \ref{ineq:o3} is also true and hence the induction follows.
\end{proof}
Below, we state our \emph{factor-revealing} LP that captures the above-described multi-phase analysis with forbidden sets. The idea is to introduce variables for the quantities in the inequalities in the previous lemma and determining the minimum ratio that can be guaranteed by these inequalities.
\begin{align}
\min \,\, & \textstyle c 							      	   								&& \mbox{s.t.} &&  \hspace{11em}  \mbox{(LP-F)} \nonumber \\
	   & \textstyle a_0 = o_0;						     	   								&& \\
           & \textstyle a_i \geq a_{i-1} +  \left(1-\frac{i}{m}\right) o_i                   								&& \forall i \in [m];~\label{lp1:c1} \\ 
           & \textstyle b_i \geq a_i +  g_i			             										&& \forall i \in [m] \cup \{0\};~\label{lp1:c2} \\ 
           & \textstyle a_i \geq \frac{i}{m} (1- f_i - \sum_{j \leq i} o_j)  + f_i - g_i 						&& \forall i \in [m] \cup \{0\};~\label{lp1:c3} \\
           & \textstyle b_i \geq f_i				             										&& \forall i \in [m] \cup \{0\};~\label{lp1:c4} \\ 
           & \textstyle f_i \leq  f_{i-1}		          	      			      							&& \forall i \in [m];\\
           & \textstyle g_i \leq  f_i			          	      			      							&& \forall i \in [m] \cup \{0\};\\
           & \textstyle f_j + \sum_{i=0}^j o_i \leq 1;           	      			      						&& \forall j \in [m] \cup \{0\};\\ 
           & \textstyle c \geq b_i           	      			      										&& \forall i \in [m] \cup \{0\};~\label{lp1:c0} \\ 
           & \textstyle a_i \geq 0, \,\, o_i \geq 0 \,\, f_i \geq 0 \,\, g_i \geq 0	     						&& \forall i \in [m] \cup \{0\}.
\end{align}

The variable $o_i$ corresponds to the marginal value $g(O_i)$ for the set $O_i$ in our analysis. Variables $a_i$, $b_i$ and $f_i$ correspond to the quantities $f(A_i)$, $f(A_i \cup F_i)$ and $g(F_i)$ for the approximate solution $A_i$ and the corresponding forbidden set $F_i$ for each phase $i=1, 2, \dots, m$. The variable $g_i$ corresponds to the value $\alpha_i g(F_i \cap O)$. We add all the inequalities we proved in Lemma~\ref{lem:submodular-one-multi-phase} as the constraints for this LP (and additional obvious inequalities). Note that since $f(O)=1$, the minimum possible value of $\max_i b_i$ will correspond to a lower bound on the approximation ratio of our algorithm, which is captured by variable $c$ in the LP above.

\paragraph{Upper and Lower Bounds for the Factor-Revealing LP}
\label{subsec:lp_bounds}
For every positive integer $m$, the above factor-revealing LP provides a lower bound on the approximation ratio of our greedy DP described in Section~\ref{subsec:submodular-one}. Ideally, we would like to analytically determine the limit to which this bound converges when $m$ tends to infinity. Unfortunately, giving such a bound seems quite intricate due to the complexity of the LP.

Therefore, we shall first analyze the LP for the case that $f_i=0$ for all $i\in\{0\}\cup[m]$. This corresponds to the assumption that the forbidden sets do not contain elements from the optimum solution. Notice that the (LP) in Section~\ref{apx-subsec:factor-rev-lp} is precisely the LP which results in the case when $f_i=0$ for all $i\in\{0\}\cup[m]$. Under this assumption, we are know by Lemma~\ref{apx-lem:submodular-const-lp} and~\ref{apx-lem:submodular-const-dual} that this simplified LP converges to $1/\mathrm{e}$ for $m\rightarrow\infty$, which is the approximation ratio of the greedy DP under the above simplifying assumption. This raises the question if the optimum solution of (LP-F) tends to $1/\e$ for increasing $m$ as well. Below we show that the bound provided by (LP-F) actually remains below $1/\e-\eps_0$ for any $m$ and for some constant $\eps_0>0$.

Fix an arbitrary positive integer $m$. To show the above claim, we start with the solution $a,o$ described above, which is optimum for (LP). The solution used there is $a_0=0$, $a_i = \frac{i}{m}(1-\frac{1}{m})^{i}$ for $i \in [m]$, $o_i = \frac{1}{m}(1-\frac{1}{m})^{i-1}$ for $i \in [m-1]$ and $o_m = 1- \sum_{i=1}^{m-1} o_i$. Note that by setting $f=g=0$ and $b=a$, we obtain a feasible solution for (LP-F) as well with the same objective function value tending to $1/\e$ for $m\rightarrow\infty$. We now alter this solution so that $f$ attains positive values. This will give us some leverage to alter $o$ and $a$ suitably to actually decrease the objective value by some small positive amount.

We fix parameters $\alpha=0.625, \beta=0.0517, \gamma=0.0647$ More specifically, we set $f_i'=g_i'=\gamma$, $o_i'=o_i-\beta/m$ for $i=1,\dots,m/2-1$.  For $i=m/2$, we set $f_i'=g_i'=0$, $o_i'=o_i+\alpha\beta$. For $i=m/2+1,\dots,m$ we set $f_i'=g_i'=0$ and $o_i'=o_i$. As we specified the values for $o',f'$ and $g'$ the optimum values for $a'$ and $b'$ are ``determined'' in a straightforward way by the inequalities of (LP-F). We give the explicit values below.

The intuition why the above alteration of the solution decreases the objective function value is as follows. We slightly decrease the values of $o_i$ for $i<m/2$. This decreases the RHS of inequality~(\ref{lp1:c1}). We have to compensate for this decrease in two ways. First, by picking $f_i=g_i$ large enough, we ensure that also the RHS of (\ref{lp1:c3}) decreases. We set, however, $f_i=g_i=0$ for $i\geq m/2$ in order to avoid that RHS of~(\ref{lp1:c2}) comes too close to $1/\e$ with increasing $i$. To ensure that RHS of~(\ref{lp1:c3}) still decreases, we increase $o_{m/2}$ by \emph{strictly more} ($\alpha>1/2$) than the total decrease of the $o_i$ for $i<m/2$. By picking $\alpha$ not too large, on the other hand, there remains a decrease of the RHS of~(\ref{lp1:c1}) as in this inequality the coefficient of the $o_{m/2}$ is smaller than the coefficients of $o_i$ for $i<m/2$.

\begin{lemma}
\label{lem:upperBoundLPF}
For every $m > 2$, there is a feasible solution $a', b', o', f', g'$ for (LP-F), such that the corresponding objective value for the solution is $1/\e-\beta(\alpha-1/2)/2 + 3\beta/(4m) < 0.3647$.
\end{lemma}
\begin{proof}
We now describe how we pick $a'$ and $b'$. In particular, we set $a_i'=a_i-\beta/ m\sum_{j=1}^i(1-j/m)$, as determined by~(\ref{lp1:c1}), and $b_i'=a_i'+\gamma$ for $i<m/2$. Note that the bound for $a'_{m/2-1}$ by ~(\ref{lp1:c3}) is $a_{m/2-1} - (1/2-1/m)(\gamma - \beta/m(m/2-1))  <   a_{m/2-1} - 0.01942 + 0.0517/m^2 + 0.013/m$ which is smaller than the bound from ~(\ref{lp1:c1}), i.e. $a_{m/2-1} -\beta/ m\sum_{j=1}^{m/2-1}(1-j/m) = a_{m/2-1} - 3\beta/8 + 3\beta/4m >  a_{m/2-1} -0.01939 + 0.038775/m$, for $m\geq 2$. 

For $i=m/2$, we set $a_{m/2}'=a_{m/2}-\beta(\alpha-1/2)/2+3\beta/(4m)$, determined by~(\ref{lp1:c1}). The bound by ~(\ref{lp1:c3}) is $a_{m/2}-\beta(3/4-\alpha)/2- \beta/(2m)$ which is smaller that the previous bound. This, in turn, determines $b_{m/2}'=a_{m/2}'$.  

Finally, we can set $a_i'=a_i-\beta(\alpha-1/2)/2 + 3\beta/(4m)$ and $b_i'=a_i'$ for $i>m/2$ as $o_i'=o_i$ for these values of $i$. These settings yield a feasible solution to (LP-F). The two potential candidates which can determine the objective value are $b'_{m/2-1}$ and $b'_m$. Out of these $b'_m$ is the larger one, which gives an objective value of $1/\e-\beta(\alpha-1/2)/2 + 3\beta/(4m) < 0.3647$, for every $m > 2$.
\end{proof}

Unfortunately, we are not able to analytically determine the approximation ratio to which (LP-F) converges. By computational experiments, we can however show that the answer is not too far from the above upper bound. The following table and the plot shows these bounds for some specific values of $m$ obtained by solving the LP to optimality by means of an LP solver. 

\begin{table}[ht]
\begin{minipage}[b]{0.4\linewidth}
\centering
 \begin{tabular}{ | c | c | }
    \hline
    $m$ & $\min c$ \\
    \hline
    	  2 & 0.25 \\
    	  5 & 0.31727598 \\
        10 & 0.33592079 \\
        50 & 0.34990649 \\
      100 & 0.35160444 \\ 
      500 & 0.35295534 \\ 
    1000 & 0.35312374 \\ 
    2000 & 0.35320790 \\
    5000 & 0.35325839 \\
    \hline
  \end{tabular}
  \caption{Solutions to the (LP-F)}
    \label{tab:opt_lpf}
\end{minipage}\hfill
\begin{minipage}[b]{0.6\linewidth}
\centering
    \includegraphics[width=0.8\textwidth]{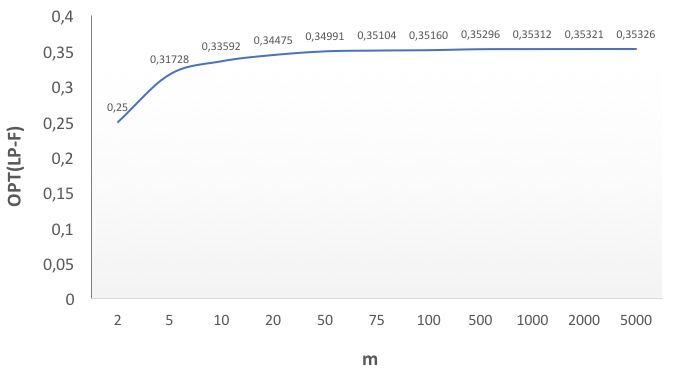}
    \captionof{figure}{Growth for the OPT (LP-F) solution}
\label{fig:opt_lpf}
\end{minipage}
\end{table}

Since each of the above (LP-F) values is a valid approximation factor for our algorithm, the following lemma follows. \begin{lemma}
 \label{lem:greedy-dp-forbidden-pseudopoly}
Assuming $p=c=1$ are constants, the Algorithm~\ref{alg:GreedyDPForbid} for \PCSM runs in pseudo-polynomial time $O(n^2p_{\max}c_{\max})$ and outputs a solution $S\subseteq \cN$ that satisfies: $(1)$ $f(S) >  0.353 \cdot f(O)$, $(2)$ $\bP \mathbf{1}_S\leq (1+\eps)\p$ and $\bC \mathbf{1}_S\geq \bc$.
\end{lemma}

Finally, if the packing constraint is actually a cardinality constraint we can assume that $\eps<1/\p$. Hence, there will be no violation of the cardinality constraint and also guessing can be avoided. These observations leads to the following lemma.

\begin{lemma}\label{lem:submodular-card-pseudopoly}
For \PCSM where $p=c=1$ and $\bP=\mathbf{1}_n^{\intercal}$, the Algorithm~\ref{alg:GreedyDPForbid} without the guessing step, outputs for any $\eps>0$ in time $O(n^3c_{\max})$ a $0.353$-approximate solution with covering value at least $\bc$ and cardinality $k$.
\end{lemma}

Hence the results for one covering and one packing constraint case stated in Theorem~\ref{thm:submodular-one} follows from Lemma~\ref{lem:greedy-dp-forbidden-pseudopoly} and~\ref{lem:submodular-card-pseudopoly} in combination with the scaling argument used in Lemma~\ref{lem:submodular-one-0.25}.





\subsubsection{A Connection to the Capacitated $k$-Median Problem}\label{subsec:cap-k-med}

In this section, we point out an interesting connection of $p=c=1$ case of \PCSM where packing constraint is a cardinality constraint, to the
well-studied capacitated $k$-median problem and give a non-trivial
approximability result without any violation for a special metrics.
We consider the well-studied $k$-median problem with non-uniform and
hard capacities\footnote{The term ``hard capacities'' refers to the
  restriction that each facility $i \in F$ can be opened at most
  once. In soft-capacitated versions, this restriction is relaxed by
  allowing multiple copies to be opened.}, which we refer to as
Capacitated $k$-Median. In this problem we are given a set of
potential facilities $F$, capacity $u_i \in \mathbf{N}^+$ for each
facility $i \in F$, a set of clients $C$, a metric distance function
$d\colon C \times F\rightarrow\mathbf{R}_{\geq 0}$ on $C \times F$,
and an integer $k$. The goal is to find a subset $F' \subseteq F$ of
$k$ facilities to open and an assignment
$\sigma \colon C \rightarrow F'$ of clients to the open facilities
such that $|\sigma^{-1}(i)| \leq u_i$ for every $i \in F'$, so as to
minimize the connection cost $\sum_{j \in C} d(j,
\sigma(j))$. Obtaining a constant-factor approximation algorithm for
this problem is one of the central open questions in the area of
approximation algorithms. So far only algorithms that either violate
the cardinality or the capacity constraints are known~\cite{dl16, bru16, Li16}. We can use our techniques to get a non-trivial
approximation algorithm under a special metric case of Capacitated
$k$-Median. More precisely, we obtain an approximation ratio of
$2.295$ (improving the trivial ratio of~$3$) for the special case where
the underlying metric space has only two possible distances between
clients and facilities (say $a, b \in \mathbf{R}_{\geq 0}$)
\emph{without violating any constraint}. In cases where $b>3a$ the
problem decomposes into separate clusters and can be easily solved
efficiently by dynamic programming. The most interest case is when
$b=3a$ and can thus be thought of having only distances one and three
between clients and facilities. Interestingly enough, this seemingly
special case of two distances provides the best known
inapproximability bound of $\approx 1.736$~\cite{jms02} for the
\emph{general} problem and also for several related facility location
problems. The only other result on Capacitated $k$-Median problem
under special metrics, that we are aware of, concerns tree metrics
where the problem can be solved exactly using a dynamic program. It
would be interesting to see if also more general metrics can be
tackled with our approach or if the lower bound is actually tight for
two-distance metrics.

To achieve the above, we reduce these two-distance instances to the
monotone submodular function maximization problem subject to one
covering and one packing constraint. We observe that the reduced
instances are instances of \PCSM subject to one cardinality constraint and one polynomially bounded covering
constraint. Hence we can apply Lemma~\ref{lem:submodular-card-pseudopoly} to these instances to prove the following theorem.

\begin{theorem}\label{thm:cap-k-med}
  There is a $2.294$-approximation algorithm running in $O(n^4)$ time
  for the $k$-median problem with non-uniform and hard capacities if
  the underlying metric space has only two possible distances between
  clients and facilities.
\end{theorem}

\begin{proof}
  Consider an instance of Capacitated $k$-Median and let $a\leq b$ be
  the two distances between clients and facilities in our metric
  space.

  Let us first consider the case in which $a>0$. Then by scaling, we
  can assume that $a=1$. For any subset $F'\subseteq F$ of facilities
  let $f(F')$ be the maximum number of clients in $C$ that can be
  connected at a distance~$1$ to the facilities in $F'$. More
  precisely, let $C'\subset C$ be a largest set of clients such that
  there is an assignment $\tau\colon C'\rightarrow F'$ with
  $d(j,\tau(j))=1$ for all $j\in C'$ and $|\tau^{-1}(i)|\leq u_i$ for
  all $i\in F'$. Then we set $f(F')=|C'|$. It is not hard to see that
  for a given set $F'$ the function $f(F')$ can be computed in
  polynomial time using $b$-matching algorithms and that the function
  is monotone and submodular (equivalent to the setting of the
  capacitated set cover problem, where it is known~\cite{cn02} that
  this function is monotone and submodular).

  Our objective is to select a subset $F' \subseteq F$ such that
  $f(F')$ is maximized subject to the constraints that $|F'|\leq k$
  (cardinality constraint) and $\sum_{i\in F'}u_i\geq n$ (covering
  constraint). By Lemma~\ref{lem:submodular-card-pseudopoly}, we
  can find a $0.353$-approximate solution for this problem.
 
  Let $F'$ be the set of facilities output by the algorithm, let $C'$
  be the set of clients connected at a distance~$1$, and $\tau$ be the
  corresponding assignment. Because of our covering
  constraint~$\sum_{i\in F'}u_i\geq n$, it is straightforward to
  extend the assignment $\tau\colon C'\rightarrow F'$ to an assignment
  $\sigma\colon C\rightarrow F'$ for all clients by using the
  sufficiently large residual capacity of at least $n-|C'|=n-f(F')$
  for connecting the clients in $C \setminus C'$ to facilities in $F'$
  at distance~$b$. Hence $f(F')$ clients are assigned at a distance
  $1$ and $n-f(F')$ clients are assigned at a distance $b$. Note that
  in a similar way we can establish a reverse correspondence and find
  a solution for the above submodular maximization problem using a
  solution for the $k$-median problem. In particular, note that if
  $F' \subseteq F$ is an optimal solution for the submodular
  maximization problem, then $F'$ is also an optimum for the
  $k$-median problem. Now our algorithm for the $k$-median problem is
  to run our $0.353$-approximation algorithm of
  Lemma~\ref{lem:submodular-card-pseudopoly} on the submodular
  maximization instance and output the corresponding set of facilities
  $F'$ and the assignment $\sigma$. Note that this algorithm runs in
  polynomial time $O(n^4)$ for the instance constructed since
  $P_{\max}\leq n$.

  Let $\opt$ be the value of the optimal solution for the submodular
  maximization instance and $A=f(F')$ be the value of the approximate
  solution output by our algorithm. Hence, the optimal value $\opt'$
  for the $k$-median instance is
  $\opt + b (n - \opt) = b n - (b-1) \opt$ and the approximate
  solution value $A'$ is $b n - (b-1) A$. We know that,
  $A \geq 0.353\opt$. Hence,
\begin{align*}
A' &= b n - (b-1) A \\
    & \leq b n - 0.353(b-1)\opt \\
    & = b (n - \opt) + (0.647b+0.353) \opt \\
    & \leq (0.647b+0.353)(b (n - \opt) + \opt) \\
    &= (0.647b+0.353) \opt'
\end{align*}

Now for the case $b \leq 3$, the approximation factor becomes
$0.647b+0.353 \leq 2.294$. In the case $b > 3$ the instance decomposes
into a collection of complete bipartite graphs between client and
facilities with distance $1$ and thus this case can be easily solved
to optimality by a DP similarly to the case $a=0$ described below.

Finally, in the case $a=0$, there is no upper bound on $b$. However,
it is not hard to solve this case optimally in polynomial
time. Observe that the metric space is clustered into subsets of
points with zero distance among each other. By greedily assigning
clients to facilities with highest capacity one can find the maximum
number of clients that can be served within a given cluster using a
given number $i$ of facilities. Using a DP approach (similar to the
knapsack DP) one can compute the optimum distribution of facilities
among the clusters.
\end{proof}



\subsection{Allowing Duplicates for reducing violation in covering constraints}\label{subsec:duplicates-greedy-dp}

In this section, we show that the vanilla Greedy-DP can be extended to get an $1/\e$- approximation algorithm for \PCSM under the relaxed set constraint by allowing to add an additional copy for some elements to our solution. For this case, the violation factors for packing and covering constraints are $1\pm \eps$.

 To this end, the greedy DP algorithm~\ref{alg:vanilla-greedydp} will be accompanied by a \emph{completion phase}. In this phase, we try to \emph{complete} the solution for each DP cell to a feasible solution by adding a suitable subset. More precisely, for every entry $T[q, \bc_q, \p_q]$, we can find a subset of $\cN$ with covering value $\bc - \bc_q$ and packing value $\p-\p_q$ (if there exists such a set) by running a simple feasibility DP similar to the multi-dimensional knapsack DP~\cite{kpp04}. For any entry, if we can find such a completion set, we mark it as \emph{valid solution} otherwise we mark it as \emph{invalid}. Clearly, each element is added at most twice to any of the \emph{completed valid solutions} of the DP
table. Finally, we output the best of the valid solutions. This gives us the following pseudo-polynomial time algorithm without any violation.

 \begin{lemma}
 \label{apx-lem:submodular-const-0.25}
 There is an algorithm for \PCSM that outputs in pseudo-polynomial time \\ $n^{O(1)}c_{\max}p_{\max}$ a $1/\e$-approximate solution with covering value at least $\bc$ and packing value $\p$ if we allow to add up to two copies of any element in $\cN$ to the solution.
 \end{lemma}
 \begin{proof}
Using Lemma~\ref{apx-lem:submodular-const-multi-phase},~\ref{apx-lem:submodular-const-dual} and~\ref{apx-lem:submodular-const-lp}, we get all the properties except the covering constraints feasibility. To argue this, let $O_m$ be as in Lemma~\ref{apx-lem:submodular-const-multi-phase}. Then the corresponding subset $O\setminus O_m \subseteq \cN$ is a witness to complete the solution $T[L_m, \bc_m, \p_m]$ to a valid feasible solution. Hence the feasibility DP in the completion phase will augment this entry to a feasible solution as well.
 \end{proof}
 
 By using standard scaling techniques this leads to the following polynomial time algorithm.
 
 \begin{corollary}
  There is an algorithm for \PCSM that outputs for any $\eps>0$ in $(n/\eps)^{O(1)}$ time a $1/\e$-approximate solution with covering value at least $(1-\eps)\bc$ and with packing value at most $(1+\eps)\p$, if we allow to add up to two copies of any element in $\cN$ to the solution.
 \end{corollary}


\section{Extensions: Matroid Independence and Multi-Objective}\label{App:Extensions}
\subsection{Matroid Independence}
We consider the \MPCSM problem as defined in Section \ref{sec:Introduction}. The main difficulty in proving Theorem \ref{thrm:MatroidExtend} is the following. $\bx^*$ is found by the continuous approach and is a convex combination of independent sets in the given matroid. However, it is {\em not} guaranteed that it is a fractional base.
Unfortunately, known rounding techniques such as randomized swap rounding \cite{chekuriVZ10-rand-exch}, require a fractional base. To this end, we present a simple extension of the swap rounding algorithm of \cite{chekuriVZ10-rand-exch} that works on independent sets as opposed to bases (see \cite{Jan18}).
It is crucial that this extension has all the concentration properties of the original swap rounding.
Algorithm \ref{alg:ExtendedSwapRounding} describes this extension of the swap rounding.
Intuitively, Algorithm \ref{alg:ExtendedSwapRounding} pads the ground set with dummy elements to obtain a fractional base.

\subsubsection{Extended Swap Rounding}

\begin{algorithm} 	
	\caption{Extended Swap Rounding: $(\cM(\cN,\cI)), x^* \in \cP(\cM)$}
	\label{alg:ExtendedSwapRounding}
	$r \leftarrow rank(\cM)$. \\
	Denote by $D$ a set of size $r$ of elements with value $0$. \\
	Define $\cN' = \cN \cup D$ and a matroid $\cM'$ whose bases are $\{S: S\setminus D \in \cI$ and $|S|=r\}$ \\
	Write $\bx^*$ as $\bx^* = \sum_i \alpha_i \textbf{1}_{S_i}$ where $\alpha_i \geq 0$, $\forall i$, $\sum_i \alpha_i = 1$, and $S_i \in \cI$, $\forall i$ (see Section 3 of \cite{chekuriVZ10-rand-exch}). \\
	For each $S_i$ define $S_i'$ s.t. $|S_i'| = r$ and $S_i' = S_i\cup Y$ (for some $Y \subseteq D$).\\
	Define $\tildx$ as: $\tildx = \sum_i \alpha_i \textbf{1}_{S'}$.\\
	Let $S'$ be the result of the swap rounding algorithm applied to $\tildx$ and $\cM'$.\\
	Output $S' \setminus D$.
\end{algorithm}
Given a matroid $\cM = (\cN,\cI)$, we define an extension matroid over $\cN' = \cN \cup D$ where $D$ is a set of dummy elements with a value of $0$. The bases are all sets $S'$ such that $S'\setminus D \in \cI$ and $|S'| = rank(\cM)$. Thus, the size of $D$ is $rank(\cM)$.
Note that $rank(\cM) = rank(\cM')$.
We define an auxiliary submodular objective function $f': 2^{\cN'} \rightarrow \R_+$ as $f'(S') \triangleq f(S'\setminus D)$. Clearly, $f'$ is a monotone submodular function.

We start by proving that the extension matroid $\cM'$ is, in fact, a matroid.

\begin{lemma} \label{lem:ExtensionMatroid}
Given a matroid $\cM = (\cN,\cI)$, the extension $\cM' = (\cN',\cI')$ is also a matroid.
\end{lemma}
\begin{proof}
Consider two sets $A',B' \subseteq \cN'$. Denote by $A = A' \cap \cN$, $B = B' \cap \cN$.
We show that for every couple of sets $A'$,$B'$ and $a \in A'\setminus B'$ there exists $b \in B'\setminus A'$ such that $(A'\setminus \{a\}) \cup \{b\} \in \cI'$. We consider 3 cases:
\begin{itemize}
\item $|A| < |B|$
By the Independent set exchange property, there exists $b \in B \setminus A$ such that $A \cup \{b\} \in \cI$.
\item $|A| = |B|$
If $|A| = |B| = rank(\cM)$ then the basis exchange property holds.
If $|A| = |B| < rank(\cM)$ then we consider two cases. for the case where $A=B$, we can swap two dummy elements from $A',B'$. For the case where $A \neq B$, if $a \in A$ then $|B| > |A\setminus\{a\}|$ and the Independent set exchange property holds. if $a \notin A$ then $B'$ contains a dummy element $b$ that is not in $A'$ which can be swapped.
\item $|A| > |B|$
$B'$ contains a dummy element $b$ that is not in $A'$
\end{itemize}
Since the definition holds for all cases, this concludes the proof.
\end{proof}

The following theorem states that Algorithm \ref{alg:ExtendedSwapRounding} outputs an integral solution that is independent in $\cM$.

\begin{theorem} \label{thrm:ExtendedSwapRoundingIndependence}
	The output of Algorithm \ref{alg:ExtendedSwapRounding} is an independent set of $\cM$.
\end{theorem}

\begin{proof}
	From Lemma \ref{lem:ExtensionMatroid}, we conclude that $\cM'$ is a matroid. Note that Steps 4,5 and 6 of the algorithm converts $\bx^*$ to a convex combination of bases in $\cM'$. Thus, applying the swap rounding algorithm guarantees that $S' \in \cI'$. Additionally, by the definition of $\cM'$, it is guaranteed that $S'\setminus D \in \cI$.
\end{proof}

The following proposition shows the concentration properties of swap romdunding as proved in \cite{chekuriVZ10-rand-exch}.
\begin{prop} \label{prop:SwapRoundingProperties}
	Let $\bx^*$ be a fractional base in the base polytope of $\cM$, and $S$ be the output of the swap rounding when applied to $\bx^*$. Denote by $\bV \in \R^{\cN}_+$ a vector of non-negative weights. Then,
	\begin{enumerate}
		\item $\Pr[\ell \in S] = \bx^*_\ell$.
		\item For every $\mu \leq \bx^* \cdot \bV$ and $\delta < 1$: $\Pr[\textbf{1}_S \cdot \bV < (1-\delta) \bx^* \cdot \bV] \leq e^{-\mu\delta^2/2}$.
		\item For every $\mu \geq \bx^* \cdot \bV$: $\Pr[\textbf{1}_S \cdot \bV > (1+\delta) \bx^* \cdot \bV] \leq \left(\frac{e^\delta}{(1+\delta)^{1+\delta}}\right)^\mu$.
		\item For a monotone submodular function $f$ and its multilinear extension $F$ it is true that:\\ $\Pr[f(S) < (1-\delta)F(\bx^*)] \leq e^{-\mu\delta^2/8}$, where $\mu = F(\bx^*)$.
	\end{enumerate}
\end{prop}

We would like to prove that the extended swap rounding (Algorithm \ref{alg:ExtendedSwapRounding}) has the same properties as in the above proposition.
The following lemmas bound the upper and lower tail of a linear function using Algorithm \ref{alg:ExtendedSwapRounding} and the lower bound of a submodular function.

\begin{lemma}
Given a non-negative weight vector $W \in (\R_+)^\cN$, $\bx^* \in \cP(\cM)$ a fractional solution, denote by $S \in \cN$ the output of Algorithm \ref{alg:ExtendedSwapRounding} on $\bx^*$. Then,
$\Pr\left[W \cdot \mathbf{1}_{\bx^*} < (1-\delta)\mu\right] \leq e^{\mu \delta^2/2}$.
\end{lemma}
\begin{proof}
First, note that from the properties of swap rounding (as shown in Proposition \ref{prop:SwapRoundingProperties}), $\E[W \cdot \textbf{1}_{S'}] = W \cdot \textbf{1}_{\bx^*}$. Also,
$W \cdot \textbf{1}_{S'}
= W \cdot \textbf{1}_{S'\setminus S} + W \cdot \textbf{1}_S
= W \cdot \textbf{1}_S$.
We conclude that:
\begin{align*}
	\Pr\left[W \cdot \textbf{1}_S < (1-\delta)\mu\right]
	&= 	\Pr\left[W \cdot \textbf{1}_{S'} < (1-\delta)\mu\right] \\
	&\leq e^{-\E[W \cdot \textbf{1}_{S'}]\delta^2/2}
	= e^{-\mu \delta^2/2}.
\end{align*}
\end{proof}
\begin{lemma}
Given a non-negative weight vector $W \in (\R_+)^\cN$, $\bx^* \in \cP(\cM)$ a fractional solution, denote by $S \in \cN$ the output of Algorithm \ref{alg:ExtendedSwapRounding} on $\bx^*$. Then,
$\Pr\left[W \cdot \mathbf{1}_S > (1+\delta)\mu\right] \leq \left(\frac{e^\delta}{(1+\delta)^{1+\delta}}\right)^\mu$.
\end{lemma}
\begin{proof}
First, note that from the properties of swap rounding (as shown in Proposition \ref{prop:SwapRoundingProperties}), $\E[W \cdot \textbf{1}_{S'}] = W \cdot \textbf{1}_{\bx^*}$. Also,
$W \cdot \textbf{1}_{S'}
= W \cdot \textbf{1}_{S'\setminus S} + W \cdot \textbf{1}_{S}
= W \cdot \textbf{1}_{S}$.
We conclude that:
\begin{align*}
\Pr\left[W \cdot \textbf{1}_S > (1+\delta)\mu\right]
&= 	\Pr\left[W \cdot \textbf{1}_{S'} > (1+\delta)\mu\right]
\leq \left(\frac{e^\delta}{(1+\delta)^{1+\delta}}\right)^\mu
\end{align*}
\end{proof}
The next lemma proves that given a monotone submodular function, Algorithm \ref{alg:ExtendedSwapRounding} guarantees high concentration properties, as the original swap rounding.
\begin{lemma} \label{lem:ExtendedSwapRoundingObjective}
Given a monotone submodular function $f: 2^\cN \rightarrow R_+$, its multi-linear extension $F$ and a fractional solution $\bx^* \in \cP(\cM)$,
denote by $S$ the output of Algorithm \ref{alg:ExtendedSwapRounding} and $\mu = F(\bx^*)$. Then: $\Pr\left[f(S) \leq (1-\delta)\mu\right] \leq e^{-\mu \delta^2/8}$.
\end{lemma}
\begin{proof}
Recall that $f': 2^{\cN'} \rightarrow \R_+$ is defined as follows: $f'(S) \triangleq f(S \setminus D)$.
Apply Proposition \ref{prop:SwapRoundingProperties}, then we conclude that:
\begin{align*}
	\Pr\left[f(S) \leq (1-\delta) \mu\right]
	= \Pr\left[f'(S') \leq (1-\delta) \mu\right]
	\leq e^{-\mu \delta^2/8}.
\end{align*}
\end{proof}


\subsubsection{Proof of Theorem}

Let us now describe our algorithm for \MPCSM.
The algorithm is identical to our main result (Algorithm \ref{alg:OneSidedMainAlgorithm}), with the following changes:
(1) Given a guess $D = (E_0,E_1,c')$ we contract the matroid $\cM$ by $E_1$ (and then remove $E_0$ from the ground set as before, thus $\cM'$ is defined with respect to the contracted $\cM$);
(2) Instead of independently rounding each element, we apply the extended swap rounding to the contracted matroid. The rest of this section, along with all the lemmas in it, consists of the proof of Theorem \ref{thrm:MatroidExtend}.

\begin{algorithm*}[H]
	\caption{$(f,\cN,\bP,\bC)$}\label{alg:MatroidExtention}
	Use Algorithm \ref{alg:OneSidedGuessesEnumeration} to obtain a list of guesses $\cL$. \\
	\ForEach{$D=(E_0, E_1,\bc') \in \cL$}{
		Use Theorem~\ref{thrm:ContGreedy} to compute an approximate solution $\bx^*$ to problem (\ref{thrm:MatroidExtend}).\label{alg:OneSidedMainAlgorithm-ContGreedy}\\
		Scale down $\x^*$ to $\bar{\x} = \x^*/(1+\delta)$\\
		Let $R_D$ be the output of Algorithm \ref{alg:ExtendedSwapRounding} on $\bar{\x}$\\
		Let $R'_D = R_D \setminus L_D$.\\
		$S_D \leftarrow E_1 \cup R'_D$.
	}
	$S_{alg} \leftarrow \argmax \left\{f(S_D) : D\in \cL, \bP \cdot \one_{S_D} \leq \one_p, \bC \cdot \one_{S_D} \geq (1-\eps)\one _c \right\}$
\end{algorithm*}

Note that since the matroid is down-monotone, scaling down the fractional solution and removing elements from the solution will not violate the matroid constraint.

The following lemma bounds the probability of violating a packing constraint by more than $\varepsilon$, when using the extended swap rounding. As before, denote by $X_\ell$ the indicator for the event that element $\ell$ is chosen by the algorithm.

\begin{lemma} \label{lem:MatroidExt-packing}
For every $i = 1,...,p$: $\Pr\left[\sum_{\ell\in\cN\setminus(E_0 \cup E_1)}\bP_{i,\ell}X_\ell > \rdi \right] < \max\{e^{-\frac{\delta^2}{3\alpha}}, e^{-\frac{\delta^2}{3\beta}}\}$.
\end{lemma}
\begin{proof}

For any non-tiny packing constraint $i \in [p] \setminus Y_D$ and for each $\ell \in \residual$, let us define the scaled matrix $\tP$ such that $\tP_{i, \ell} = \bP_{i, \ell}/(\alpha \rdi) \leq 1$. The last inequality follows from Defn.~\ref{apx-def:OneSidedCorrectGuess}.\ref{apx-item:oneSidedConsistentGuessNoBig}. Notice that $\E[\sum_{\ell \in R_D} \tP_{i, \ell}] \leq \nicefrac{\rdi}{(1+\delta) \alpha \rdi} = \nicefrac{1}{(1+\delta)\alpha}$ by Claim~\ref{apx-clm:exp-indep-round}.\ref{apx-item:exp-pack}. Now, applying the concentration property of the extended swap rounding (\ref{prop:SwapRoundingProperties}.3) with $X=\sum_{\ell \in \residual} \tP_{i, \ell} X_\ell$, we obtain
\begin{align*}
	\Pr\left[ \sum_{\ell \in \sol} \bP_{i, \ell} > 1\right] & = \Pr\left[\sum_{\ell \in R_D} \bP_{i, \ell} > \rdi \right] \\ 
	&= \Pr\left[\sum_{\ell \in R_D} \tP_{i, \ell} > \nicefrac{1}{\alpha}\right]\\
	& = \Pr\left[\sum_{\ell \in R_D} \tP_{i, \ell} > (1+\delta)\frac{1}{(1+\delta)\alpha}\right]\\
	& \leq \exp\left(-\frac{1}{(1+\delta)\alpha} \cdot (\delta^2)/2\right)\\
	& \leq \exp\left(-\frac{1}{(1+\delta)\alpha} \cdot (\delta^2)/3\right)\\
	& \leq \exp\left(-\frac{\delta^2}{3\alpha}\right)\,.
\end{align*}

As for tiny packing constraints $i\in Y_D$ and for each $\ell \in S_D^i$, we again define the scaled matrix $\tP$ such that $\tP_{i, \ell} = \bP_{i, \ell}/(\beta \rdi) \leq 1$. The last inequality follows from Defn.~\ref{apx-def:OneSidedCorrectGuess}.\ref{apx-item:oneSidedConsistentGuessNoBig}. Notice that $\E[\sum_{\ell \in R_D \cap S_D^i} \tP_{i, \ell}] \leq \E[\sum_{\ell \in R_D} \tP_{i, \ell}] \leq \nicefrac{\rdi}{(1+\delta) \beta \rdi} = \nicefrac{1}{(1+\delta)\beta}$ by Claim~\ref{apx-clm:exp-indep-round}.\ref{apx-item:exp-pack}. Applying the concentration property of the extended swap rounding (\ref{prop:SwapRoundingProperties}.3) with $X=\sum_{\ell \in S_D^i} \tP_{i, \ell} X_\ell$, we obtain
	\begin{align*}
		\Pr\left[ \sum_{\ell \in E_0 \cup (R_D \cap S_D^i)} \bP_{i, \ell} > 1\right] & = \Pr\left[\sum_{\ell \in R_D \cap S_D^i} \bP_{i, \ell} > \rdi \right] \\ 
		&= \Pr\left[\sum_{\ell \in R_D \cap S_D^i} \tP_{i, \ell} > \nicefrac{1}{\beta}\right]\\
		& = \Pr\left[\sum_{\ell \in R_D \cap R_D} \tP_{i, \ell} > (1+\delta)\frac{1}{(1+\delta)\beta}\right]\\
		& \leq \exp\left(-\frac{1}{(1+\delta)\beta} \cdot (\delta^2)/2\right)\\
		& \leq \exp\left(-\frac{1}{(1+\delta)\beta} \cdot (\delta^2)/3\right)\\
		& \leq \exp\left(-\frac{\delta^2}{3\beta}\right)\,.
	\end{align*}
			
\end{proof}


Now, consider the covering constraint. The following lemma shows that for every non-tiny constraint, the output will not violate the constraint with high probability.

\begin{lemma} \label{lem:MatroixExt-covering}
For every non-tiny covering constraint $j: Pr\left[\sum_{\ell \in \cN \setminus (E_0 \cup E_1)}\bC_{j,\ell} X_\ell < \sdj - \varepsilon\right] < e^{-\frac{\delta^2}{2\alpha}}$.
\end{lemma}
\begin{proof}

For each non-tiny covering constraint $j \in [c] \setminus Z_D$, and each $\ell \in \residual$, we define the scaled matrix $\tC$ such that $\tC_{j, \ell} = \bC_{j, \ell}/(\alpha \sdj) \leq 1$. $\E[\sum_{\ell \in R_D} \tC_{j, \ell}] \geq \nicefrac{\sdj}{(1+\delta) \alpha \sdj} = \nicefrac{1}{(1+\delta)\alpha}$ by Claim~\ref{apx-clm:exp-indep-round}.\ref{apx-item:exp-cover}. Again, applying the concentration property of the extended swap rounding, ((\ref{prop:SwapRoundingProperties}.2)) with $X=\sum_{\ell \in \residual} \tC_{j, \ell} X_\ell$, we obtain

\begin{align*}
	\Pr\left[ \sum_{\ell \in \sol} \bC_{j, \ell} < (1 - 2 \delta) \bc'_j \right] & = \Pr\left[\sum_{\ell \in R_D} \bC_{j, \ell} <(1 - 2 \delta) \sdj \right] \\ 
	&= \Pr\left[\sum_{\ell \in R_D} \tC_{j, \ell} < \nicefrac{(1 - 2 \delta)}{\alpha}\right]\\
	& = \Pr\left[\sum_{\ell \in R_D} \tC_{j, \ell} < (1 - \delta)\frac{1}{(1+\delta)\alpha}\right]\\
	& \leq \exp\left(-\frac{1}{(1+\delta)\alpha} \cdot (\delta^2)/2\right)\\
	& \leq \exp\left(-\frac{\delta^2}{3\alpha}\right)\,.
\end{align*}
\end{proof}

We apply Lemma \ref{lem:ExtendedSwapRoundingObjective}, similarly to the main result for \PCSM, to bound the probability that the value of the output is low.
Along with Lemmas \ref{lem:MatroidExt-packing}, \ref{lem:MatroixExt-covering} this concludes our extension \MPCSM. 
\subsection{Multi-Objective}
We now consider the \PCMSM problem. We are given $p$ packing constraints, $c$ covering constraints, and $t$ monotone submodular functions $f_1,...,f_t$:$n^\cN \rightarrow \R_+$. The goal is to maximize all $f_1(S),...,f_t(S)$ simultaneously given the mixed packing and covering constraints.
This problem is not well defined, hence we focus on the problem of finding a \textit{pareto set}: a collection of solutions such that for each $S$ in the collection there is no other solution $S'$ that $f_i(S') > f_i(S)$ and $f_j(S') \geq f_j(S)$ for some $i$ in $1,..,t$ and every $j \neq i$.

Since finding a pareto set is computationally hard, we settle for less demanding goals, finding an $\alpha$-\textit{approximate} pareto set: a collection of solutions $\{S: \bP \textbf{1}_S \leq \textbf{1}_p , \bC \textbf{1}_S \geq (1-\varepsilon)\textbf{1}_c\}$ such that for each $S$ in the collection there is no solution $S'$ satisfying $\bP \textbf{1}_{S'} \leq \textbf{1}_p$ and $\bC \textbf{1}_{S'} \geq \textbf{1}_c$ for which the following is true: $f_i(S') \geq (1+\alpha)f_i(S)$ for all $i$. Formally,
Papadimitriou and Yannakakis \cite{PY00} proved that there is a polynomial time algorithm for finding an $\alpha$ approximate pareto set if and only if there is an algorithm as described in Theorem \ref{thrm:MultiExtend} for the \PCMSM problem (see Section \ref{sec:Introduction}).
Thus, we focus on proving Theorem \ref{thrm:MultiExtend} (Note that in our case $\alpha = \frac{\nicefrac{1}{e} + \varepsilon}{1 - \nicefrac{1}{e} - \varepsilon}$).


\begin{proof}[Proof (Sketch)]
The preprocessing is almost identical to Algorithm \ref{alg:OneSidedGuessesEnumeration}. The difference is that the $\gamma$ most valuable elements are guessed for each $f_i$. For every such guess we apply the following:
According to Lemma 7.3 of \cite{chekuriVZ10-rand-exch} we can find a fractional solution $\bx^*$ such that for all $i$: $F_i(\bx^*) \geq (1-\nicefrac{1}{e})v_i$ or receive a certificate that there is no feasible solution for which the condition holds.
Now, we can use independent rounding to round the solution, and for a small enough $\varepsilon$, receive with positive probability a solution $S$ such that $f_i(S) \geq (1-\nicefrac{1}{e}-\varepsilon)v_i$ for all $i$ that violates the covering constraints by $\varepsilon$.
\end{proof}

\section{Running Time Lower Bounds for achieving One-Sided Feasibility}
\label{apx-sec:lower-bound-reduction}
In this section, we argue that the exponential dependence on
$\nicefrac{1}{\eps}$ in the running time is necessary when we insist on
one-sided feasibility. In particular, we show
\begin{lemma}
  Assume we are given an instance of \PCSM with three packing
  constraints and one (uniform) covering constraint for which there is
  a feasible solution to this instance. If there is an algorithm that
  computes every $\eps>0$ in time $n^{o(\nicefrac{1}{\eps})}$ a
  solution $S\subseteq\cN$ with $\bP\one_S\leq\one_p$ and
  $\bC\one_S\geq (1-\eps)\one_c$ then the exponential time hypothesis
  (ETH) fails.
\end{lemma}
\begin{proof}
  We provide a reduction from the $k$-sum problem~\cite{PW10}. In this
  problem we are given a ground set~$\cN$, a vector
  $\mathbf{A}\in\mathbb{Z}^{1\times n}$, a target value $t\in\N$, and
  cardinality $k\in\N$. We aim at deciding if there is a set
  $S\subseteq\cN$ such that $\mathbf{A}\one_S=t$. Let
  ${a_{\max}=\|\mathbf{A}\|_{\infty}}$ We can assume that
  $\mathbf{A}\in\N^{1\times n}$ by considering the modified instance
  $\mathbf{A}':=\mathbf{A}+a_{\max}\cdot\one_n^{\intercal}$
  and $t'=t+k\cdot a_{\max}$. It was shown by Pătrașcu
  and Williams~(Corrollary 5.1 in~\cite{PW10}) that if there is an
  algorithm that solves each $k$-sum instance whose number have
  $O(k\log n)$ bits in $n^{o(k)}$ time then ETH fails.

  We construct an instance of \PCSM as follows. The ground set $\cN$. To
  define the constraints let $\bP\in\N^{3\times n}$ such that for all
  $\ell\in\cN$ we have $\bP_{1,\ell}=\mathbf{A}_{1,\ell}$,
  $\bP_{2,\ell}=a_{\max}-\mathbf{A}_{1,\ell}$, and
  $\bP_{3,\ell}=1$. Moreover, $\bC\in\N^{1\times n}$ such that
  $\bC_{1,\ell}=1$. Finally, we set
  $\p=(t,k\cdot a_{\max}-t,k)^{\intercal}$ and $\bc=k$.

  Assume there is a feasible solution to the $k$-sum instance. Assume
  further there is an algorithm $\mathcal{A}$ outputting for any
  $\eps>0$ in $n^{o(\nicefrac{1}{\eps})}$ time a one-sided feasible
  solution.

  Now we run this algorithm on our \PCSM instance setting
  $\eps=\nicefrac{1}{(k+1)}$.  It is easy to verify that $S$ is a
  feasible solution to the \PCSM instance. Moreover, algorithm
  $\mathcal{A}$ will output a solution $S'\subseteq\cN$ with
  $\bP\one_{S'}\leq\p$ implying $|S'|\leq k$. Additionally,
  $\bC\one_{S'}\geq(1-\eps)\bc=(1-\nicefrac{1}{(k+1)})k>k-1$ implying
  $|S'|=k$. But then $S'$ must be a strictly \emph{feasible} solution
  to the $k$-sum instance because of $\bP\one_{S'}\leq\p$ and thus
  $\mathbf{A}\one_{S'}=t$. The total running time of the algorithm is
  $n^{o(1/(1/(k+1))}=n^{o(k)}$. Thus ETH fails.
\end{proof}


\end{document}